\documentclass[11pt]{article}%
\pdfoutput=1
\usepackage[utf8]{inputenc}%
\usepackage[T1]{fontenc}%
\usepackage{amsmath, amsthm, a4, latexsym, amssymb, mathrsfs,graphicx}%
\usepackage{natbib}%
\usepackage{bm}%
\usepackage{lineno}%
\usepackage[all]{xy}%
\usepackage{color}%
\usepackage{rotating}%
\usepackage{a4wide,fullpage}%
\usepackage{setspace}%
\DeclareMathAlphabet{\mathpzc}{OT1}{pzc}{m}{it}%
 \newcommand{\eps}{\varepsilon}%
\newcommand{\N}{\mathbb{N}}%
 \newcommand{\prob}{\mathbb{P}}%
 \newcommand{\E}{\mathbb{E}}%
 \renewcommand{\P}{\mathbb{P}}%
 \newcommand{\osf}{{\sf 1}}%
 \renewcommand{\N}{{\mathbb N}}%

 \newcommand{\loc}{\mathbb L}%
\newcommand{\ind}{\mathbb I}%
\newcommand{\eN}{\varepsilon_{_N}}%
\newcommand{\rN}{r_{_N}}%
\newcommand{\sub}[2]{\ensuremath{{#1}_{_{#2}}}}%
 \newcommand{\EE}[1]{\ensuremath{\E\left[ #1 \right] }}%
\newcommand{\crr}[2]{\ensuremath{\textrm{cor}_i\left( #1, #2 \right)}}%
\newcommand{\corr}[2]{\ensuremath{\textrm{cor}\left( #1, #2 \right)}}%
 \newcommand{\PP}[1]{\ensuremath{\P\left( #1 \right) }}%
\newcommand{\jp}{\ensuremath{\hat{\jmath}}}
\newcommand{\jt}{\ensuremath{\tilde{\jmath}}}%
\newcommand{\achr}{\boldsymbol{\vdash}}%
\newcommand{\bchr}{\boldsymbol{\dashv}}%
\newcommand{\chr}{\boldsymbol{\vdash}\!\!\boldsymbol{\dashv}}%
\newcommand{\cacb}{\boldsymbol{\bullet}\!\!\!\boldsymbol{-}\!\!\!\boldsymbol{\bullet}}%
\newcommand{\cachr}{\bullet\!\!\boldsymbol{-}}
\newcommand{\cabchr}{\boldsymbol{\bullet}\!\!\!\boldsymbol{\dashv}}
\newcommand{\ou}[2]{\genfrac{}{}{0pt}{3}{#1}{#2}}%
\newcommand{\acb}{\boldsymbol{-}\!\!\bullet}%

\renewcommand{\subsection}{\secdef \subsct\sbsect}
\newcommand{\subsct}[2][default]{\refstepcounter{subsection}
\vspace{0.15cm}
{\flushleft\bf \arabic{section}.\arabic{subsection}~\bf #1  }
\nopagebreak\nopagebreak}
\newcommand{\sbsect}[1]{\vspace{0.1cm}\noindent
{\bf #1}\vspace{0.1cm}}

{\nopagebreak {\hfill\rule{2mm}{2mm}}\\ }

\newtheorem{theorem}{Theorem}[figure]
\newtheorem{lemma}[theorem]{Lemma}

\newtheorem{prop}[theorem]{Proposition}

\newtheorem{definition}[theorem]{Definition}
\newtheorem{remark}[theorem]{Remark}

\newcommand{\TMI}{\ensuremath{\ind''}}
\newcommand{\OMI}{\ensuremath{\ind'}}%
\newtheoremstyle{thm}{1.5ex}{1.5ex}{\itshape\rmfamily}{}
{\bfseries\rmfamily}{}{2ex}{}

\newtheoremstyle{rem}{1.3ex}{1.3ex}{\rmfamily}{}
{\itshape\rmfamily}{}{1.5ex}{}
\theoremstyle{rem}
\refstepcounter{subsubsection}

\def\thebibliography#1{\section*{References}
  \list%
  {\arabic{enumi}.}
    {\settowidth\labelwidth{[#1]}\leftmargin\labelwidth
    \advance\leftmargin\labelsep
    \parsep0pt\itemsep0pt
    \usecounter{enumi}}
    \def\newblock{\hskip .11em plus .33em minus .07em}
    \sloppy                   
    \sfcode`\.=1000\relax}

\begin{document}%
\begin{center}%
\LARGE
An ancestral recombination graph for diploid populations with skewed offspring distribution %
\end{center}
\begin{center}%
  \Large
  Matthias Birkner$^1$, Jochen Blath$^2$, Bjarki Eldon$^{2,*}$ 
  \end{center}
\noindent   $^1$Institut f\"ur Mathematik, Johannes-Gutenberg-Universit\"at Mainz, 55099 Mainz, Germany\\
\noindent    $^2$Institut f\"ur Mathematik, Technische Universit\"at Berlin, 10623 Berlin, Germany \\
  \date{\today}%

\clearpage\pagebreak\newpage

Running title:\\
  An ancestral recombination graph admitting simultaneous multiple mergers \\


  Corresponding author:\\ 
  Bjarki Eldon\\
  Institute f\"ur Mathematik,  Technische Universit\"at Berlin, Strasse des 17.\ Juni 136, 10623 Berlin, Germany \\
  office: +49 303 142 5762 \\
  eldon@math.tu-berlin.de

\clearpage\pagebreak\newpage

\begin{abstract} 
  A large offspring number diploid biparental multilocus population
  model of Moran type is our object of study.  At each timestep, a
  pair of diploid individuals drawn uniformly at random contribute
  offspring to the population.  The number of offspring can be large
  relative to the total population size.  Similar `heavily skewed'
  reproduction mechanisms have been considered by various authors
  recently, cf.\ e.g.\ Eldon and Wakeley (2006, 2008), and reviewed by
  Hedgecock and Pudovkin (2011).  Each diploid parental individual
  contributes exactly one chromosome to each diploid offspring, and
  hence ancestral lineages can only coalesce when in distinct
  individuals. A separation of timescales phenomenon is thus observed.
  A result of M\"{o}hle (1998) is extended to obtain convergence of
  the ancestral process to an ancestral recombination graph
  necessarily admitting simultaneous multiple mergers of ancestral
  lineages. The usual ancestral recombination graph is obtained as a
  special case of our model when the parents contribute only one
  offspring to the population each time.
  
  Due to diploidy and large offspring numbers, novel effects appear.
  For example, the marginal genealogy at each locus admits
  simultaneous multiple mergers in up to four groups, and different
  loci remain substantially correlated even as the recombination rate
  grows large.  Thus, genealogies for loci far apart on the same
  chromosome remain correlated.  Correlation in coalescence times for
  two loci is derived and shown to be a function of the coalescence
  parameters of our model.  Extending the observations by Eldon and
  Wakeley (2008), predictions of linkage disequilibrium are shown to
  be functions of the reproduction parameters of our model, in
  addition to the recombination rate.  Correlations in ratios of
  coalescence times between loci can be high, even when the
  recombination rate is high and sample size is large, in large
  offspring number populations, as suggested by simulations, hinting
  at how to distinguish between different population models.

\end{abstract}

\clearpage\pagebreak\newpage

Diploidy, in which each offspring receives two sets of chromosomes,
one from each of two distinct diploid parents, is fairly common among
natural populations.  Mathematical models in population genetics tend
to assume, however, that all individuals in a population are haploid,
simplifying the mathematics.  Mendel's Laws describe the mechanism of
inheritance as composed of two main steps, equal segregation (First
Law), and independent assortment (Second Law).  The First Law
proclaims gametes are haploid, i.e.\ carry only one of each pair of
homologous chromosomes.  Most models in population genetics are thus
models of chromosomes, or gene copies.  Mendel's Second Law proclaims
independent assortment of alleles at different genes, or loci, into
gametes.  Linkage of alleles on chromosomes, resulting in non-random
association of alleles at different loci into gametes, is of course an
important exception to the Second Law.

Coalescent processes \citep{K82,K82b,H83a,T83} describe the ancestral
relations of chromosomes (or gene copies) drawn from a natural
population.  The coalescent was initially derived from a
\cite{C74} haploid exchangeable population model.  Related ancestral
processes take into account population structure \citep{N90,H97},
selection \citep{KN97,NK97,EGT10}, and recombination between linked loci
\citep{H83b,G91,GM97}.  The coalescent has proved to be an important
advance in theoretical population genetics, and a valuable tool for
inference of evolutionary histories of populations.

Ancestral recombination graphs (ARG) \citep{H83b,G91,GM97} trace
ancestral lineages of gene copies at linked loci, in which linkage is
broken up by recombination.  An ARG is a branching-coalescing graph,
in which recombination leads to branching of ancestral chromosomes,
and coalescence to segments rejoining.  Coalescence events in an ARG
may not lead to coalescence of gene copies at individual loci.  An
example ARG for two linked loci is given below, labelled as $ARG(1)$,
with notation borrowed from \cite{D02}.  The labels $a$ and $b$ refer
to the two alleles (types) at locus 1 and 2, respectively.  A single
chromosome with two linked alleles is denoted by $(ab)$, while
chromosomes carrying ancestral alleles at only one locus are denoted
$(a)$ and $(b)$.  When coalescence occurs at either
  locus, the number of alleles at the corresponding locus is reduced
  by one.  The absorbing state, either $(ab)$ or $(a)(b)$, is reached when
alleles at both loci have coalesced.

\begin{displaymath}%
  \begin{split}%
    ARG(1):&\quad  \boldsymbol{(ab)}\boldsymbol{(ab)} \overset{r}{\to} \boldsymbol{(a)}\boldsymbol{(b)}\boldsymbol{(ab)}\overset{c}{\to}\boldsymbol{(ab)(b)}\\
    & \quad \overset{r}{\to} \boldsymbol{(a)(b)(b)}\overset{c}{\to} \boldsymbol{(a)(b)} \\\\
ARG(2):&\quad \boldsymbol{(ab)}\boldsymbol{(ab)} \overset{r}{\to} \boldsymbol{(a)}\boldsymbol{(b)}\boldsymbol{(ab)}  \overset{r}{\to} \boldsymbol{(a)}\boldsymbol{(b)}\boldsymbol{(a)}\boldsymbol{(b)}\\
& \quad \overset{c}{\to}\boldsymbol{(a)}\boldsymbol{(b)} \\
\end{split}\end{displaymath}%
In $ARG(1)$, the first transition is a recombination, denoted by
$\overset{r}{\to}$, followed by a coalescence $(\overset{c}{\to})$, in
which the two alleles at locus 1 coalesce.   Graph $ARG(1)$
serves to illustrate two important concepts we will be concerned with,
namely correlation in coalescence times between alleles at different
loci, and the restriction to binary mergers of ancestral lineages.

Correlation in coalescence times between types at different loci
follows from linkage.  Alleles at different loci can become associated
due to a variety of factors, including changes in population size,
natural selection, and population structure.  Within-generation
fecundity variance polymorphism induces correlation between a neutral
locus and the locus associated with the fecundity variance
\citep{T09}.  Sweepstake-style reproduction
\citep{H82,H94,B94,A88,PW90,A04,HP11}, in which few individuals
produce most of the offspring, has also been shown to induce
correlation in coalescence times between loci \citep{EW08}.
Understanding genome-wide correlations in coalescence times becomes
ever more important as multi-loci genetic data becomes ubiquitous.

The ARG exemplified by $ARG(1)$ is characterised by admitting only
binary mergers of ancestral lineages, i.e.\ exactly two lineages
coalesce in each coalescence event.  The restriction to binary mergers
follows from bounds on the underlying offspring distribution, in which
the probability of large offspring numbers becomes negligible in a
large population \citep{K82,K82b}.  Sweepstake-style reproduction, in
which few individuals contribute very many offspring to the
population, have been suggested to explain the `shallow' gene
genealogy observed for many marine organisms
\citep{H82,H94,A88,PW90,B94,A04,HP11}.  Large offspring number models
are models of extremely high variance in individual reproductive
output.  Namely, individuals can have very many offspring, or up to
the order of the population size with non-negligible probability
\citep{schweinsberg03,EW06,SW08,S03,BB09}.  Such models do predict shallow gene
genealogies, and can be shown to give better fit to genetic data
obtained from Atlantic cod \citep{A04} than the Kingman coalescent
\citep{BB08, BBS11, E11, birkner12}.  Different large offspring number models
will no doubt be appropriate for different populations, and the
identification of large offspring number population models for each
population is an open problem.  For the sake of simplicity and
mathematical tractability, the simple large offspring number model
considered by \cite{EW06} will be adapted to our situation.

The coalescent processes derived from large offspring number models
belong to a large class of multiple merger coalescent processes
introduced by \cite{DK99}, \cite{P99}, and \cite{S99}.  Multiple
merger coalescent processes ($\Lambda$-coalescents), as the name
implies, admit multiple mergers of ancestral lineages in each
coalescence event, in which any number of active ancestral lineages
can coalesce, and at most one such merger occurs each time.  In
simultaneous multiple merger coalescent processes \citep{MS01,S00},
any number of multiple mergers can occur each time, i.e.\ distinct
groups of active ancestral lineages can coalesce each time.  The
ancestral recombination graph derived from our diploid large offspring
number model admits simultaneous multiple mergers of ancestral
lineages, as exemplified in $ARG(2)$.  The last transition in $ARG(2)$
is a simultaneous multiple merger, in which the two types at each
locus coalesce to separate ancestral chromosomes.

In order to investigate correlations in coalescence
  times among loci due to skewed offspring distribution, we
  \emph{formally} derive an ancestral recombination graph, or a
  coalescent process for many linked loci, from our diploid large
  offspring number model.  The key to the proof of convergence to an
ancestral recombination graph from our diploid model lies in resolving
the separation of timescales phenomenon we observe.  Following
Mendel's Laws, the two chromosomes of an offspring come from distinct
diploid parents.  Chromosomes can therefore only coalesce when in
distinct individuals.  The ancestral process will consist of two
phases, a dispersion phase occurring on a `fast' timescale, and a
coalescence and recombination phase occurring on a `slow' timescale.
In the dispersion phase, chromosomes paired together in diploid
individuals disperse into distinct individuals.  Coalescence and
recombination will only occur on the slow timescale.  Similar
separation of timescales issues arise in models of populations
structured into infinitely many subpopulations (demes) \citep{TV09}.
When viewing the diploid individuals in our model as `demes', our
scenario departs from those describing structured populations by
allowing only active ancestral lineages residing in {\em separate}
`demes' to coalesce.  A simple extension of a result of \cite{Moe98}
yields convergence in our case.

The limiting process we formally obtain is an ancestral recombination
graph for many loci admitting \emph{simultaneous} multiple mergers of
ancestral chromosomes (lineages).  In simultaneous multiple merger
coalescent processes, so-called $\Xi$-coalescents, different groups of
active ancestral lineages can coalesce to different ancestors at the
same time.  Such coalescent processes were first studied as more
abstract mathematical objects by \cite{S00}, and derived from general
single-locus population models by several authors
\citep{MS01,S03,SW08,BBMST09}.  A $\Xi$-coalescent with necessarily up
to quadruple simultaneous multiple mergers arises at each marginal
locus (ie.\ considering each locus separately) in our model, since
four parental chromosomes are involved in each reproduction event.
This structure is intrinsically owed to our diploidy assumptions.

Formulas for the correlation in coalescence times between two alleles
at two loci are obtained using our ancestral recombination graph
(ARG). As predicted by J.E.\ Taylor (personal communication), these
correlations will not necessarily be small even for loci separated by
high recombination rate.  This is a novel effect not visible in
classical models.  The correlation structure will of course depend on
the underlying coalescent parameters introduced by the large offspring
number model we adopt. An approximation of the expected value of the
statistics $r^2$, commonly used to quantify linkage disequilibrium, is
also investigated using our ARG.  In addition, we employ our ARG to
investigate correlations in ratios of coalescence times between loci
for samples larger than two at each locus, using simulations.
\smallskip

\textbf{ A diploid population model with multilocus recombination and skewed offspring distribution}

\textbf{\textsl{ The forward population model}}

Consider a population consisting of $N \in \mathbb{N} \equiv \{1, 2,
\ldots \}$ diploid individuals, meaning that each individual contains
two \emph{chromosomes}.  Each chromosome is structured into  
$L\in \mathbb{N}$ loci.  We assume  Moran-type dynamics: 
At each timestep (`generation'), either a
\emph{small} or a \emph{large} reproduction event occurs.
In a \emph{small} reproduction event, a single individual chosen
uniformly at random from the population dies, and two other distinct
individuals are chosen as \emph{parents}.  A diploid \emph{offspring}
is then formed by choosing one chromosome from each parent (see
Figure~\ref{figure0}).  The parents always persist.  A small
reproduction event occurs with probability $1 - \varepsilon_{_N}$, in
which $\varepsilon_{_N} \in (0,1)$ depends on $N$.  In a
$\emph{large}$ reproduction event, a fraction $\psi \in (0,1)$ of the
population perishes, meaning that $\lfloor \psi N \rfloor $
individuals die ($\lfloor x \rfloor$ for $x \geq 0$ denotes the
largest integer smaller than $x$).  Two distinct individuals are then
chosen uniformly from the remaining $N - \lfloor \psi N \rfloor$
individuals to act as parents of $\lfloor \psi N \rfloor$ offspring,
and each offspring is formed independently by choosing one
(potentially recombined) chromosome from each parent (see Figure~1).
The population size always stays constant at $N$ diploid individuals.
Individuals that neither reproduce nor die simply persist.

\begin{figure}%
\includegraphics[height=3in,width=3in,angle=-90]{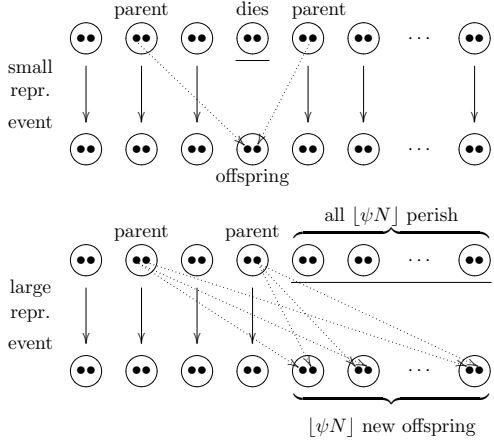}%
\caption{Illustration of `small' and `large' reproduction events
  without recombination.  The dotted arrows indicate the copying of
  parental chromosomes into offspring chromosomes.  The solid arrows
  indicate individuals that persist.  }
\label{figure0}%
\end{figure}%

Given the two parents, genetic types of the offspring individuals will
then be obtained as follows.  Each parent generates a large number of
potential offspring chromosomes, of which a fraction $1 - r_{_N}$ are
exact copies of the original parental chromosomes, and a fraction
$r_{_N}$ are $\emph{recombinants}$. Each chromosome is structured into
$L$ loci. Recombination occurs only between loci, and never within.
If recombination between a pair of chromosomes in a parent occurs
between loci $\ell$ and $\ell + 1 \in \{1, \dots, L\}$ (where we say
that $X \in \{1, \dots, L-1\}$ is the {\em crossover point}), the two
chromosomes exchange types at all loci from $\ell + 1$ to $L$.  Only
one crossover point is allowed in each recombination event.  Let
$\rN^{(\ell)}$ denote the probability of recombination between loci
$\ell$ and $\ell + 1$ (i.e., the probability that the potential
crossover point $X$ equals $\ell$).  An offspring chromosome is a
recombinant with probability $\rN = \rN^{(1)} + \cdots + \rN^{(L -
  1)}$.  Given that recombination happens, we thus have
$$
\mathbb{P}\{X = \ell\} =  \frac{\rN^{(\ell)}}{\rN^{(1)} + \cdots + \rN^{(L - 1)} },  \quad 1 \le \ell \le L - 1.
$$
Each pair of recombined chromosomes is formed independently of all other pairs.
From this large pool of chromosomes, each new
offspring is randomly assigned (independently of all other offspring in the case of a large reproduction event), one potentially recombined chromosome generated by each parent.  In addition, the reproduction mechanism in different generations is assumed to be independent.



\textbf{\textsl{ Ancestral relationships - notation}}

Now we switch from the forward population model to its ancestral
process, running backwards in time.  Our sample will consist of $n \in
\{1, \dots, 2N\}$ chromosomes, each subdivided into $L$ loci. Hence,
we need to keep track of the ancestry of $nL$ segments
(types/alleles).  This implies that the different segments could end
up on up to $nL$ distinct chromosomes in $nL$ distinct ancestral
individuals.  The required notation will now be introduced, and our
discourse will therefore necessarily become a little bit technical.
However, we believe that a precise description of the objects we are
working with is essential.  The key to understand our notation is that
we are working with enumerated chromosomes, and ordered loci on
chromosomes.

At present (that is, time step $m=0$), assume that we consider an
even number $n$ of chromosomes carried by $n/2$ individuals. The
chromosomes are enumerated from 1 to $n$, attaching consecutive
numbers to chromosomes found in the same individual. Our ancestral
process will keep track of the chromosomal ancestral information, that
is, which locus is ancestral to which set of sampled chromosomes. That
is, in each generation $m \in \N_0$ (backward in time), we will record
all chromosomes which are {\em active} in the sense that they carry at
least one locus which is ancestral to the same locus of at least one
chromosome in generation 0.  Denote the number of active chromosomes
in generation $m\in \N_0$ by $\beta(m) \in \N$.  The number $\beta(m)$
of active chromosomes can both increase, due to recombination, and
decrease, due to coalescence, going back in time.

Now we explain our notation for the loci.     For each chromosome $j \in
[n] := \{1, \dots, n\}$, denote by $\loc^{(j)}_\ell(m)$ locus $\ell \in [L]$
on chromosome $j$ at time $m$. The subsets $\loc^{(j)}_\ell(m)$ of
$[n]$ contain all the numbers of chromosomes at present (time step $0$) to
which locus $\ell$ on active chromosome number $j$ at 
time step $m$
is {\em ancestral}.  With this convention, and for each $m \in \N$ and
$\ell \in [L]$, the collection
$$
\{\loc^{(j)}_\ell(m), j=1, \dots, \beta(m)\}
$$ which describes the configuration of segments (i.e.\ which have
coalesced and which have not) at locus $\ell$ at time $m$, is a
partition of $[n]$, i.e.
\begin{displaymath}%
  \loc^{(j)}_\ell(m) \cap \loc^{(\hat{\jmath})}_\ell(m) = \emptyset \quad \mbox{ for } \quad j \neq \jp;
  \end{displaymath}%
  and
  \begin{displaymath}%
 \bigcup_{j=1}^{\beta(m)} \loc^{(j)}_\ell(m) = [n].
\end{displaymath}%
Thus, with our notation we can correctly describe the configuration of
segments among chromosomes at any given time.  By $C^{(j)}(m)$ we
denote chromosome number $j$ at time $m$. At time $m = 0$,
$$
C^{(j)}(0):=\left\{ \loc^{(j)}_1(0), \dots, \loc^{(j)}_L(0)\right\}:= \big\{ \{j\}, \ldots , \{j\}\big\}. 
$$

\newcommand{\set}[1]{\ensuremath{ \left\{ #1 \right\} }}

For $m > 0$, consider the $j$-th active chromosome at generation $m$,
where $j \in [\beta(m)]$.  The corresponding ancestral
information at generation $m$ is encoded via 
an ordered list of subsets of
$[n]$, setting
\begin{equation}
\label{eq:partrepr}%
\begin{split}%
C^{(j)}(m) &:=\left\{ \loc^{(j)}_1(m), \dots, \loc^{(j)}_L(m)\right\}, \\
 &  \loc^{(j)}_\ell(m) \subset [n], \quad  \ell \in [L].
\end{split}%
\end{equation}%

Chromosomes are carried by diploid {\em individuals}. Keeping track of
the grouping of active chromosomes into individuals will be important,
since by our diploid reproduction mechanism, chromosomal lineages can
only coalesce when in {\em distinct} individuals (see Example B
below).  In analogy with our previous nomenclature for our ancestral
process, an {\em active} individual will carry at least one (and at
most two) active chromosome(s). Let $b(m)$ denote the number of active
individuals at generation $m$ where $\beta(m)/2 \le b(m) \le \beta(m)$
for all $m$.  The ordered list of active chromosomes and the number of
active individuals (called a `configuration') at time $m \ge 0$ is
denoted by
\begin{equation}
\label{eq:config}
\xi^{n,N}(m) := \left\{ C^{(1)}(m), \ldots,  C^{(\beta(m))}(m); \, b(m) \right \}. 
\end{equation}
An individual number $i$ at generation $m$ is denoted by $\ind_i(m)$,
for $i \in [b(m)]$. An active individual is \emph{single-marked}, if
carrying one active chromosome, and is \emph{double-marked}, if
carrying two active chromosomes.  Specifying the arrangement of
chromosomes in individuals completes our description of the
(prelimiting) ancestral process.  However, since all active
individuals are single-marked in the limiting process, our description
of the arrangement of chromosomes in individuals is given in Section
1.1.1 in the Appendix.
%
%
  That is, each configuration $\xi^{n,N}(m)$ begins with the
  $2(\beta(m)-b(m))$ ordered consecutive chromosomes of the
  $\beta(m)-b(m)$ double marked individuals, followed by the
  $2b(m)-\beta(m)$ chromosomes contained in single-marked individuals.
  With this convention, the set of single- and double marked
  individuals and the grouping of chromosomes into individuals at
  generation $m$ is uniquely determined by a configuration $\xi^{n,N}(m)$ of form  \eqref{eq:config}.
  For notational convenience, the time index $m$ will be omitted if there is no ambiguity.\\

For a given sample size $n$,  the set of all possible 
  ancestral configurations $\xi^{n,N}$
  will be denoted by $\mathscr{A}_n$.  The subset
  $\mathscr{A}_n^{\texttt{sm}} \subset \mathscr{A}_n$ of all
  configurations $\xi^{n,N}=\left\{ C^{(1)}, \ldots, C^{(\beta)}; b
  \right \}$ with $b=\beta$, i.e.\ configurations consisting only of
  single-marked individuals, will play an important r\^{o}le later
  on.  Indeed, all configurations in the limiting model  will be confined to the set 
  $\mathscr{A}_n^{\texttt{sm}}$,  and the pairing of chromosomes in individuals will become irrelevant. 
  
  The mapping $\mathsf{cd}$ (`complete
  dispersion') $$\mathsf{cd} : \mathscr{A}_n \to
  \mathscr{A}_n^{\texttt{sm}}$$ breaks up the pairing of chromosomes
  into diploid double-marked individuals. More precisely, we define
\begin{equation}
\label{eq:cd}%
\mathsf{cd}\Big(\left\{ C^{(1)}, \ldots,  C^{(\beta)}; b \right \}
\Big) := \left\{ C^{(1)}, \ldots,  C^{(\beta)}; \beta \right \}.
\end{equation}
Configurations in $\mathscr{A}_n^{\texttt{sm}}$ describe 
configurations in which all active individuals are single marked, i.e.\ carry only one active chromosome.

\medskip

The effects of recombination and coalescence on the ancestral
configurations in the case of two typical situations will now be illustrated.  Example A will illustrate 
recombination, and Example B will illustrate  coalescence of two chromosomes. 

\medskip

{\bf Example A.} Suppose the most recent previous event in the history of a given configuration $\xi^{n,N}(m)$  was a {\em small reproduction event} (at time $m+1$), and suppose that the resulting offspring individual is currently part of our configuration at time $m$, but neither of its parents is, and that the offspring individual is single-marked, i.e.\ carries one active chromosome. We obtain $\xi^{n,N}(m+1)$ as follows:
%
\begin{itemize}
\item If there is no recombination during the reproduction event, then the configuration in the
  previous generation remains unchanged, i.e.\ $\xi^{n,N}(m+1)=\xi^{n,N}(m)$.
\item If there is recombination, say at a crossover point $X \in
  \{1, \dots, L-1\}$, suppose the (single) offspring chromosome
  is $$C^{(j)}(m) = \left\{ \loc_1^{(j)}(m), \ldots,
    \loc_L^{(j)}(m)\right\}.$$ Necessarily, the two parental
  chromosomes will be part of the configuration $\xi^{n,N}(m+1)$,
  residing in the same double-marked individual.  More precisely, the
  two parental chromosomes, say $C^{(\jt)}(m+1)$ and $C^{(\jt + 1)}(m+1)$, are determined by (for $\ell \in [L]$)
\begin{displaymath}%
\loc_\ell^{(\jt)}(m+1) = 
\begin{cases}
\loc_\ell^{(j)}(m): & 1 \le \ell \le X, \\
\emptyset: & X + 1 \le \ell \le L,
\end{cases} 
\end{displaymath}%
and
\begin{displaymath}
\loc_{\ell}^{(\jt + 1)}(m+1) = 
\begin{cases}
\emptyset: &  1 \le \ell \le X, \\
\loc_\ell^{(j)}(m): & X + 1 \le \ell \le L.
\end{cases} 
\end{displaymath}%
in which $\emptyset$ denotes loci not carrying any ancestral segments.   
The offspring chromosome is of course not part of $\xi^{n,N}(m+1)$.
This transition can be partially trivial (a `silent recombination'
event), if the crossover point is not in an `active' area, i.e.\ if
$\loc_\ell^{(j)} = \emptyset$ for  $X + 1 \le \ell \le L$ (or for 
all $1 \le \ell \le X$).   
By way of example, with $L=3$, if chromosome $C^{(j)} = \big\{\set{j},
  \set{j}, \set{j}\big\}$ was a recombinant, and the crossover point
occurred between loci 2 and 3,   the two parental chromosomes are given by 
$C^{(\tilde{\jmath})} =  \big\{\set{j}, \set{j}, \emptyset\big\}$  and   $C^{(\tilde{\jmath} + 1)} =  \big\{\emptyset, \emptyset, \set{j} \big\}$.
\end{itemize}

{\bf Example B.}  Suppose the most recent previous event in the history of a given configuration $\xi^{n,N}(m)$ of chromosomes at generation $m$ is a {\em small reproduction event} at time $m+1$, leading to a coalescence of lineages. This is the case e.g.\ if both a single-marked offspring individual with active chromosome $C^{\jp}(m)$ is in our configuration $\xi^{n,N}(m)$, as well as its single marked parent (say with currently active chromosome $C^{j}(m)$), from which it actually obtained its active chromosome. Then, to obtain the configuration $\xi^{n,N}(m+1)$, the offspring  chromosome $C^{(\jp)}(m)$ is deleted, and the resulting ancestral chromosome $C^{(j)}(m+1)$ is given by the family of the union of the sets $\loc_\ell^{(j)}$ and $\loc_\ell^{(\jp)}$,
\begin{multline}%
C^{(j)}(m+1) = \left\{\loc_1^{(j)}(m) \cup \loc_1^{(\jp)}(m), \ldots , \right.\\ \left. \loc_L^{(j)}(m) \cup \loc_L^{(\jp)}(m) \right \}.
\end{multline}%
All other chromosomes in $\xi^{n,N}(m+1)$ are copied from $\xi^{n,N}(m)$.  
Again, taking $L=3$, if  chromosomes  $C^{(j)} = \big\{ \set{j},\set{j}, \set{j} \big\}$ and   $C^{(k)} = \big\{\set{k},\set{k}, \set{k}\big\}$  coalesce, the resulting   ancestral chromosome is given by $C^{(j)} =  \big\{\set{j, k}, \set{j, k}, \set{j, k} \big\}$.


\textbf{\textsl{ Scaling and classification of transitions}}

In order to obtain a non-trivial scaling limit for $\{\xi^{n,N}(m)\}$
as $N \to \infty$, 
the limit theorem of \citep{MS01} (cf also the special case considered 
in \citep{EW06}) suggests one should, for
some constant $c > 0$, choose probability $1 - c/N^{2}$ for the small
reproduction events, $c/N^2$ for the large reproduction events, i.e.,
setting
\begin{equation}
\label{eq:largereprodratescaling}
\varepsilon_N= c/N^{2}, 
\end{equation}
and speed up 
time by $N^2$.  For the recombination rate to be non-trivial in the limit 
(i.e.\ neither 0 nor infinitely large),
 we require that
all recombination values $\rN^{(\cdot)}$ scale in
units of $N$, i.e.\ for each crossover point  $\ell \in [L]\setminus \{L\}$,  
\begin{equation}%
\label{eq:recombratescaling}
\rN^{(\ell)} := \tfrac{r^{(\ell)}}{N}, \quad 0 < r^{(\ell)} < \infty. 
\end{equation}
Thus, even though our timescale is in units of $N^2$ timesteps,
recombination is scaled in units of $N$ timesteps.  
On the level of single lineages the probability of recombination is of
the order $O\left(N^{-2}\right)$. Indeed, after a small reproduction
event, the probability of drawing an offspring is $1/N$.  The
probability that the offspring carries a recombined chromosome is of
order $O\left(1/N\right)$.

Given the cornucopia of possible transitions from $\xi^{n,N}(m)$ to
$\xi^{n,N}(m+1)$, it will be important to identify those transitions
which are expected to be visible in the limiting process.

\medskip

All possible transitions fall  into the following three regimes:
\begin{itemize}%
\item 
Those transitions which happen at probability of order $O(N^{-2})$ per generation, which will be visible in the
limit (since time will be scaled by $N^2$). They will be called {\em effective transitions} and will appear at a finite positive rate in the limit. 
\item
Further, there are transitions which happen less frequently, 
typically with probability of order $O(N^{-3})$ or smaller per generation, which will thus become negligible as $N \to \infty$ and hence be invisible in the limit. These will be called {\em negligible transitions}.
\item Finally, there are transitions which happen much more frequently
  (with probability of order $O(N^{-1})$ or even $O(1)$ per generation).
  At first sight, one might think that their presence might lead to
  chaotic behaviour in the limit.  However, this will not be
  case. Instead, these transition will happen `instantaneously' in the
  limit, and result in a projection of the states of our process from
  $\mathscr A_n$ into the subspace $\mathscr A_n^{\texttt{sm}}$, which
  will be the limiting statespace. This will be proved below. Such
  transitions will be called {\em projective} or {\em instantaneous}
  transitions.  The identity transition is a special case of a
  projective transformation.
\end{itemize}

In the Appendix (section \ref{sn:transitions}), a full classification
of all transitions into the above groups is provided.


\textbf{\textsl{ Instantaneous and effective transitions}}

The most important transitions and their effect for the limiting
process will now be described in detail.
Consider the following most recent events in the history of a set of
lineages, i.e.\ events occurring at time $m+1$, from the perspective
of the ancestral process $\xi^{n,N}(m)$ at time $m$:

\begin{itemize}

\item {\bf Event 1 (silent)}: A small reproduction event occurs, but the
  offspring is not active. This is the most likely event, and is of
  the order $O(1)$, but does not affect our ancestral configuration
  process $\xi^{n,N}(m)$, i.e.\ $\xi^{n,N}(m+1)=\xi^{n,N}(m)$. This
  event leads to an identity transition (a trivial {\em instantaneous
    transition}).

\item {\bf Event 2 (dispersion)}: A small reproduction event occurs, the offspring
  is active in our sample but neither parent is, and recombination
  does not occur. This is a relatively frequent event which occurs
  with a probability of the order $O(N^{-1})$ per generation (since
  the probability that the offspring is in the sample is $b(m)/N$).
  If the offspring carries only one active chromosome, we again see an
  identity transition, i.e\ $\xi^{n,N}(m+1)=\xi^{n,N}(m)$.  If the
  offspring carries two active chromosomes, i.e.\ is a double-marked
  individual, the two active chromosomes will disperse to two separate
  individuals, who will then become single-marked individuals.
  Formally, for $\xi=\left\{ C^{(1)}, \ldots, C^{(\beta)}; b \right \}
  \in \mathscr{A}_n$ with at least one double-marked individual $(b <
  \beta)$, define the map $\mathsf{disp}_i(\cdot): \mathscr A_n \to
  \mathscr A_n $ dispersing the chromosomes paired in individual $i$,
\begin{multline}
\label{eq:splitdef}
\mathsf{disp}_i(\xi) = \!
 \left\{ C^{(1)}, \ldots, C^{(2i-2)}, C^{(2i+1)}, 
C^{(2i+2)}, \right. \\ \left. \ldots,  C^{(2(\beta-b))}, C^{(2i-1)}, C^{(2i)}, C^{(2(\beta-b)+1)}, \right. \\ \left. \ldots, C^{(\beta)}; b+1 \!\right \}
\end{multline}
if $1 \le i \leq \beta-b$ and $\mathsf{disp}_i(\xi) := \xi$ otherwise.
Recall that the $i-th$ double-marked individual has chromosomes
labelled $2i-1$ and $2i$.  For $\xi^{n,N}(m)$, if the $i$-th double
marked individual is affected, we have the transition
$\xi^{n,N}(m+1)=\mathsf{disp}_i(\xi^{n,N}(m))$.

The dispersion events will happen instantaneously as $N \to \infty$
(recall we are speeding time up by $N^2$), and thus will, in the
limit, lead to an immediate complete dispersion of all chromosomes
paired in double-marked individuals. If in the course of events, a new
double-marked individual emerges due to pairing of active chromosomes
in the same diploid individual, a dispersion of the chromosomes will
occur immediately. {\bf Event 2} will hence result in a permanent {\em
  instantaneous transition}, mapping our current state $\xi \in
\mathscr{A}_n$ into the subspace $\mathscr{A}_n^{\texttt{sm}}$ by
means of the map $\mathsf{cd}$ defined in \eqref{eq:cd}.  Our limiting
process will thus live, with probability one for each given $t >0$, in
$\mathscr{A}_n^{\texttt{sm}}$, even if we start with a configuration
from $\mathscr{A}_n \setminus \mathscr{A}_n^{\texttt{sm}}$ at time
$t=0$.

\item {\bf Event 3 (recombination)}: A small reproduction event occurs, a
  single-marked offspring but neither parent is in our sample, and
  recombination affecting the active chromosome at a crossover point
  $x$.  This event has probability of the order $O(N^{-2})$ per
  generation, and will thus be visible with finite positive rate in
  the limit.  It is an {\em effective transition}, which can be
  described formally as follows.  Define the recombination operation
  $\mathsf{recomb}$ acting on chromosome $j$ and crossover point
  $x$ for a configuration $\xi \in \mathscr{A}_n^{\texttt{sm}}$
  as
\begin{multline} 
\mathsf{recomb}_{j,x}(\xi) := 
\left\{ C^{(1)}, \ldots, C^{(j-1)}, \tilde{C}^{(j,1)},  \right. \\
\tilde{C}^{(j,2)}, \left. C^{(j+1)}, \dots, C^{(\beta)}; 
\beta+1 \right \},
\end{multline} 
where 
\begin{displaymath}%
\tilde{C}^{(j,1)}=\left\{ \tilde{\loc}^{(j,1)}_1, \dots, 
\tilde{\loc}^{(j,1)}_L \right\} 
\end{displaymath}
with
\begin{displaymath}%
\quad \tilde{\loc}_\ell^{(j,1)} = \begin{cases}
\loc_\ell^{(j)} & : \quad 1 \le \ell \le x - 1, \\
\emptyset  & : \quad x \le \ell \le L,
\end{cases} 
\end{displaymath}%
and 
\begin{displaymath}
\tilde{C}^{(j,2)}=\left\{ \tilde{\loc}^{(j,2)}_1, \dots, 
\tilde{\loc}^{(j,2)}_L \right\} 
\end{displaymath}%
with
\begin{displaymath}%
\tilde{\loc}_\ell^{(j,2)} = 
\begin{cases}
\emptyset: & \quad 1 \le \ell \le x - 1, \\
\loc_\ell^{(j)}: &\quad x \le \ell \le L
\end{cases} 
\end{displaymath}%
(if one of $\tilde{C}^{(j,1)}$, $\tilde{C}^{(j,2)}$ equals $\left\{
  \emptyset, \dots, \emptyset \right\}$, we define
$\mathsf{recomb}_{j,\alpha}(\xi) := \xi$, giving rise to a {\em silent} recombination event).

\item {\bf Event 4 (pairwise coalescence)}: A small reproduction event
  occurs, one single-marked parent and a single-marked offspring are
  in the sample, the active chromosome is inherited from the parent in
  the sample, and recombination does not occur. This event occurs with
  probability of order $O(N^{-2})$ and will therefore be visible in
  the limit with finite positive rate, hence gives rise to an {\em
    effective transition}.  It will lead to a binary coalescence of
  lineages and can formally be described as follows.   The ancestral chromosome $\tilde{C}^{(j_1)}$  
  formed by the coalescence of chromosomes $j_1$ and $j_2$ is given by 
    \begin{align}%
\tilde{C}^{(j_1)} = \left\{\loc_1^{(j_1)} \cup \loc_1^{(j_2)}, \ldots , \loc_L^{(j_1)} \cup \loc_L^{(j_2)} \right \}
\end{align}%
if $1 \le j_1 < j_2 \le \beta$.   
     Define the
  binary coalescence operation $\mathsf{pairmerge}$ acting on
  chromosomes $j_1$ and $j_2$ $(1 \le j_1 < j_2)$ in a configuration
  $\xi \in \mathscr{A}_n^{\texttt{sm}}$ as
\begin{multline} 
\mathsf{pairmerge}_{j_1,j_2}(\xi) := 
\left\{ C^{(1)}, \ldots, \tilde{C}^{(j_1)}, \dots, 
 \right. \\ C^{(j_2-1)},  \left.  C^{(j_2+1)}, \dots, C^{(\beta)}; 
\beta-1 \right \} 
\end{multline}%
if $1 \le j_1 < j_2 \le \beta$ (otherwise, we put
$\mathsf{pairmerge}_{j_1,j_2}(\xi):=\xi$).

\item {\bf Event 5 (multiple merger coalescence)}: A large reproduction event occurs,
  neither parent but (possibly several) single marked offspring are in
  our sample, and recombination does not occur. This is again an event
  with probability of order $O(N^{-2})$ per generation and therefore
  will be visible in the limit with finite positive rate, hence gives
  rise to an {\em effective transition}. The offspring chromosomes
  will be assigned their parental chromosomes independently and
  uniformly at random, since due to an immediate `complete dispersion'
  via {\bf Event 2} each offspring individual will carry precisely one
  active chromosome. Now  we formally define the multiple coalescence
  operation $\mathsf{groupmerge}$ for $\xi \in
  \mathscr{A}_n^{\texttt{sm}}$ and pairwise disjoint subsets
  $J_1,J_2,J_3,J_4 \subset [\beta]$ in which either at least one $|J_i|\geq 3$
  or at least two of the $|J_i| \ge 2$.  This transition is, thus, really
  different from a $\mathsf{pairmerge}$ transition. Let 
 $J_j$ denote the set of offspring chromosomes derived from
parental chromosome $j$. Then 
\begin{multline} 
\mathsf{groupmerge}_{J_1,J_2,J_3,J_4}(\xi) := 
\left\{ \tilde{C}^{(1)}, \tilde{C}^{(2)}, 
\tilde{C}^{(3)},  \right. \\  \tilde{C}^{(4)},  
\left.   C^{(j)}, j \in [\beta] \setminus (J_1 \cup J_2 \cup J_3 \cup J_4) ; 
\tilde{\beta} \right \} 
\end{multline}
with $\left((x)^+ := \max(x, 0)\right)$ $$\tilde{\beta} = \beta - \sum_{j = 1}^4(|J_j|-1)^+  $$
and the four parental chromosomes, at least one of which is involved in a merger,  are given by  $(1 \le i \le 4)$, 
\[
\tilde{C}^{(i)} 
= \left\{ \bigcup_{j \in J_i} \loc_1^{(j)}, \ldots , 
\bigcup_{j \in J_i} \loc_L^{(j)} \right \}.
\]
The chromosome(s) $C^{(j)}$ appaering in 
$\mathsf{groupmerge}_{J_1,J_2,J_3,J_4}(\xi)$ denote the chromosomes in
$\xi$ that are not involved in a merger.

\item{\bf All other events}: Will either not affect our ancestral process, or have a probability of order smaller than $N^{-2}$ so that they will be absent in the limit after rescaling. A complete classification of these events will be given in the Appendix (section~\ref{sn:transitions}).

\end{itemize}%

\textbf{\textsl{ The limiting dynamics and state space}}

  The expected dynamics of the limiting
continuous time Markov chain $\{\xi(t), t \ge 0\}$, taking values in
$\mathscr{A}_n$, as $N\to \infty$, will now briefly be discussed.

\begin{itemize}%
\item Complete dispersion ({\bf Event 2}) of the sampled chromosomes
  is the first event to occur (between times $t=0$ and $t=0^+$).  By
  $\ind_i$ we denote individual number $i$ (see section 1.1.1 in
  Appendix).  At time $t = 0$ when $\xi(0) \in \mathscr{A}_n$ we
  assume all $n$ sampled chromosomes are paired in double-marked
  individuals ($n$ even);
\begin{multline}%
  \label{eq:initcond}
  \xi(0) = \left\{ \ind_i: \ind_i = \left\{ C_0^{(2i - 1)}, C_0^{(2i)} \right\}, \right. \\  1 \le i \le n/2   \Big\}. 
\end{multline}%
Immediately (at time $0^+$), the chromosomes disperse into single-marked individuals,
\begin{equation}%
   \label{eq:initcond.split}%
  \begin{split}%
\xi(0^{+})  & =  \mathsf{cd}(\xi(0))\\ 
& =\left\{ \ind_i: \ind_i = \left\{ C_0^{(i)}, \emptyset \right\}, \, 1 \le i \le n   \right\}  \\
& =  \left\{ C_{0^{^+}}^{(1)}, \ldots, C_{0^{^+}}^{(n)}; n  \right\} \in \mathscr{A}_n^{\texttt{sm}}.
 \end{split}%
\end{equation}

\item Throughout the evolution of the process, whenever double marked individuals appear (e.g.\ from a coalescence of lineages event), {\bf Event 2} will immediately change our configuration 
to the corresponding `all dispersed'-configuration, i.e., for each $t>0$, 
$$
\xi(t^+)= \mathsf{cd}(\xi(t)) \in \mathscr{A}_n^{\texttt{sm}}.
$$
Such `flickering' states will not affect any quantities of interest of
our genealogy, so we can assume that they will be removed from the
limit by choosing the c\`adl\`ag modification of $\{\xi(t), t \ge
0\}$, taking only values in $\mathscr{A}_n^{\texttt{sm}}$ for all
$t>0$ (this modification does not affect the finite-dimensional
distributions of $\{\xi(t), t \ge 0\}$).

\item Recombination ({\bf Event 3})  appears in the limiting process at total rate 
 $r = r^{(1)} + \cdots + r^{(L - 1)}$, where a certain recombination involving 
 a given crossover point $\ell$  appears with rate $r^{(\ell)}$ 
 on any lineage. 
 Indeed, from our scaling considerations, we have for the
 probability of not seeing a recombination at $\ell$ in a small
 resampling event for more than $N^2t$ scaled time units for a given single-marked individual
 satisfies $(r_N^{(\ell)} = r^{(\ell)}/N)$
\[
\Big(1- \Big(1-\frac {c}{N^2}\Big) r_N^{(\ell)} \frac 1N\Big)^{N^2 t} \to e^{-r^{(\ell)}t},
\] 
as $N \to \infty$ (recall (\ref{eq:recombratescaling}); the
probability for any given individual to be the child in a small
reproduction event is $1/N$), hence the waiting time for
this event to happen is exponential with rate $r^{(\ell)}$.

\item Coalescences appear according to the effective transitions described by {\bf Event 4} and {\bf Event 5}. From the point of view of a given pair of active chromosomes in different individuals, a single pairwise coalescence will occur at rate $1 + c \tfrac{\psi^2}{4} C_{\beta;2;\beta-2}$ with 
$C_{\beta;2;\beta-2}$ from (\ref{eq:Xirate}) (with $r=1$, $s=\beta-2$), 
where the $1$ comes from a pairwise coalescence 
according to a small reproduction event, and the $c \tfrac{\psi^2}{4} C_{\beta;2;\beta-2}$ from a large merger event (the rates can be easily derived from considerations similar to the recombination rate $r$ above), recalling that both coalescing chromosomes have to `successfully flip a $\psi$-coin' in order to take part in the large coalescence event, and then are uniformly distributed into four groups according to the
choice of any of the four potential parental chromosomes. 

Given large coalescence events (involving at least three individuals,
or at least two simultaneous pairwise mergers) happen with overall
rate $c \tfrac{\psi^2}{4}$ times the corresponding coalescence rate of
a $\Xi$-coalescent, obtained from the number of individuals taking
part in the merger independently with probability $\psi$.  The
participating individuals are then being distributed uniformly into
four groups according to the chosen parental chromosome. The
corresponding rate is given in the third line of \eqref{eq:DefG} (cf
also \eqref{eq:Xirate}).
\end{itemize}

\textbf{\textsl{ The limiting ancestral process}}

According to the above consideration, it is now plausible to consider
the following limiting Markov chain as the ancestral limiting
process. This fact will be proved below, with most computations
provided in the Appendix.  The $m$-th falling factorial is given by
$(a)_m := a(a-1)\cdots(a-m+1)$, $(a)_0 := 1$.  The operations
$\mathsf{pairmerge}$, $\mathsf{recomb}$ and $\mathsf{groupmerge}$ for
elements of $\mathscr{A}_n^{\texttt{sm}}$ were defined above in the
section on scaling.  Now we define the generator of the
continuous-time ancestral recombination graph derived from our model.

\begin{definition}[Limiting multilocus diploid ancestral recombination graph]
\label{def:limit}
The 
continuous-time Markov chain $\{\xi(t), t \ge 0\}$ with values in $\mathscr{A}_n^{\texttt{sm}}$, initial condition $\xi(0):= \mathsf{cd}(\xi)$ for $\xi \in \mathscr{A}_n$ and transition matrix $G$, with entries for elements $\xi', \xi \in \mathscr{A}_n^{\tt sm}, \,\xi'\neq\xi,$ is given by 
$(J := \left(J_1, \ldots, J_4\right))$, 
\begin{equation}%
\label{eq:DefG}%
\begin{split}%
& G(\xi,\xi')  \\
= & \begin{cases} 
1 + c \frac{\psi^2}{4} C_{\beta;2;\beta-2}
& \text{if} \; 
\xi'= \mathsf{pairmerge}_{j_1,j_2}(\xi) \\
r^{(\ell)} & \text{if} \; \xi'= \mathsf{recomb}_{j,\ell}(\xi) \\
c \frac{\psi^2}{4} C_{\beta;|J|} 
& \text{if} \; \xi'= \mathsf{groupmerge}_{J}(\xi) \\
0 & \text{for all other $\xi'\neq\xi$}
\end{cases}%
\end{split}%
\end{equation}
(where in the penultimate line we only consider cases where either at least 
one $|J_i|\geq 3$ or at least two of the $|J_i| \ge 2$), with 
$$C_{\beta; |J|} := C_{\beta;|J_1|,|J_2|,|J_3|,|J_4|;\beta-(|J_1|+|J_2|+|J_3|+|J_4|)}$$ and $(s = b - k_1 - \cdots - k_r \ge 0, \quad x \wedge y := \min(x,y))$
\begin{multline}
\label{eq:Xirate}
C_{b;k_1,\dots,k_r;s} = \frac{4}{\psi^2}  
\sum_{l=0}^{s \wedge (4-r)} {s \choose l} \frac{(4)_{r+l}}{4^{k_1+\cdots+k_r+l} }  \\
 \cdot (1-\psi)^{s-l} (\psi)^{k_1+\cdots+k_r+l}
\end{multline}%
\end{definition}
For the diagonal elements, one has of course
\begin{align} 
\label{eq:DefGdiag}
G(\xi,\xi) = - \sum_{\xi'\neq\xi, \, \xi' \in \mathscr{A}_n^{\tt sm}} 
\hspace{-1em}G(\xi,\xi').
\end{align}

The 
rates in \eqref{eq:Xirate} are the transition
rates of the $\Xi$-coalescent (a \emph{simultaneous} multiple merger
coalescent) with
$$\Xi=\delta_{(\psi/4,\psi/4,\psi/4,\psi/4,0,0,\dots)},$$ when
$r$ distinct groups of ancestral lineages merge. The number of
lineages in each group is given by $k_1,\dots,k_r$, given $\beta$
active ancestral lineages.  The number $s=\beta-(k_1+\cdots+k_r) \ge
0$ gives the number of lineages (ancestral chromosomes) unaffected by
the merger (cf.\ \cite{S00}, Thm.~2).  The particular form of $\Xi$
given above follows from the fraction $\psi$ of the population
replaced by the offspring of the two parents in a large reproduction
event, and our assumption that each parent contributes exactly one
chromosome to each offspring.   We have the following convergence result.

\begin{theorem}
\label{thm:conv}
Let $\{\xi^{n,N}(m), m \ge 0\}$ be the ancestral process of a sample of $n$ chromosomes
in a population of size $N$ and 
assume the scaling relations 
(\ref{eq:largereprodratescaling}, \ref{eq:recombratescaling}). 
Then, starting from $\xi^{n,N}(0) \in \mathscr{A}_n$, we have that
\[
\{\xi^{n,N}(\lfloor N^2 t \rfloor)\} \to \{\xi(t)\}, \quad \mbox{ as } \quad N \to \infty,
\]
in the sense of the finite-dimensional distributions on the interval
$(0,\infty)$. The initial value of the limiting process is given by
$$\xi(0)  =  \mathsf{cd}(\xi^{n,N}(0)) \in \mathscr{A}_n^{\texttt{sm}}.$$
\end{theorem}

A proof can be found in the Appendix.  If $c=0$, the classical
ancestral recombination graph for a diploid population with
recombination in the spirit of \cite{GM97} results.


\textbf{  General diploid Moran-type models: ``random'' $\psi$}

\newcommand{\PsiN}{\ensuremath{\sub{\Psi}{N}}}%

One of the aims of the present work is to understand the genome-wide
correlations in gene genealogies induced by sweepstake-style
reproduction.  So far, we have discussed this for a very simple
example of a sweepstake mechanism (analog to the one considered in
\cite{EW06}).  More precisely, the {\em fraction} $\psi \in (0,1)$ of
the population replaced by the offspring of a single pair of
individuals in a large offspring number event has hitherto been
assumed to be (approximately) constant.  Along the lines of the
previous discussion, an ancestral recombination graph with a {\em
  randomized} offspring distribution can be derived (a comprehensive
discussion of single-locus haploid Moran models in the domain of
attraction of $\Lambda$-coalescents can be found in a recent article
of \cite{HM12}).  Even though $\psi$ is now considered a random
variable, the population size stays constant at $N$ diploid
individuals.  Allowing $\psi$ to be random may be biologically more
realistic than taking $\psi$ to be a constant.  On the other hand, the
problem of identifying suitable classes of probability distributions
for $\psi$, reflecting the specific biology of  given natural
populations, is still open and an area of active research.

To explain the convergence arguments when $\psi$ is random, let the
random variable $\sub{\Psi}{N}$, taking values in $[N-2]$, denote the
random number of diploid offspring contributed by the single
reproducing pair of parents at each timestep; a new realisation of
$\sub{\Psi}{N}$ is drawn before each reproduction event. Again, we
consider the effect of such a reproduction mechanism on coalescence
events in a {\em sample}.  The probability that two given chromosomes
residing in two single-marked individuals in the sample coalesce in
the previous timestep given the value of $\PsiN$ is
\begin{equation}%
  \label{eq:pairwisecoalprob}%
  \begin{split}
  &  \mathbb{P}(\{\textrm{pair coalescence}\} | \PsiN = k ) \\
 =&  \tfrac{1}{4}\delta_{\{k = 1\}}\tfrac{4}{N(N - 1)}\\
 + &    \tfrac{1}{4}\delta_{\{ k > 1 \}}\left( \tfrac{4k}{N(N - 1)} + \tfrac{k(k - 1) }{N(N - 1) }  \right), \\
   \end{split}
  \end{equation}%
  where the first and second terms on the right-hand side describe the
  case where one parent and one offspring are drawn, the third term
  covers the case where two offspring are drawn, and the $1/4$
  accounts for the probability that the two chromosomes in question
  must descend from the same parental chromosome.  Define 
  \begin{align}%
  \label{eq:Ppaircoal}%
    \sub{c}{N} & := 4\mathbb{P}(\{  \textrm{pair coalescence}\}) \\
    & = \sum_{k=1}^{N - 2} 
    \mathbb{P}( \{  \textrm{pair coalescence}\} | \PsiN = k )  \mathbb{P}( \PsiN = k ) \notag \\
    & = \EE{\tfrac{\PsiN(\PsiN + 3)}{N(N - 1)} } 
  \end{align}%
  (the factor $4$ facilitates comparison with the haploid case). 
  The sequence of laws
  $\mathcal{L}(\sub{\Psi}{N})$, $N \in \N$, will be assumed to satisfy
  the following three conditions: 
  \begin{equation}
    \label{cond1:cnvanishes}
     \sub{c}{N} \to 0 \quad \text{as} \; N\to\infty, 
     \end{equation}%
  \begin{multline} 
    \label{eq:tscaleratiogen}
    \frac{\sub{c}{N}}{\E\left[\PsiN / N\right]} 
    = \frac{1/\E\left[\PsiN / N\right]}{1/\sub{c}{N}}
    = \frac{\E\left[ \PsiN(\PsiN + 3) \right]}{%
      (N-1)\E\left[ \PsiN \right]} \\  \to 0 
    \quad \text{as}\; N\to\infty, 
  \end{multline}
     and there exists a probability measure $F$ on $[0,1]$ such that
     \begin{equation}%
    \label{cond2:sagitovs}
     \frac{1}{\sub{c}{N}}\mathbb{P}(\PsiN > Nx) 
    \mathop{\longrightarrow}_{N\to\infty} \int_x^1 \tfrac{1}{y^2} F(dy) 
  \end{equation}%
  for all continuity points $x \in (0,1]$ of $F$.
  
  Condition
  \eqref{cond1:cnvanishes} is necessary for any limit process of the
  genealogies to be a continuous-time Markov chain, 
  condition \eqref{eq:tscaleratiogen} ensures that a separation 
  of time scales phenomenon occurs, 
  and
  \eqref{cond2:sagitovs} 
  fixes the limit dynamics of the large merging events 
  (it is analogous to \cite[necessary
  condition~(13)]{S99} in the haploid case). 
  In the proof of convergence to a limit process we will recall
  equivalent conditions to \eqref{cond2:sagitovs} (see Appendix,
  section ~\ref{subsect:genPsiproofs}).  
  Condition~\eqref{cond1:cnvanishes} implies (see  
  Section~\ref{subsect:genPsiproofs} in Appendix)
  \begin{equation}
    \label{eq:tscale1gen}
    \E\left[\PsiN / N\right] \to 0
    \quad \text{as}\; N\to\infty,
  \end{equation}
  i.e.\ the probability for 
  a given individual to be an offspring in a given reproduction 
  event becomes small. 
 Hence, \eqref{eq:tscale1gen} and \eqref{eq:tscaleratiogen} together
 show that there will be two diverging time-scales: The ``short''
 time-scale $1/\E\left[\PsiN / N\right]$ on which chromosomes paired
 in double-marked individuals disperse into single-marked individuals
 and the ``long'' time-scale $1/\sub{c}{N}$ over which we observe
 non-trivial ancestral coalescences.




  In order to obtain a
  non-trivial genealogical limit process, we will then speed up time
  by a factor of $4/\sub{c}{N}$, i.e., $4/\sub{c}{N}$ reproduction
  events correspond to one coalescent time unit (see
  Thm.~\ref{thm:conv.general} below).  This time rescaling is chosen
  in order for two chromosomes to coalesce at rate 1 in the limit.
  The required scaling relation for the recombination rates is now
  \begin{equation} 
    \label{eq:recombratescaling.general}
    \rN^{(\ell)} \sim \frac{\sub{c}{N}}{4\E\left[\PsiN / N\right]} r^{(\ell)}
    \quad \text{as}\; N\to\infty
  \end{equation}
  with $r^{(\ell)} \in [0,\infty)$ fixed for $\ell=1,\dots,L-1$ 
  (where $f(N) \sim g(N)$ means $\lim_{N\to\infty} f(N)/g(N)=1$). 
  An intuitive explanation for the requirement 
  \eqref{eq:recombratescaling.general} is that since the probability for 
  a given individual to be an offspring in a given reproduction 
  event is $\E\left[\PsiN / N\right]$, after 
  speeding up time by $4/\sub{c}{N}$, on any lineage recombination 
  events between locus $\ell$ and $\ell + 1$ occur as a Poisson 
  process with rate $r^{(\ell)}$. 
  \smallskip

  A simple sufficient condition for \eqref{eq:tscaleratiogen} is the
  following: For any $\varepsilon > 0$,
  \begin{equation}%
    \label{eq:simplecond}%
  N\prob\left(\Psi_N > \varepsilon N\right) \to 0 \quad \textrm{as} \quad N \to \infty.
  \end{equation}%
  Indeed, we have, by assuming  $N > \varepsilon N$,  
  \begin{displaymath}%
    \begin{split}%
    \E\left[\Psi_N^2\right] &  =  \sum_{k = 1}^{\lfloor \varepsilon N \rfloor} k^2 \P\left(\Psi_N = k\right)  +   \sum_{k = \lfloor \varepsilon N \rfloor + 1}^{N}  k^2 \P\left( \Psi_N = k \right) \\ 
    & \le \sum_{k = 1}^{\lfloor \varepsilon N \rfloor}k\varepsilon N \P\left( \Psi_N = k \right)  +  \sum_{k = \lfloor \varepsilon N \rfloor + 1}^N N^2\P\left( \Psi_N = k  \right) \\
    & \le \varepsilon N\E\left[\Psi_N \right]  +  N^2 \P\left( \Psi_N > \varepsilon N \right).
    \end{split}%
    \end{displaymath}%
Dividing by  $N\E\left[\Psi_N\right]$ gives
$$
\frac{ \E\left[ \Psi_N^2 \right] }{N\E\left[N\right]} \le \varepsilon  +  \frac{N\P\left(\Psi_N > \varepsilon N \right) }{\E\left[ \Psi_N \right] },
$$
and, since $\E\left[\Psi_N\right] > 1$, 
$$
\limsup_{N \to \infty} \frac{ \E\left[ \Psi_N^2 \right] }{N\E\left[N\right]}   < \varepsilon  +   \limsup_{N \to \infty}N\P\left(\Psi_N > \varepsilon N \right) = \varepsilon.
 $$
Thus, condition \eqref{eq:tscaleratiogen}  is obtained since we can choose $\varepsilon$ to be as small as we like.   
\smallskip 

  The limiting genealogical process will then be a continuous-time Markov chain 
  on $\mathscr{A}_n^{\tt sm}$ with generator matrix $G$ whose off-diagonal 
  elements are given by (for the values on the diagonal we again have 
  \eqref{eq:DefGdiag})
  \begin{align}
    \label{eq:DefG.general}
    G(\xi,\xi') = 
    \begin{cases} 
      C_{\beta;2}
      & \text{if} \; 
      \xi'= \mathsf{pairmerge}_{j_1,j_2}(\xi) \\
      r^{(\ell)} & \text{if} \; \xi'= \mathsf{recomb}_{j,\ell}(\xi) \\
      C_{\beta;|J| } 
      & \text{if} \; \xi'= \mathsf{groupmerge}_{J_1,J_2,J_3,J_4}(\xi) \\
      0 & \text{for all other $\xi'\neq\xi$}
    \end{cases}
  \end{align}
  where $$C_{\beta;|J|} :=
  C_{\beta;|J_1|,|J_2|,|J_3|,|J_4|;\beta-(|J_1|+|J_2|+|J_3|+|J_4|)},$$
  $k = (k_1, \ldots, k_r)$, $|k| = k_1 + \cdots + k_r$,  and
  \begin{align} 
    \label{eq:Xirate.general}
    C_{b;k;s} & = 
     4 \sum_{l=0}^{s \wedge (4-r)} {s \choose l} \tfrac{(4)_{r+l}}{4^{|k| + l}}\notag \\ 
   & \cdot \int_{[0,1]} x^{|k|+l} (1 - x)^{s - l} \frac{1}{x^2} F(dx) \notag \\
   & =   F(\{0\}) \delta_{\{r=1, k_1=2\}} \notag  \\
   & +  4 \sum_{l=0}^{s \wedge (4-r)} {s \choose l}  \tfrac{(4)_{r+l}}{4^{|k| + l}}\notag \\
   & \cdot  \int_{(0,1]} x^{|k|+l} (1 - x)^{s - l} \frac{1}{x^2} F(dx) 
  \end{align}
  with $F$ from \eqref{cond2:sagitovs}.
  As in the case of constant $\psi$, the third line in  \eqref{eq:DefG.general} gives the  transition rates for a given merger into $r$ $(\leq 4)$
  groups of sizes $k_1,\dots,k_r$ when $\beta$ active ancestral
  lineages are present, with $s=\beta-|k| \ge 0$ lineages
  unaffected by a given merger of the $\Xi$-coalescent with
  $$\Xi=\int_{[0,1]} \delta_{(x/4,x/4,x/4,x/4,0,0,\dots)} \,F(dx),$$ (cf.\ \cite{S00}, Thm.~2).    By way of example,    $C_{2;2;0} = 1$.  
  Now we can state the convergence of our ancestral recombination
  graph process with random $\psi$. The analogue of
  Theorem~\ref{thm:conv} is the following:
  \begin{theorem}
    \label{thm:conv.general}
    Let $\{\xi^{n,N}(m), m \ge 0\}$ be the ancestral process of a sample of $n$ chromosomes
    in a population of size $N$ with offspring laws $\mathcal{L}(\sub{\Psi}{N})$ 
    which satisfy \eqref{cond1:cnvanishes}, \eqref{eq:tscaleratiogen} 
    and \eqref{cond2:sagitovs}, and 
    assume the scaling relation \eqref{eq:recombratescaling.general} for the 
    recombination rates. 
    Then, starting from $\xi^{n,N}(0) \in \mathscr{A}_n$, we have that
    \[
    \{\xi^{n,N}(\lfloor 4t/\sub{c}{N} \rfloor)\} \to \{\xi(t)\}, \quad \mbox{ as } \quad N \to \infty,
    \]
    in the sense of the finite-dimensional distributions on the interval $(0, \infty)$. The process
    $\{\xi(t)\}$ is the Markov chain with generator matrix 
    \eqref{eq:DefG.general} and initial value $\xi(0)$ given by
    $$\xi(0)  =  \mathsf{cd}(\xi^{n,N}(0)) \in \mathscr{A}_n^{\texttt{sm}} .$$ 
  \end{theorem}
  The proof is given in  Section~\ref{subsect:genPsiproofs} in Appendix.

    While $\sub{c}{N} \ge 1/N^2$ by definition, in principle any decay 
    behaviour of $\sub{c}{N}$ that is consistent with 
    $\liminf_{N\to\infty} N^2 \sub{c}{N} \ge 1$, and hence any therefrom 
    derived scaling relation between coalescent time scale and 
    model census population size, is possible via a suitable 
    choice of the family $\mathcal{L}(\sub{\Psi}{N})$, $N \in \N$.

    For an extreme example, let $\sub{\Psi}{N} \equiv \lfloor N^\gamma \rfloor$
    for some $\gamma \in (0,1)$, then 
    $\sub{c}{N} \sim N^{-2(1-\gamma)}$ and \eqref{cond2:sagitovs} 
    is satisfied with $F=\delta_0$. 
    

  The relation with the ``fixed $\psi$'' model 
  is as follows: For Theorem~\ref{thm:conv},
  we used the simple mixture distribution for $\PsiN$: 
  \begin{equation}
    \label{eq:PsiNsimplemixture}
    \P(\PsiN = \lfloor \psi N \rfloor ) = 1 -  \P(\PsiN = 1) = \frac{c}{N^2}
  \end{equation}
  for $\PsiN$, in which $\psi \in (0,1)$ and $c > 0$ are both constants. 
  Our choice \eqref{eq:PsiNsimplemixture} of law for $\PsiN$ gives, using \eqref{eq:pairwisecoalprob}, 
  \begin{displaymath}%
  \begin{split}%
    \sub{c}{N} & = \E\left[ \tfrac{\PsiN(\PsiN + 3)}{N(N - 1)} \right] \\ 
& = \big(1-\tfrac{c}{N^2}\big) \frac{4}{N(N-1)} +  \tfrac{c}{N^2} \frac{\psi N (\psi N + 3) }{N(N-1)}  \\
  & \sim \frac{1}{N^2} (4+c\psi^2).
  \end{split}%
  \end{displaymath}
  Define $ 1_{(0,\psi)}(x) = 1$ if $x \in (0,\psi)$, and $
  1_{(0,\psi)}(x) = 0$ otherwise.  Our choice
  \eqref{eq:PsiNsimplemixture} further gives
  \begin{displaymath}%
    \begin{split}
  \P(\PsiN > Nx ) & = 1_{(0,\psi)}(x) \P(\PsiN > Nx )  \\
  & = 1_{(0,\psi)}(x) cN^{-2},
  \end{split}%
  \end{displaymath}%
  and therefore
  \begin{displaymath}%
    \begin{split}%
  & \frac1{\sub{c}{N}} \P(\PsiN > \lfloor Nx \rfloor ) \longrightarrow 
   1_{(0,\psi)}(x) \frac{c}{4 + c \psi^2}\\ 
   & = \int_{(x,1]} y^{-2} \, F(dy)
  \end{split}%
  \end{displaymath}
  with $$F = \tfrac{4}{4 + c \psi^2} \delta_0 + \tfrac{c \psi^2}{4 + c \psi^2} \delta_\psi.$$ 
  Furthermore, $\E\left[\PsiN/N\right] = 1/N+O(1/N^2)$, thus 
  \[
  \frac{\sub{c}{N}}{4 \E\left[\PsiN / N\right]} \sim \frac1N \frac{4 + c\psi^2}{4}
  \]
  and Theorem~\ref{thm:conv} follows from Theorem~\ref{thm:conv.general} 
  (after rescaling time in the limit process $\{\xi(t)\}$ by a factor of 
  $(4+c\psi^2)/4$). 


  The constant $C_{b;k} := C_{b;k_1, \ldots, k_r ; s}$
  (\ref{eq:Xirate.general}) depends on the probability measure $F$.
  The form of $F$ will no doubt be different for different
  populations.  We reiterate that resolving the mechanism of
  sweepstake-style reproduction will require detailed knowledge of the
  reproductive behaviour and the ecology 
  of the organism in question, along with
  comparison of model predictions to multi-loci genetic data.  A
  candidate for $F$ may be the beta distribution with
  parameters $\vartheta > 0$ and $\gamma > 0$, in which case the
  constant $C_{b;k}$ in  \eqref{eq:DefG.general} takes the form $(|k| := k_1 + \cdots + k_r)$
  \begin{multline}%
  C_{b;k}  = 4\sum_{\ell}\binom{s}{\ell}(4)_{r + \ell}\left(\tfrac{1}{4}\right)^{|k| + \ell}\\\cdot \frac{B(|k| + \ell + \vartheta - 2, s +    \gamma - \ell)}{B(\vartheta, \gamma)},
  \end{multline}%
$B(\cdot,\cdot)$ being the Beta function.

\textbf{Different scaling regimes}

The mechanism of sweepstake-style reproduction may be different for
different populations, and the frequency of large offspring number
events may also be different.  The particular timescale of the large
reproduction events (we chose $\sub{\varepsilon}{N} = c/N^2)$ results
in a separation of timescales of the limit process.  Resolving the
separation of timescales problem results in the ARG with generator
(\ref{eq:DefG}).  Different scalings of $\sub{\varepsilon}{N}$ result
in different limit processes.  By way of example, if $N^2\eN \to 0$,
large offspring number events are negligible in a large population,
and we obtain the ARG associated with the usual Wright-Fisher
reproduction, which can be read off Equation~(\ref{eq:DefG}) by taking
$c = 0$.  One other scaling regime may seem reasonable, namely 
taking large offspring number events to be more frequent than 
in Assumption~\eqref{eq:largereprodratescaling}, but not too frequent.  In mathematical notation,
$N^2\eN \to \infty$ and $N\eN \to 0$.  The ancestral process in this
regime is again characterised by instantaneous separation of marked
chromosomes into single-marked individuals, followed by coalescence
and recombination occurring on the slow timescale.  The probability of
recombination is proportional to $N\eN$ since the slow timescale must
be in units proportional to $1/\eN$.  Hence, small reproduction events
become negligible in the limit, and the generator of the limit process
is given by
\begin{align}
\label{eq:DefGb}
G(\xi,\xi') = 
\begin{cases} 
 \frac{\psi^2}{4} C_{\beta;2;\beta-2}
& \text{if} \; 
\xi'= \mathsf{pairmerge}_{j_1,j_2}(\xi) \\
r^{(\ell)}/r & \text{if} \; \xi'= \mathsf{recomb}_{j,\alpha}(\xi) \\
 \frac{\psi^2}{4} C_{\beta;|J|} 
& \text{if} \; \xi'= \mathsf{groupmerge}_{J}(\xi) \\
0 & \text{for all other $\xi'\neq\xi$}
\end{cases}
\end{align}%
in which $C_{\cdot;\cdot;\cdot}$ is given by
Equation~(\ref{eq:Xirate}).  The requirement $N\eN \to 0$ is needed to
prevent unreasonably high rate of recombination.

\textbf{Haploid analogs}

A haploid version of the above model, where only one parent 
contributes offspring at each timestep, is a specific example of a $\Lambda$-coalescent,  where
$$
\Lambda(dx)=\delta_0(dx) +  c\psi^2\delta_\psi(dx), \quad \psi \in (0,1), \quad c \in [0,\infty), 
$$ see e.g.\ \cite{EW06} and \cite{BB09}.  More precisely, as the population
  size $N$ tends to infinity, assume probability $1 - c/N^2$ for the
  small reproduction events, $c/N^2$ for the large reproduction events
  (i.e., choose $\varepsilon_N= c/N^2$), and speed up generation time
  by $N^2$.
Again, by randomising $\psi$ and/or switching
to different scaling regimes, it is possible to obtain any given
$\Lambda$-coalescent as limiting genealogy. 

\textbf{Two-sex extensions}

Recent studies of the spawning behaviour of Atlantic cod
  indicate that cod adopts a lekking behaviour, in which males compete
  for females, and  females exercise mate choice
  \citep{nordeide00}.  Direct microsatellite DNA analysis indicates
  that although multiple paternity is sometimes detected, the
  reproductive success is highly skewed among the males, i.e.\ most of
  the successfully fertilized eggs can be attributed to a single male
  \citep{hutchings99}.  Our model thus seems a good approximation to
  the actual reproduction mechanism of cod.  Modifications to allow
  two distinct genders, and multiple paternity, 
  are in principle 
  straightforward.

\textbf{More general recombination models}

Our model can easily be enriched to allow also more general recombination 
events involving more than one crossover point at a time. 
Furthermore, by letting the number $L$ of loci tend to infinity, 
a continuous model, where $[0,1]$ represents a whole chromosome 
(as in \cite{GM97}),  
can be accomodated into our framework.

\textbf{Correlations in coalescence times}

\textbf{\textsl{The marginal process}}

Every marginal process (marginal with respect to one fixed locus under
consideration) of our ancestral recombination graph is a
$\Xi$-coalescent (see \cite{S00} for notation and details) with
$$
\Xi=\delta_0+c \frac{\psi^2}{4}\delta_{\big(\tfrac\psi4,\tfrac\psi4,\tfrac\psi4,\tfrac\psi4,0,0,\dots\big)}
.$$ For $r=0$, all marginals are identical (realization wise), in particular
times to the most recent common ancestor for different loci have
correlation $1$.  However, in contrast to the classical setting, for
$r \to \infty$ one expects that the loci will not completely
decorrelate, but instead keep positive correlations, as pointed out to
us by J.E.\ Taylor (personal communication). In particular, one will not
obtain the product distribution. This observation is a potential
starting point for designing tests for the presence of large
reproduction events, by comparing correlations for loci at large
distance (hence with high recombination rate) under a Kingman- and a
$\Xi$-coalescent based ARG.

\textbf{\textsl{Correlation in coalescence times at two loci}}

 Correlations in coalescence times between two loci have been
considered in the context of quantifying association between loci
\citep{M02}.  \cite{EW08} consider correlations in coalescence times
for a haploid population model admitting large offspring numbers, in
which the ancestral process only admits asynchronous multiple mergers
of ancestral lineages.  To illustrate the effects of the reproduction
parameters on the coalescence times, we also consider the probability
that coalescence occurs at the same time at the two loci, as well as
the expected time until coalescence. 

The calculations to obtain the correlations for a sample of size two
at two loci (following the approach and notation of \cite{D02}) are
shown in the Appendix, Section \ref{sn:correlation}.  As we are now
considering the gene genealogy of unlabelled lineages, let us briefly
state the sample space.  Let $a$ and $b$ denote the types at loci $a$
and $b$, respectively.  The three sample states before coalescence at
either locus has occurred can be denoted as $(ab)(ab)$, $(ab)(a)(b)$,
and $(a)(a)(b)(b)$.  By $(ab)(ab)$ we denote the state of two
chromosomes each carrying ancestral material at both loci.
By $(ab)(a)(b)$ we denote the state of one
$(ab)$ chromosome in addition to two chromosomes $(a)$ and $(b)$ carrying ancestral types at locus 1 and 2 only, resp.
The notation $(a)(a)(b)(b)$
denotes the state of four chromosomes each carrying ancestral types at only one locus.
Let $$h(i) := \mathbb{P}\left( \{T_a =
T_b\}|i \right), \quad i \in \{0,1,2\}$$ denote the probability that
coalescence at the two loci occurs at the same time, given that the
process starts in state $i$, in which $i$ refers to the number of
double-marked chromosomes (2, 1, or 0).  As we are working with the
limiting model, all marked individuals are effectively single-marked.
Under the usual (Kingman coalescent-based) ARG, $\lim_{r \to
  \infty}h(i) = 0$ as one would expect.  Our model yields
\begin{equation}%
  \label{eq:hilimrinf}%
\lim_{r \to \infty}h(i) =  \tfrac{c\psi^4}{32 + 8c\psi^2  -  c\psi^4}, \quad i \in \{0,1,2\};
\end{equation}%
indicating that even unlinked loci remain correlated due to
sweepstake-style reproduction. Figure~\ref{figure1} shows graphs of $h(i)$
as a function of $\psi$ for different values of $c$ and $r$.  As
expected, $h(i)$ increases with $\psi$, at a rate which increases with
$c$.

Under the usual ARG, the expected time $\mathbb{E}_i[T_s]$ until
coalescence at either loci, starting from state $i$ is given by
$\mathbb{E}_i[T_s] = (1 + h(i))/2$.   The random variable $T_s$ can be
viewed as the minimum of the time until coalescence occurs at the two
loci.  As $r \to \infty$, the times $T_1$ and $T_2$ until
coalescence at the two loci, resp., become independent and identically
distributed  exponentials ({i.i.d.e.})  with rate $1$, whose minimum
has expected value $1/2$.  Under our model, the mean of $T_s$ is
\emph{not} the minimum of two i.i.d.e.\ with
rate $1 + c\psi^2/4$, another reflection of the correlation in gene
genealogies induced by sweepstake-style reproduction. 
Indeed, our model gives
$$
\lim_{r \to \infty} \mathbb{E}_i[T_s] = \tfrac{1}{2}\left( \tfrac{1}{1 + \chi c\psi^2/4}  \right), \quad i \in \{0,1,2 \}.  
$$
in which $\chi = 1 - \psi^2/8$.

Under our model, $\mathbb{E}_i[T_s]$ decreases with $\psi$, and the
rate of decrease increases with $c$ (Figure~\ref{figure2}).  The same
pattern holds for the expected time $\mathbb{E}_i[T_l]$ until
coalescence has occurred at both loci (Figure~\ref{figure3}).  As $r
\to \infty$, $\mathbb{E}_i[T_l]$ associated with the usual ARG
approaches the expected value $(3/2)$ of the maximum of two i.i.d.e.\
with rate 1.  Under our model,
\begin{multline*}
\lim_{r \to \infty}\mathbb{E}_i[T_l] = \frac{3}{2}\frac{1}{1 + \tfrac{c\psi^2}{4} }\frac{1}{1 + \tfrac{c\psi^2}{4} - \tfrac{c\psi^4}{32} } \\ +  \frac{  c\psi^2(6 - \psi^2)  }{ (c\psi^2 + 4)(4 + c\psi^2 - c\psi^4/8)  }
\end{multline*}
while the maximum of two i.i.d.e.\ with rate $\lambda$ has expected value  $3/(2\lambda)$.

The correlation $\crr{T_1}{T_2}$ between $T_1$ and $T_2$ when starting
from one of the three possible sample states $i \in \{0,1,2 \}$ (see
Appendix) increases with $\psi$, and more so if $c$ is large
(Figure~\ref{figure4}).  One obtains the following limit relations
between $h(i)$ and $\crr{T_1}{T_2}$ for $i \in \{0,1,2\}$:
    \begin{displaymath}%
      \begin{split}%
        \lim_{r \to \infty}\crr{T_1}{T_2} & = \lim_{r \to \infty}h(i), \quad \textrm{(see Eq. \eqref{eq:hilimrinf})} ; \\
        \lim_{r \to 0}\crr{T_1}{T_2} & = \lim_{r \to 0}h(i), \quad \textrm{(see Eq.\ \eqref{eq:hir0})};\\
         \lim_{c \to \infty}\crr{T_1}{T_2} & = \lim_{c \to \infty}h(i), \quad \textrm{(see Eq.\ \eqref{eq:limhi})}. \\
        \end{split}%
      \end{displaymath}%

Quantifying the association between alleles at different loci can give insight into the evolutionary history of 
populations.    Let  $f_{a}$ and $f_b$ denote the frequencies of  alleles $a$ at locus 1, and $b$ at 
locus 2, and let $f_{ab}$ denote the  frequency of chromosome $ab$ 
in the total population.   The statistic $D_{ab} := f_{ab} - f_af_b$ 
measures the deviation from independence, since if the two loci were evolving independently,  
$f_{ab} = f_af_b$.   A related quantity is the  $r^2$ statistic, defined as 
	$$r^2 :=  \frac{D^2}{f_a(1 - f_a)f_b(1 - f_b)}$$ \citep{hill68},  assuming  $f_a, f_b \notin \{0,1\}$.    
  In applications, one would like to  compare  observed values of  $r^2$ calculated from data to
  the expected value $\mathbb{E}\left[r^2\right]$, obtained under an appropriate population model.  
  Calculating the expected value of $r^2$ is not straightforward,  since $r^2$ is a ratio of correlated
  random variables.    The expected value of $r^2$ is, instead, approximated by the ratio 
  $\mathfrak{D} =  \mathbb{E}[D^2]/\mathbb{E}[f_a(1 - f_a)f_b(1 - f_b)]$ \citep{OK71}.

      A prediction $\mathfrak{D}$ of linkage disequilibrium in the
      population can be framed in terms of correlations in coalescence
      times between two loci for a sample of size two, assuming a
      small mutation rate \citep{M02}.  The prediction rests on
      approximating the expected value $\EE{r^2}$ of the squared
      correlation statistic $r^2$ \citep{HR68} of association between
      alleles at two loci by the ratio of expected values
      \citep{OK71}. Following e.g.\ \cite{D02} one can obtain
      expressions for correlations in coalescence times between two
      loci for a sample of size two (see Appendix).  Under our
      model, one obtains the limit results
\begin{displaymath}
  \begin{split}%
\lim_{r \to \infty }\mathfrak{D} & = 0, \\
\lim_{c \to \infty}\mathfrak{D} & = \tfrac{\psi^3 - 16\psi^2 + 56\psi - 80 }{\psi^3 - 10\psi^2 + 88\psi - 176}. 
\end{split}%
\end{displaymath}
When $\psi$ is small but $c$ large, one obtains   
$$
\mathfrak{D} = \tfrac{5 - 7\psi/2 }{11 - 11\psi/2 } + O(\psi^2).
$$
Under the usual ARG, $\lim_{r \to 0}\mathfrak{D} = 5/11$.  Thus, even
in the presence of a \emph{high} recombination rate, if large offspring
number events are frequent enough, one may only see evidence of
\emph{low} recombination rate in data.  Further, the prediction
$\mathfrak{D}$ can be substantially higher than Kingman-coalescent based
predictions if $c$ is large, and the recombination rate is not too small
(Figure~\ref{figure7}).

For particular examples of probability measures $F$ 
from Equation~\eqref{eq:Xirate.general} associated with the generator derived from our random
offspring distribution model
one can compute the quantities considered above in
relation to fixed $\psi$.  One such example distribution can be the
Beta$(\vartheta,\gamma)$ distribution.    One obtains for $i \in \{0,1,2\}$,
$$
\lim_{r \to \infty}h(i) =  \frac{4\gamma(1 + 2\vartheta + \gamma) }{8\gamma(1 + \gamma) + 10\gamma\vartheta  + 7\vartheta(1  +          \vartheta)}.
$$
Define $\tilde{h}(i) := \lim_{r \to \infty}h(i)$.  For $i \in \{0,1,2\}$ one obtains
\begin{equation}%
  \label{eq:limrETsTlbeta}%
  \begin{split}%
\lim_{r \to \infty}\mathbb{E}_i[T_s] & =  4\tilde{h}(i) +  \frac{4\gamma(1 + 2\vartheta + \gamma)}{8\gamma(1 + \gamma) + 10\gamma\vartheta  + 7\vartheta(1  +          \vartheta) }, \\
\lim_{r \to \infty}\mathbb{E}_i[T_l] & = \frac{3}{2} - \frac{1}{2}\tilde{h}(i) + \frac{3\gamma}{2(8\gamma(1 + \gamma) + 10\gamma\vartheta  + 7\vartheta(1  + \vartheta))}. \\
\end{split}%
\end{equation}%
The form of the relation shown in \eqref{eq:limrETsTlbeta} between
$h(i)$ and $\mathbb{E}_i[T_s]$ and $\mathbb{E}_i[T_l]$ resembles the
one obtained for the Kingman coalescent-based ARG, with the addition
of a `correction' term due to simultaneous multiple mergers.

\textsl{\textbf{  Variance of pairwise differences}} 

The expected variance of pairwise differences was employed by
\cite{W97} to estimate the recombination rate in low offspring number
(Wright-Fisher) populations, under the usual ancestral recombination
graph.     Let the random variable  $K_{ij}$ denote the number of differences  between 
sequences $i$ and $j$, with $K_{ii} = 0$.    The average number $\pi$ of pairwise differences 
for $n$ sequences is 
$$
	\pi = \frac{2}{n(n - 1)}\sum_{i < j}K_{ij}.
$$
The (empirical) variance $S_\pi^2$ of pairwise differences is defined as
$$
	S_\pi^2 =  \frac{2}{n(n - 1)}\sum_{i < j}\left(K_{ij} - \pi\right)^2.
$$
   In the Appendix we derive the expected variance of pairwise
differences $\EE{S_\pi^2}$ under the ancestral recombination graph
described by the generator $G$ \eqref{eq:DefG} derived from our large
offspring number model.  Under our model, $\EE{S_\pi^2}$ is a function
of the parameters $c$ and $\psi$, in addition to being a function of
$r$ and $\theta$ (Figures~\ref{fig:pairvarrho0} and
\ref{fig:pairwnn0}).  In Figure~\ref{fig:pairvarrho0}, $\EE{S_\pi^2}$,
when only two loci are considered, is graphed as a function of the
recombination rate, and in Figure~\ref{fig:pairwnn0} as a function of
sample size.  Figures~\ref{fig:pairvarrho0} and \ref{fig:pairwnn0}
show that $\EE{S_\pi^2}$ is primarily influenced by the mutation rate
$(\theta)$, when the values of $c$ and $\psi$ are fairly modest.
However, $\EE{S_\pi^2}$ can be quite low when both $c$ and $\psi$ are
large, even when $\theta$ is also large (Figure~\ref{fig:pairwnn0}).
When $c$ and $\psi$ are both large, two sequences are more likely to
coalesce before a mutation separates them.

  The variance of pairwise differences alone will not suffice to yield
  estimates of $r$ if both $c$ and $\psi$ are unknown.  To jointly
  estimate the four parameters ($c$, $\psi$, $r$, $\theta$) of our
  model one probably needs to employ computationally-heavy likelihood
  and importance sampling methods in the spirit of \cite{FD01}.
  However, given knowledge of $c$ and $\psi$, one can, in principle,
  use the variance of pairwise differences to quickly obtain estimates
  of the recombination rate.

\textsl{\textbf{ Correlations in ratios of coalescence times }}

  The behaviour of the correlations in ratios of coalescence times for
  sample sizes  larger than two is investigated using Monte Carlo simulations.

  Let $L_i$ denote the total length of branches ancestral to $i$
  sequences at one locus, let $L$ denote the total length of the
  genealogy at the same locus, and define $R_i := L_i/L$. Thus, $R_1$ is
  the total length of external branches to the total size of the
  genealogy.  The idea behind estimating the expected value $\EE{R_i}$
  is as follows. Assuming the infinitely many sites mutation model,
  let $S_i$ denote the total number of mutations in $i$ copies, $S$
  the total number of segregating sites, and define $V_i := S_i/S$.
  The key idea behind deriving the coalescent was to separate the
  (neutral) mutation process from the genealogical process.  The same
  principle also applies to predicting patterns of genetic variation
  using the coalescent: first one constructs the genealogy, and then
  superimposes mutations on the genealogy.  The shape of the
  genealogy is thus a deciding factor in the genetic patterns one
  predicts.  The relative lengths $R_i$ of the different types
  of branches should therefore predict the relative number $V_i$ of mutations of
  each class.  This idea is exploited by \cite{E11} to estimate
  coalescence parameters in the large offspring number models
  introduced by \cite{schweinsberg03} and \cite{EW06}.  Namely, the claim is
  \begin{equation}%
      \label{eq:limEriVialpha}%
  \lim_{n \to \infty}\EE{R_i} = \lim_{n \to \infty}\EE{V_i} = f(\varpi, i)
  \end{equation}%
  where $n$ denotes the sample size, $\varpi$ denotes the coalescence
  (reproduction) parameters.  Indeed, 
  it follows from the results of   \cite{BBS07,BBS08}, that  ($1 < \alpha < 2$)
  \begin{displaymath}
  \lim_{n \to \infty}\EE{R_i} = \lim_{n \to \infty}\EE{V_i} = \tfrac{\Gamma(i + \alpha - 2)(\alpha - 1)(2 - \alpha)}{\Gamma(\alpha)i!}
  \end{displaymath}
  when associated with the Beta$(2 - \alpha, \alpha)$ coalescent
  derived by \cite{schweinsberg03} from a population model in which the offspring
  law is stable with index $\alpha$.  A key feature of expression
  \eqref{eq:limEriVialpha} is the absence of mutation rate in the
  function $f(\varpi, i)$; thus given large number of DNA sequences
  (possibly in the thousands), one hopes to be able to obtain
  estimates of the coalescence parameters $\varpi$ without having  to jointly
   estimate the mutation rate.  In our model, there are four
  parameters to estimate, namely mutation and recombination rates,
  along with the coalescence parameters $c$ and $\psi$.  Even though
  full likelihood methods exist \citep{BB08,BBS11}, applying them to
  large datasets consisting of thousands of sequences may represent a
  challenge.

  Estimates of $\EE{R_i}$ as functions of the sample size $n$, and the
  coalescence parameters $c$ and $\psi$ are shown in
  Table~\ref{tab:ERi}.  In nearly all cases the estimates
  $\overline{R}_i$ decreased as sample size increased; the exception
  was $\overline{R}_1$ when $(c,\psi) = (1000, 0.5)$
  (Table~\ref{tab:ERi}).  When both $c$ and $\psi$ are large enough,
  we observe a non-monotonic behaviour in $\overline{R}_1$ as sample
  size increases (results not shown).  The non-monotonic behaviour may
  be related to the property of the marginal haploid process (the
  point-mass part obtained as $c \to \infty$) of a single locus of not
  coming down from infinity \citep{schweinsberg00inf}, i.e.\ when one
  starts with an infinite number of lineages (sample size), the number
  of lineages stays infinite.  For such processes that don't come
  down from infinity, the ratio $R_1$ should go to one, i.e.\ the gene
  genealogy should become completely star-shaped (see
  e.g.\ \cite{E11}).   As both $c$ and $\psi$ increase, one expects the
  deviation from Kingman-coalescent based predictions to increase.  By
  way of example, for sample size $50$ the vector $(\EE{R_1}, \ldots,
  \EE{R_4})$ is estimated to be approx.\ $(0.24,0.12,0.08,0.06)$ when
  associated with the Kingman coalescent $(c = 0)$, while being
  approx.\ $(0.58,0.20,0.09,0.05)$ when $(c,\psi) = (1000, 0.5)$.  In
  all cases the estimate $\widehat{R}_i$ of the standard deviation of
  $R_i$ decreases as sample size increases, indicating convergence.

  The rationale behind comparing the statistics in
  Tables~(\ref{tab:corXYa}--\ref{tab:corXYb}) is as follows.  As
  sequencing technologies advance, and the genomic sequences of more
  organisms become available, a case in point being the recently
  published genomic sequence of Atlantic cod \citep{S11}, genomic scans
  of thousands of individuals will become more common.  Given DNA
  sequence data for many loci, one could calculate correlations for
  counts and ratios of counts of mutations, and compare them to
  predictions based on different ancestral recombination graphs.
  Similarly for the single-locus statistics (Table~\ref{tab:ERi}), the
  idea is that the correlations of the coalescence time statistics
  ($L_i$ and $R_i$) should reflect correlations of mutation counts
  $(S_i)$.  In particular, under the usual ARG one expects (see
  Tables~\ref{tab:corXYa}--\ref{tab:corXYb})
  $$
  \lim_{r \to \infty}\text{cor}\left(L^{(1)}_i,L^{(2)}_j\right) = \lim_{r \to \infty}\text{cor}\left(R^{(1)}_i, R^{(2)}_j \right) = 0,
  $$
  where the superscript refers to locus number one and two, respectively, 
  while under an ARG admitting simultaneous multiple mergers one expects
   \begin{displaymath}
     \begin{split}%
  \lim_{r \to \infty}\text{cor}\left(L^{(1)}_i,L^{(2)}_j\right) & =  f(i,j,\varpi)\\
   \lim_{r \to \infty}\text{cor}\left(R^{(1)}_i, R^{(2)}_j \right) & = g(i,j,\varpi)
  \end{split}\end{displaymath}
  where $f$ and $g$ are 
  functions of the particular
  statistics indicated by $i$ and $j$ as well as the vector $\varpi$ of
  coalescence (reproduction) parameters.

  In general, the results reported in
  Tables~\ref{tab:corXYa}--\ref{tab:corXYb} indicate that high values
  of \emph{both} $\psi$ and $c$ are required for high correlations
  when recombination rate is high, when associated with our model.  In
  particular, the correlations between 
  $R^{(1)}_i$ and $R^{(2)}_i$ 
  (i.e.\ between corresponding $R_i$'s at different loci) 
  can be quite high,
  even when recombination is high, when both $c$ and $\psi$ are large
  enough; another indicator of the genome-wide correlations induced by
  sweepstake-like reproduction.

  A different question concerns the limit behaviour as sample size $n$
  increases. Fix the recombination rate and consider the limits
  \begin{equation}%
    \label{eq:limRiRjntoinf}
  \lim_{n \to \infty }\text{cor}\left(R^{(1)}_i, R^{(2)}_j \right), \quad
  \lim_{n \to \infty} \text{cor}\left(V^{(1)}_i, V^{(2)}_j \right) 
  \end{equation}%
  Under the usual ARG, one expects the limits in
  \eqref{eq:limRiRjntoinf} to be only functions of the recombination
  rate (and $i$ and $j$). If the ARG also admits simultaneous multiple
  mergers, one expects the limits in \eqref{eq:limRiRjntoinf} also to
  be functions of $\varpi$.  Considering unlinked loci, one would be
  interested in the limits
  \begin{equation}%
     \label{eq:limRiRjrntoinf}
    \lim_{r \to \infty}\lim_{n \to \infty }\text{cor}\left(R^{(1)}_i, R^{(2)}_j \right), \quad
  \lim_{r \to \infty}\lim_{n \to \infty} \text{cor}\left(V^{(1)}_i, V^{(2)}_j \right) 
    \end{equation}%
    Resolving the limits \eqref{eq:limRiRjrntoinf} for different ARG's
    promises not only to yield insights into genome-wide correlations,
    but also to provide tools for inference; e.g.\ to distinguish
    between different population models.

    The C program written to perform the simulations was checked by comparing
    correlation in coalescence times for sample size two at two loci
    to analytical results.  The program is available upon request.

    \textbf{Comparison with} \cite{EW08}   \\
      \cite{EW08} consider correlations in coalescence times, and the prediction
      $\mathfrak{D}$ of linkage disequilibrium, under a modified
      Wright-Fisher sweepstake-style reproduction model, and observe
      correlations in coalescence times between loci despite high
      recombination rate.  Our work differs from theirs in important
      ways.  To begin with, we treat diploidy in detail, in which each
      offspring receives its two chromosomes from two distinct
      diploid parents.  This leads to a separation of timescales of
      the ancestral process.  We formally derive an
      ancestral recombination graph which admits simultaneous multiple
      mergers of ancestral lineages, which naturally arise in diploid
      models.  
      Eldon and Wakeley observed correlations in coalescence times when
      considering only sample size two at each locus in a model that 
      contains diploid individuals only implicitly, it is not
      a priori obvious that the correlations would still hold for large sample
      sizes. 
      We confirm this using our formally obtained ARG, that  
      allows us also to investigate correlations in
      coalescence times, and in ratios of coalescence times, for
      sample sizes larger than two at each locus.   In addition,
      one can  apply our ARG to inference problems. Indeed, we show
      how the variance of pairwise differences can, in principle, be
      used to obtain estimates of the recombination rate.  Finally, we
      obtain a large class of ARGs by randomizing the offspring
      distribution; thus one is not restricted to the simple case of
      fixed $\psi$.

      Furthermore, since the estimate $\mathfrak{D}$ of the expected
      value of $r^2$ can be expressed in terms of correlations in
      coalescence times, Eldon and Wakeley consider $\mathfrak{D}$
      under their modified Wright-Fisher model.  However,
      $\mathfrak{D}$ is based on approximating an expected value of a
      ratio of correlated random variables by the ratio of expected
      values of the corresponding random variables, and is also
      derived for a sample of size two at two loci. Thus,
      $\mathfrak{D}$ may not be the ideal quantity to quantify
      association between loci for large sample sizes.  A more natural
      way may be to investigate correlations in coalescence times for
      samples larger than two the way we do.


\textbf{Discussion}

Understanding the genome-wide effects of sweepstake-like reproduction on
gene genealogies was our main aim.  To this end, we derived ancestral
recombination graphs for many loci arising from population models
admitting large offspring numbers.  High variance in individual
reproductive success, or sweepstake-style reproduction, has been suggested
to explain the low genetic diversity observed in many marine
populations \citep{H82,H94,A88,PW90,B94,A04}.  \cite{HP11} review the
sweepstake-style reproduction hypothesis, and conclude that it provides the
correct framework in which to investigate many natural marine
populations.

Multiple \citep{DK99,P99,S99} and simultaneous \citep{S00, MS01}
multiple merger coalescent models arise from population models
incorporating sweepstakes reproduction by admitting large offspring
numbers \citep{S03,EW06,SW08}.  While multiple merger coalescent
processes describing the ancestral relations of alleles at a single
locus have received the most attention from mathematicians, ancestral
processes for multiple linked loci have hitherto remained unexplored.
We derive an ancestral recombination graph for many loci from a
diploid biparental population model, in which 
one pair of
diploid individuals (parents) contribute offspring to the population
at each timestep.  Thus, each offspring necessarily receives her
chromosomes from distinct individuals, as diploid individuals tend to
do.  Incorporating diploidy into our model the way we do leads to a
separation of timescales problem.  Our limiting object is essentially
a `haploid' process, in which chromosomes either coalesce or
recombine.  By extending a result of \cite{Moe98}, we show that
diploidy, a fundamental characteristic of many natural populations,
can thus be treated as a `black box', since the limiting object does
not depend on the location of chromosomes in individuals.

By adopting a Moran type model, in which only a single pair of
individuals gives rise to offspring at each reproduction event, we
chose mathematical tractability over more biologically realistic
scenarios; in which, for example, many individuals contribute
offspring at each timestep.  It should be straightforward to extend
our model in many ways, for example allowing random number of parents,
or introducing population structure.  Indeed, we do extend our model
in one way, by taking a random offspring distribution.  These
extensions still leave open the question of distinguishing among
different large offspring number models.  Our work on ancestral
recombination graphs incorporating information from many loci is a
step in this direction.

Sweepstake-style reproduction induces correlation in coalescence times
even between loci separated by high rate of recombination.  The
correlation follows from the multiple merger property of our ancestral
recombination graph, since many chromosomes coalesce at the same time
in a multiple merger event.  The correlation remains a function of the
coalescence parameters ($c$ and $\psi$) of our population model.  An
immediate question is the effects on predictions of linkage
disequilibrium (LD).  The approximation $\mathfrak{D}$ by \cite{M02}
predicts low LD when recombination rate is high.  However, when the
rate of large reproduction events is high $(c \to \infty)$,
$\mathfrak{D}$ remains a function of the coalescence parameters.  The
dependence of $\mathfrak{D}$ on coalescence parameters has
implications for the use of LD in inference for populations exhibiting
sweepstake-style reproduction.  Using simulations, \cite{davies07}
found little effect of multiple mergers on the prediction $r^2$ of
linkage disequilibrium, when comparing the exact Wright-Fisher model
with recombination to the usual (continuous-time) ARG.
However, by directly incorporating large offspring number events the
way we do, we can show that large offspring number events do induce
correlation in coalescence times, and hence influence predictions of
linkage disequilibrium.

The genome-wide correlation in coalescence times
(Tables~\ref{tab:corXYa}--\ref{tab:corXYb}) induced by
sweepstake-style reproduction offers hints about how to distinguish
between large offspring number and ordinary Wright-Fisher
reproduction.  We are unaware of any published multi-loci methods
derived to distinguish among different population models.
Full likelihood methods may be preferable to the simple moment-based
methods we consider.  However, likelihood-based inference tends to be
computationally intensive, and more so for large samples.  For large
samples, one should be able to quickly obtain a good idea of the
underlying processes by comparing correlations in ratios of mutation
counts with predictions based on different population models.


In conclusion, ancestral recombination graphs admitting simultaneous
multiple mergers of ancestral lineages are derived from a diploid population
model of sweepstake-style reproduction, suggested to be common in many
diverse marine populations.  Our calculations show that
sweepstake-style reproduction results in genome-wide correlation of
gene genealogies, even for large sample sizes.  Estimates of linkage
disequilibrium and of recombination rates are confounded by the
coalescence parameters of our population model.  The genome-wide
correlation in gene genealogies induced by sweepstake-style reproduction
implies that examining correlations between loci should provide means
of distinguishing between ordinary Wright-Fisher and sweepstake-style
reproduction.  \medskip

\noindent
{\small 
  We gratefully acknowledge the comments of two anonymous referees 
  which helped to improve the presentation; one referee also spotted 
  an error in our original proof of Theorem~\ref{thm:conv.general}. 

  B.E.\ was supported in part by EPSRC grant EP/G052026/1, and
  by a Junior Research Fellowship at Lady Margaret Hall, Oxford.
  J.B.\ and B.E.\ were supported in part by DFG grant BL 1105/3-1.
  J.B.\ and B.E.\ would like to thank Institut f\"{u}r Mathematik,
  Johannes-Gutenberg-Universit\"{a}t Mainz, for hospitality.
  M.B.\ was in part supported by DFG grant BI 1058/2-1 and through ERC
  Advanced Grant 267356 VARIS.  M.B.\ would like to thank Mathematisch
  Instituut, Universiteit Leiden, for hospitality. }

\bibliographystyle{genetics}
\bibliography{recombbibB}

\section{Appendix}

\subsection{Overview of transitions and their probabilities in the finite population model} 
\label{sn:transitions}

\subsubsection{Basic setup and notation}

We will now classify all transitions and their probabilities of our
population model relevant for the ancestral process under the scaling
$\eN = c/N^2$, in which  $N$ denotes the population size.  Fix a
sample size $n$ for this section. Usually we suppress the
dependence on the sample size in the notation below.  Recall the state
space $\mathscr{A}_n$ of our ancestral process (resp.\
$\mathscr{A}_n^{\tt sm}$ for the `effective' limiting model).

Let $\Pi_N$ be the transition matrix of the Markov chain
$\{\xi^{n,N}(m)\}_{m=0,1,\dots}$ on $\mathscr{A}_n$ describing the ancestral
states of an $n$-sample in a population of size $N$.
Our aim is to decompose $\Pi_N$ into
\begin{equation} 
\label{eq:PiNdecomp}
\sub{\Pi}{N} = \sub{A}{N} + \frac1{N^2} \sub{B}{N} + \sub{R}{N} 
\end{equation}
where the matrix $\sub{A}{N}$ contains all transitions whose
probability is  $O(1)$ or $O(N^{-1})$ per generation, so that
they will happen `instantaneously' in the limit, and either are
identity transitions, or projections from $\mathscr{A}_n$ to
$\mathscr{A}_n^{\tt sm}$ by means of dispersing chromosomes paired in
double-marked individuals.  The matrix $\sub{B}{N}$ contains all
transition probabilities which are positive and finite after
multiplication with $N^2$ and $N \to \infty$, that is, our `effective
transitions'. The remainder matrix $\sub{R}{N}$ carries only
transition probabilities that are of order $O(N^{-3})$ or smaller,
that will thus vanish after scaling.

Once we have established this decomposition, we can apply
Lemma~\ref{lemma:twoscales} below in a suitable way in order to identify the limit given in Definition \ref{def:limit} and establish the convergence result, i.e.\ Theorem \ref{thm:conv}.\\

In Tables~1 -- 3 we will schematically deal with all possible transitions 
that can happen to a current sample over one timestep.    

 Analogous to the notation and convention of \cite{MS03}, we assume that in every configuration $\xi^{n,N}(m)$ from \eqref{eq:config}, the order of chromosomes in
individuals $\ind_i$ for $i \in [b(m)]$ we have
\begin{equation}%
   \label{eq:indgrouping}
   \begin{split}%
  \ind_i(m) & =  \left\{ C^{(2i - 1)}(m), C^{(2i)}(m) \right\}\\
  & \quad \textrm{if  $1 \le i \le \beta(m) - b(m)$;}   \\
  \ind_i(m)  & = \left\{ C^{(\beta(m) - b(m) + i)}(m), \emptyset \right\} \\
   & \quad \textrm{if  $\beta(m) - b(m) + 1 \le i \le b(m) $}.\\
    \end{split}%
  \end{equation}%

For ease of presentation, we denote by 
\begin{itemize}%
\item [ $\ind'$] a single-marked individual carrying one active chromosome;
\item [ $\ind''$] a double-marked individual carrying two active chromosomes;
\item [ $\tilde{\ind}'$] a single-marked individual (parent) whose marked 
chromosome is not passed on in the sample during a given reproduction event;
\item [$\hat{\ind}''$] a double-marked individual (parent) where 
one marked chromosome is passed on and the other not during a given reproduction event.
  \end{itemize}%

The symbols $(A)$, $(B)$ and $(R)$ in the tables denote whether the 
corresponding transitions belong to $\sub{A}{N}$ $(A)$, to $B_N$ $(B)$ or 
the `remainder term' $(R)$ in (\ref{eq:PiNdecomp}) according to the decomposition
mentioned above. After that, we compute all the important probabilities explicitly. The order of the probability of each transition is also noted
  in Tables~1--3.

\subsubsection{Transition type 1: Small or large reproduction event, no offspring in the sample} 
\label{subssnochild}

If a reproduction event takes place, say at generation $m$, that does not affect our sample, this will not affect the state of our ancestral process at $m+1$, and we have $\xi^{n,N}(m)=\xi^{n,N}(m+1)$. Hence, we see an identity transformation.
We now compute the probability that our sample is not affected.
Given current state $\xi \in \mathscr{A}_n$ with $b$ individuals and $\beta$ chromosomes 
(hence $\beta-b$ double-marked and $2b-\beta$ single-marked individuals), the probability
that no child is in the sample is 
\[
(1-\eN) \frac{N-b}{N} + 
\eN \frac{{N-b \choose \lfloor \psi N \rfloor}}{{N \choose \lfloor \psi N \rfloor}} 
 = 1-O(N^{-1}).
\]

%

\subsubsection{Transition type 2: Small reproduction event, offspring in sample, at most one parent in the sample, no recombination} 

Here, we only need to distinguish whether the offspring is single or
double marked, and whether there is a parent in the sample. For
example, it is immediate to see that the probability of a transition
from a double-marked $(\ind'')$ offspring to two single-marked
$(\{\ind', \ind'\})$ individuals is of order $O(N^{-1})$ when no
parent is in the sample and no recombination
happens. Table~\ref{tab:subssmrnoreconepar} lists all corresponding
events. By way of example, the state labelled
  $\{\ind^\prime, \ind^\prime \}$ denotes that two single-marked
  individuals, each carrying one active chromosome, is reached from the
  sample configuration.  One such configuration is if the sample
  contains one offspring, but neither parent $(\emptyset)$, and the
  offspring is carrying two active chromosomes $(\ind^{\prime \prime}
  )$. 

\begin{table}[h!]
  \centering
  \caption{Transitions of type 2.} 
\label{tab:subssmrnoreconepar}
\begin{tabular}{|c||c|c|c|c|} 
\hline
 & \multicolumn{2}{l|}{Parent with marked chromosome(s)} \\
 & \multicolumn{2}{l|}{($\emptyset$ means no parent in sample)} \\
Offspring & $\emptyset$ & $\ind'$ \\ \hline
$\ind''$ & $\{\ind', \ind'\}$ $(A)$ & $\{\ind',\ind'\},\{\ind'',\ind'\}$   \\
 & $(\ast)$ \;\; $O(N^{-1})$ & $O(N^{-2})$, $(B)$  \\ \hline
$\ind'$ & $\{\ind'\}$ $(A)$ & $\{\ind'\}$, $\{\ind''\}$, $\{\ind',\ind'\}$, $(B)$  \\
& $(\ast\ast)$ \;\; $O(N^{-1})$ & 
$(\dagger)$ \;\; $O(N^{-2})$  \\  \hline
\rule{0ex}{3ex} &   $\tilde{\ind}'$ & $\hat{\ind}''$ \\ \hline 
$\ind''$ 
  &  $O(N^{-2})$, $(B)$ &  $(\ddagger)$ \;\; $O(N^{-2})$, $(B)$   \\ \hline
$\ind'$ & $\{\ind''\}$, $\{\ind', \ind'\}$, $(B)$ &  $\{\ind', \ind''\}$, $\{\ind'\}$, $(B)$   \\ 
 &  $O(N^{-2})$  &  $O(N^{-2})$  \\
\hline
\end{tabular}
\end{table}

\subsubsection{Transition type 3: Small reproduction event, offspring in sample, 
both parents in the sample}

If both parents and offspring are in the sample in a small event, this immediately 
gives a transition probability of order $O(N^{-3})$ or smaller (depending on the presence of
recombination, hence will be irrelevant, and be part of $R_N$.
We omit a detailed table listing the different single- and double marked individuals.
%
%
%

\subsubsection{Transition type 4: Small reproduction event, offspring and at most one parent in sample, recombination occurs}

Table~\ref{subssmrrecnopar} lists transitions due to recombination, and when
neither parent is in the sample.  The probability of the presence of
both an offspring and at least one parent in a sample, when
recombination occurs, is of order $O(N^{-3})$, and so will vanish in
the limit. 

\begin{table}[h!]
  \centering
  \caption{Transitions of type 4, neither parent in sample }%
\label{subssmrrecnopar}\ \\
\begin{tabular}{|c||c|c|c|c|}
\hline
 & \multicolumn{1}{l|}{Parent} \\
Offspring & $\emptyset$ \\ \hline 
$\ind''$ & $\{ \ind'', \ind'\}$, $O(N^{-2})$, $(B)$ \\
 & $\{ \ind'', \ind''\}$, $O(N^{-3})$, $(R)$\\ \hline
$\ind'$ & $\ind''$ , $O(N^{-2})$, $(B)$ \\ \hline
\end{tabular}%
\end{table}%

\subsubsection{Transition type 5: Large reproduction event, offspring in sample,
no parent in sample, no recombination}

Table~\ref{subsslargerepr} lists all possible transitions when a large
reproduction event occurs, no parent is in the sample, and
recombination does not occur. The probabilities of the events listed
in Table~4 are of order $O(N^{-2})$, and so will appear as effective transitions in the limit.

\begin{table}[h!]
\centering
\caption{Transitions of type 5.}
\label{subsslargerepr}\ \\
\begin{tabular}{|c||c|c|c|c|}%
\hline
 & \multicolumn{1}{l|}{Parent} \\
Offspring & $\emptyset$ \\ \hline 
$k_1$ \OMI \hspace{0.1em}, $k_2$ \TMI & 
  $\{ \TMI, \TMI\}$\hspace{0.0em}, $O(N^{-2})$, $(B)$\\ 
& $\{ \TMI, \OMI\}$\hspace{0.0em}, $O(N^{-2})$, $(B)$\\
& $\{ \OMI, \OMI\}$\hspace{0.0em}, $O(N^{-2})$, $(B)$\\
& \TMI\hspace{0.0em}, $O(N^{-2})$, $(B)$\\
& \OMI\hspace{0.0em}, $O(N^{-2})$, $(B)$\\\hline
\end{tabular}%
\end{table}%

\subsubsection{Transition type 6: Large reproduction event, offspring in sample, recombination occurs and / or at least one parent in sample}

The probability that a large reproduction event takes place, at least
one child and at least one parent are in the sample is $O(N^{-3})$.
In addition, the probability that a large reproduction event takes
place, at least one child is in the sample and also a recombination
event happens in the sample is $O(N^{-3})$. Hence all such events are negligible.

\subsection{The convergence result}
\label{ssn:resultproof}

\subsubsection{The limit of the projection matrix $\sub{A}{N}$}


Some care is needed in order to make sure $\sub{A}{N}$ converges in the
right sense to the desired projection matrix.  The only relevant
transitions of order $O(1)$ or $O(N^{-1})$ are transitions of type 1
and 2.  The only one which is not an identity transition is the first
dispersion event of Table~\ref{tab:subssmrnoreconepar}.  For $\xi \in
\mathscr{A}_n$ with $b < \beta$ (i.e.\ at least one marked individual
is double-marked), that is
$$
\xi \mapsto \mathsf{disp}_i(\xi).
$$
This event will become part of $A_N$, and has probability
\begin{align}%
\label{eq:ANsplit}
\sub{A}{N}(\xi, \mathsf{disp}_i(\xi)) = 
(1-\sub{\varepsilon}{N}) \frac{1}{N} \frac{{N-b-1 \choose 2}}{{N \choose 2}} 
(1-\rN)^2, \qquad 1 \le i \le \beta-b 
\end{align}
(this is the probability of the event $(\ast)$ listed in row~1, column~1 of 
Table~\ref{tab:subssmrnoreconepar}, note that the 
event $(\ast\ast)$ listed in row~2, column~1 there leads to an 
identity transition).
Otherwise, we have 
\[
\sub{A}{N}(\xi, \xi) = 
1 - (1-\varepsilon_N) \frac{\beta-b}{N} \frac{{N-b-1 \choose 2}}{{N \choose 2}} 
(1-r_N)^2
\]
Of course, $A_{_N}$ has to leave elements of the subspace $\mathscr{A}_n^{\tt sm}$
invariant, hence we set,
for $\xi$ with $b=\beta$,
\[
\sub{A}{N}(\xi, \xi') := \osf_{\{\xi=\xi'\}}.
\]

\begin{prop}
\label{prop:an}
With the above settings, $\sub{A}{N}$ is a stochastic matrix for each $N$ and 
\begin{equation} 
\label{eq:ANlimit}%
\lim_{C\to\infty} 
\lim_{N\to\infty} \sup_{r \geq C N} || \sub{A}{N}^r - P || = 0
\end{equation}
for all $C>0$ large enough, where $P$ is the canonical projection from $\mathscr{A}_n$
to $\mathscr{A}_n^{\tt sm}$, i.e.\ 
\[
P(\xi,\xi') = \osf_{\{\xi'=\mathsf{cd}(\xi)\}}. 
\]
\end{prop}

\noindent \emph{Proof of Proposition \ref{prop:an} }
The Markov chain with transition matrix $\sub{A}{N}$ can only change state by dispersing the chromosomes paired in a  double-marked individual.  We see
from (\ref{eq:ANsplit}) that 
$$
\sub{A}{N}(\xi, \mathsf{disp}_i(\xi)) \ge \frac{K(n, r, c)}{N}
$$ 
for some suitable constant $K(n, r, c)$, uniformly in $b$ and $i \le \beta-b$ and $N$ (for all $N$ large enough). Hence, starting from $\xi$ with $\beta-b$ double-marked individuals, the number of
$\sub{A}{N}$-steps required until complete dispersion has occurred is
dominated by the sum of $\beta-b$ independent geometric random variables $\gamma^{(N)}_1+\cdots+\gamma^{(N)}_{\beta-b}$, 
with success probability $K(n, r, c)/N$.
By Markov's inequality,
$$
\sup_{N \in \N}
\P\Big\{\gamma^{(N)}_1+\cdots+\gamma^{(N)}_{\beta-b} \ge C N \Big\} 
\le \frac{1}{CN}\E\big[\gamma^{(N)}_1+\cdots+\gamma^{(N)}_{\beta-b}\big] = \frac{N(\beta-b)}{C \cdot N \cdot K(n, r, c)}
\to 0 \quad \mbox{ as } \;\;
C\to\infty.
$$  
The proof can now be completed with a coupling argument, noting that two Markov chains run according to $A_N$ resp.\ $P$, started in $\xi \in \mathscr{A}_n$ get both stuck in $\mathsf{cd}(\xi)$, and this happens after at most $CN$ steps with high probability (for $C$ large).


\hfill $\qed$ \medskip


\subsubsection{Proof of the convergence result}

With the definition of $A_N$ from the previous section, put
\begin{equation} 
\label{eq:BNstar}
B^*_N := N^2 (\Pi_N - \sub{A}{N}),
\end{equation}
and let $P$ be the canonical projection from $\mathscr{A}_n$
to $\mathscr{A}_n^{\tt sm}$ defined in Proposition \ref{prop:an}.
The following Lemma will identify $G$ as the limit containing all the `effective' transitions
of $B^*_N$ when projecting on the subspace $\mathscr{A}_n^{\tt sm}$.

\begin{lemma}
\label{lemma:hatBNconv}
We have 
\begin{equation}
\label{eq:BNhat}
\widehat{B}_N := P B_N^* P \to G \quad \text{as}\;\; N\to\infty
\end{equation}
with $G$ from (\ref{eq:DefG}).
\end{lemma}
\begin{remark} 
We do believe that in fact the sequence of (formally larger) matrices 
$B_N^*$ on $\mathscr{A}_n$ converges as well, but 
the statement about $\widehat{B}_N$ is sufficient 
for our purposes below (see (\ref{lemma:twoscales:eq:BNlimit}) 
in Lemma~\ref{lemma:twoscales}) and 
simpler to prove since it allows to restrict to 
the `completely dispersed' configurations in $\mathscr{A}_n^{\tt sm}$.
\end{remark}

\begin{proof}[Proof of Lemma~\ref{lemma:hatBNconv}]
  We inspect the types of events listed in Tables~\ref{tab:subssmrnoreconepar}-\ref{subsslargerepr} that are marked
  with $(B)$. Events that are marked with $(R)$ have probability of
  order at most $O(N^{-3})$, hence their total contribution to any
  entry of $\widehat{B}_N$ is at most $O(N^{-1})$ (since we are
  following a finite sample, there are only finitely many possible
  one-step events altogether).  It suffices to consider
  $\widehat{B}_N(\xi, \mathsf{cd}(\eta))$ for $\xi =\left\{ C^{(1)}, \ldots,  C^{(\beta)}; \beta \right \}\in \mathscr{A}_n^{\texttt{sm}},
  \eta \in \mathscr{A}_n$ (because $P$ projects to
  $\mathscr{A}_n^{\texttt{sm}}$).  \smallskip

Regarding $\xi'=\mathsf{pairmerge}_{j_1,j_2}(\xi)$: This transition can
happen in a small reproduction event 
(these events are listed at $(\dagger)$ in Table~\ref{tab:subssmrnoreconepar} 
in row~2, column~2, note that events listed at $(\ddagger)$ 
in Table~\ref{tab:subssmrnoreconepar} 
lead to a trivial transition once $P$ is applied)
or in a large reproduction event as in
Table~\ref{subsslargerepr} if the grouping is suitable. Up to
four parental chromosomes are involved in any reproduction event.
Hence, a large reproduction event can lead to a given pair merger in
the sample if up to $5$ individuals in the sample are children.
Thus
\begin{align} 
\widehat{B}_N(\xi,\xi') = & 
N^2 (1-\eps_N)(1-r_N) 2 \times \frac1{N} \frac{1 \cdot (N-b)}{{N-1 \choose 2}} 
\frac12 \frac12 \notag \\
& {} + N^2 \eps_N \sum_{c=2}^5 (1-r_N)^c  {\beta-2 \choose c-2} 
\frac{{N-\beta \choose \lfloor N\psi \rfloor-c}}%
{{N \choose \lfloor N\psi \rfloor}} (4)_{c-1}\big(\tfrac14\big)^c
 + O(N^{-1})
\end{align}
{\small (For the first term on the right note that either $j_1$ or 
$j_2$ can be the child, the two factors of $\tfrac12$ come from the 
requirement that the chromosome in the child we are following is the 
one from the parent in the sample and is also the one we are following 
in the parent. For the second term on the right note that once we 
decide on $c$ children in the sample (${\beta-2 \choose c-2}$ choices 
because $j_1$ and $j_2$ are already chosen), there are $(4)_{c-1}$ ways 
to assign them to the $4$ parental chromosomes. 
For comparison with (\ref{eq:Xirate}) and the first line in 
(\ref{eq:DefG}) observe 
\[
\frac{{N-\beta \choose \lfloor N\psi \rfloor-c}}%
{{N \choose \lfloor N\psi \rfloor}} 
= \frac{(N-\beta)! \lfloor N\psi \rfloor! (N-\lfloor N\psi \rfloor)!}{%
(\lfloor N\psi \rfloor-c)! (N-\beta-\lfloor N\psi \rfloor+c)! N!}
\sim \frac{(N\psi)^c(N(1-\psi))^{\beta-c}}{N^\beta}=\psi^c(1-\psi)^{\beta-c}.
\] 
}

Regarding $\xi'= \mathsf{recomb}_{j,\ell}(\xi)$ (assuming that 
$\alpha$ is such that $C^{(j)}$ can be non-trivially cut into 
two by a recombination event between loci $\ell-1$ and $\ell$): 
This transition can happen in a small reproduction event as 
listed at $(\ast\ast)$ in 
Table~\ref{subssmrrecnopar} or in another event that has 
probability $O(N^{-3})$. 
Hence 
\begin{align} 
\widehat{B}_N(\xi,\xi') = 
N^2 (1-\eps_N) \times \frac1{N} \frac{{N-b \choose 2}}{{N-1 \choose 2}} 
\frac{r^{(\ell)}}{N} + O(N^{-1}) = r^{(\ell)} + O(N^{-1}).
\end{align}

Regarding $\xi'= \mathsf{groupmerge}_{J_1,J_2,J_3,J_4}(\xi)$: 
This can only occur through a large reproduction event as listed in 
Subsection~\ref{subsslargerepr}. 
Write $k_i:=|J_i|$, we assume $k_1 \ge \dots \ge k_a \ge 2$ for some 
$a \in [4]$, $k_{a+1}=\dots=k_4=0$ (if $a=1$, $k_1\ge 3$), 
$s:=\beta-(k_1+\dots+k_a)$ is the number of singletons (non-participating 
chromosomes) in the merger. Note that by the structure of the diploid 
model, with $a$ groups merging there can be up to $k_1+\dots+k_a+(4-a)^+$ 
children in the sample (put differently: up to $(4-a)^+$ 
`non merging children'). Then
\begin{equation}
  \begin{split}%
\widehat{B}_N(\xi,\xi') = \,&
N^2 \eps_N \sum_{c'=0}^{(4-a)^+} 
{\beta-k_1-\dots-k_a \choose c'} (1-r_N)^{k_1+\dots+k_a+c'} \notag \\
& \hspace{8em} \times \frac{{N-\beta \choose \lfloor N\psi \rfloor -(k_1+\dots+k_a+c')}}{{N \choose \lfloor N\psi \rfloor}} (4)_{a+c'} 
\big(\tfrac14\big)^{k_1+\cdots+k_a +c'} \\
 &  + O(N^{-1}).
 \end{split}
\end{equation}
\medskip

It remains to check that the diagonal terms behave correctly, i.e.\ that 
as $N\to\infty$, 
\begin{align}
\label{eq:hatBNdiaglim1}
\widehat{B}_N(\xi,\xi) \to G(\xi,\xi) = 
- \sum_{\xi'\neq\xi, \, \xi' \in \mathscr{A}_n^{\tt sm}}\hspace{-1em}G(\xi,\xi'). 
\end{align}
Because $\Pi_N$ and $\sub{A}{N}$ are both stochastic matrices (as is $P$), 
we have 
\begin{align}
\label{eq:hatBNdiaglim2}
\widehat{B}_N(\xi,\xi) = 
- \sum_{\xi'\neq\xi, \, \xi' \in \mathscr{A}_n^{\tt sm}}\hspace{-1em} 
\widehat{B}_N(\xi,\xi')
\end{align}
for each $N$. By inspection and the discussion above, all terms in 
$\Pi_N$ with decay rate $1/N$ 
are accounted for in $\sub{A}{N}$, and all non-diagonal terms in $\Pi_N-\sub{A}{N}$ with 
decay rate $1/N^2$ appear after multiplication with $N^2$ in $\widehat{B}_N$ 
with their correct limits, namely the corresponding terms in $G$, 
while terms with a faster decay rate disappear in the limit. 
Hence (\ref{eq:hatBNdiaglim2}) implies (\ref{eq:hatBNdiaglim1}). 
\end{proof}


\subsection{Markov chains with two time-scales --- 
a variation on a lemma of M\"ohle}

Conceptually, our convergence result rests on a separation of time-scales phenomenon. It can be
established with the help of a variant of a well-know result, see Lemma~1 from 
\cite{Moe98}. 

Let $E$ be a finite set. We equip matrices $A=(A(x,y))_{x,y \in E}$ on $E$ 
with the matrix norm $||A||:= \max_{x\in E} \sum_{y\in E} |A(x,y)|$. 
Note that then $||AB|| \leq ||A|| \, ||B||$ and $||A||=1$ if $A$ 
is a stochastic matrix.

\begin{lemma} 
\label{lemma:twoscales}
Assume that for $N \in \N$, $\sub{A}{N}$ is a stochastic matrix on $E$ such that 
\begin{equation} 
\label{lemma:twoscales:assumpAN}
\lim_{C\to\infty} \lim_{N\to\infty} \sup_{r \geq C N} || \sub{A}{N}^r - P || = 0
\end{equation}
for some matrix $P$. Then we have for any $0 < c, K, t < \infty$
\begin{equation} 
\label{lemma:twoscales:eq:bdperturb}
\lim_{N\to\infty} \sup_{||B|| \leq K} 
|| (\sub{A}{N} + cN^{-2}\,B)^{[t N^2]} - (P + cN^{-2}\,B)^{[t N^2]}|| = 0.
\end{equation}

Furthermore, if $(\sub{B}{N})_{N\in\N}$ is a sequence of matrices on $E$ such 
that 
\begin{equation}
\label{lemma:twoscales:eq:BNlimit}
G:=\lim_{N\to\infty} P \sub{B}{N} P \qquad \text{exists}, 
\end{equation}
then 
\begin{equation} 
\label{lemma:twoscales:limit}
\lim_{N\to\infty}  (\sub{A}{N} + cN^{-2}\, \sub{B}{N})^{[t N^2]} = Pe^{ctG} \qquad 
\text{for all}\;\; t > 0.
\end{equation}
\end{lemma}

\begin{remark} 
  \label{rem:twoscales.general} 
  Instead of time scales $N$ and $N^2$ one can allow more 
  generally any $a_N, b_N\to\infty$ with $b_N/a_N\to\infty$, with 
  only notational modifications in the proof. 
\end{remark}
%
%
%

\begin{proof}[Proof of Lemma~\ref{lemma:twoscales}]
We begin with (\ref{lemma:twoscales:eq:bdperturb}). 
W.l.o.g.\ assume $K=1$, otherwise replace $B$ by $B/K$ and $c$ by $cK$. 
Fix $c, t > 0$ and a matrix $B$ with $||B||\leq 1$, 
abbreviate $m:=[tN^2]$. 
Let $\varepsilon>0$, choose $C_0 < \infty$ and $N_0\in\N$ such that 
\begin{equation} 
|| \sub{A}{N}^r - P || \leq \varepsilon \qquad 
\text{for $N\geq N_0$, $r \geq C_0N$}
\end{equation}
(as guaranteed by (\ref{lemma:twoscales:assumpAN})). Note that 
\begin{align*} 
|| & (\sub{A}{N} + cN^{-2}\,B)^m - (P + cN^{-2}\,B)^m|| \\
& \, \le ||\sub{A}{N}^m - P || + 
\sum_{k=1}^m \Big(\frac{c}{N^2}\Big)^k 
\sum_{m_1,\dots,m_{k+1}\in \N_0 \atop m_1+\cdots+m_{k+1}=m-k} 
\Big|\Big| \sub{A}{N}^{m_1} \prod_{j=2}^{k+1}\big(B\sub{A}{N}^{m_j}\big) - 
P^{m_1} \prod_{j=2}^{k+1}\big(BP^{m_j}\big) \Big|\Big|.
\end{align*} 
Mimicking the proof in \cite{Moe98}, we split the second summand into 
(the ellipses refer to the term inside the large norm brackets 
on the right of the last line of the previous formula)
\[
S_1 := \sum_{k=1}^m \Big(\frac{c}{N^2}\Big)^k 
\sum_{m_1,\dots,m_{k+1}\ge C_0 N \atop m_1+\cdots+m_{k+1}=m-k} ... 
\quad \text{and} \quad 
S_2 := \sum_{k=1}^m \Big(\frac{c}{N^2}\Big)^k \hspace{-2em}
\sum_{\begin{array}{c}\scriptstyle m_1,\dots,m_{k+1}\in\N_0 \\[-0.8ex] 
\scriptstyle m_1+\cdots+m_{k+1}=m-k \\[-0.8ex] 
\scriptstyle \exists\, j \, : \: m_j < C_0 N \end{array}} ... 
\]
As in \cite{Moe98}, p.~509 we have $S_1 \leq 2e^t(t+1) \varepsilon$ 
for all $N$ large enough, our estimate for $S_2$ is a small variation 
of the corresponding estimate in \cite{Moe98}: Note that each of the 
matrix norms appearing in the big sum in $S_2$ is at most $2$, hence 
\begin{align*} 
S_2 \leq & \, 2 \sum_{k=1}^m \Big(\frac{c}{N^2}\Big)^k 
\# \Big\{ (m_1,\dots,m_{k+1}) \in \N_0^k \, : 
\begin{array}{l} m_1+\cdots+m_{k+1}=m-k, \\
\exists\, j \, : \: m_j < C_0 N \end{array} \Big\} \\
\leq & \, 2 \sum_{k=1}^m \Big(\frac{c}{N^2}\Big)^k 
(k+1) \sum_{m_{1}=0}^{C_0 N \wedge (m-k)} {m-m_1-1 \choose k-1} \\
\leq & \, 2 \sum_{k=1}^m \Big(\frac{c}{N^2}\Big)^k (k+1) C_0N {m-1 \choose k-1}
= 2C_0 N \frac{c}{N^2} 
\sum_{k=0}^{m-1} \Big(\frac{c}{N^2}\Big)^k (k+2) {m-1 \choose k} \\
\leq & \, C'\frac1N.
\end{align*} 
{\small (We use in the last estimate that for $|x|<1$, $n\in\N$, 
$\sum_{n=0}^\infty {n \choose k} x^k =(1+x)^n$ and 
$\sum_{n=0}^\infty k {n \choose k} x^k =nx(1+x)^{n-1}$.)}

The derivation of (\ref{lemma:twoscales:limit}) from 
(\ref{lemma:twoscales:eq:bdperturb}) is literally the same as in 
\cite{Moe98}, p.~509-511 (read $c_{_N}=c/N^2$ there).
\end{proof}

%

\subsection{The convergence result with general random $\sub{\Psi}{N}$}
\label{subsect:genPsiproofs}

In this section we briefly indicate how the proof of
Theorem~\ref{thm:conv} can be modified to yield
Theorem~\ref{thm:conv.general}.  In each reproduction event, a random
number $\sub{\Psi}{N}$ of individuals die and are replaced by the same
number of offspring, and recall Assumptions~\eqref{cond1:cnvanishes},
\eqref{cond2:sagitovs} and \eqref{eq:recombratescaling.general}.  By
``short'' time-scale we refer to the scaling $a_N$ given by
\[
a_N = \frac{N}{\E\left[\PsiN \right]}
\]
and by ``long'' time-scale the scaling $b_N$ given by 
\[
b_N = \frac{1}{\sub{c}{N}} = \frac{N (N-1) }{\E\left[\PsiN (\PsiN + 3) \right]} .
\]

Assumption~\eqref{cond1:cnvanishes} yields $b_N\to\infty$ as
$N\to\infty$, and $b_N/a_N \to \infty$ by 
Assumption~\eqref{eq:tscaleratiogen}. 
To check \eqref{eq:tscale1gen}, 
i.e.\ that indeed $a_N\to\infty$,  
observe that $\PsiN/N$ is a positive random variable, bounded
by $1$.  Condition \eqref{cond1:cnvanishes} is equivalent to $\E\left[
  \big(\PsiN/N\big)^2 \right] \to 0$, which implies $\PsiN/N \to 0$ in probability and
$\E\left[ \PsiN/N \right] \to 0$, hence \eqref{eq:tscale1gen}.  

\smallskip

For use below, we recall implications of 
\eqref{cond2:sagitovs} provided that \eqref{cond1:cnvanishes} holds 
(cf \cite{S99}): 
\begin{align} 
\text{For all} \: j \ge 3\, : \;\; 
\frac{1}{\sub{c}{N}} \E\left[ \left(\frac{\PsiN}{N}\right)^j \right] 
\mathop{\longrightarrow}_{N\to\infty} \int_{[0,1]} x^{j-2} \, F(dx). 
\end{align}
Indeed, integration by parts yields 
\begin{align} 
\label{cond2:sagitovs-v2}
\frac{1}{\sub{c}{N}} \E\left[ \left(\frac{\PsiN}{N}\right)^j \right] 
&= \frac{1}{\sub{c}{N}} \int_{(0,1]} jx^{j-1} \P\left(\frac{\PsiN}{N}>x\right)dx
\notag \\
& \mathop{\longrightarrow}_{N\to\infty} 
 \int_{(0,1]} jx^{j-1} \int_{(x,1]} y^{-2} F(dy)\, dx 
= \int_{(0,1]} \left( \int_{(0,1]} 1_{\{x \le y \}}jx^{j - 1}dx  \right)y^{-2}F(dy)
\notag \\
& \hspace{2.5em} = \int_{(0,1]} y^{j-2}F(dy).
\end{align}
Furthermore for the case $j=2$ one obtains
\begin{align} 
\label{cond2:sagitovs-v2b}
\limsup_{N\to\infty} \frac{1}{\sub{c}{N}} 
\E\left[ \left(\frac{\PsiN}{N}\right)^2 \right] 
= \limsup_{N\to\infty} \frac{\E\left[ \PsiN^2 \right]}{%
\E\left[ \PsiN(\PsiN+3) \right]} \le 1 < \infty.
\end{align}

Let $\sub{\widetilde\Psi}{N}$ have the following 
reweighted distribution (relative to $\PsiN$): 
\begin{align}
  \label{eq:PsiNweighted}%
\mathbb{P}\big(\sub{\widetilde\Psi}{N} = k \big) 
= \frac{k(k+3)}{\E[\PsiN(\PsiN +3)]} \mathbb{P}\big(\PsiN = k \big), 
\quad k=1,\dots,N-2, 
\end{align}
then 
\begin{align} 
\label{cond2:sagitovs-v3}
\frac{\sub{\widetilde\Psi}{N}}{N} \, \mathop{\longrightarrow}^d \, 
F \quad \text{as}\; N\to\infty.
\end{align}
Indeed, for any $\ell \in \N$ 
\begin{align} 
\label{cond2:sagitovs-v3b}
\E\left[ \left(\frac{\sub{\widetilde\Psi}{N}}{N}\right)^\ell \right] 
& =  \frac{N (N-1)}{\E[\PsiN(\PsiN +3)]} 
\E\left[ \left(\frac{\PsiN}{N}\right)^{\ell+1} \frac{\PsiN+3}{N-1} 
\right] \notag \\
& = \frac1{\sub{c}{N}} \E\left[ \left(\frac{\PsiN}{N}\right)^{\ell+2}
\right] \frac{N}{N-1} 
+ \frac{3}{(N-1)} \frac1{\sub{c}{N}} \E\left[ \left(\frac{\PsiN}{N}\right)^{\ell+1} \right] 
\mathop{\longrightarrow}_{N\to\infty} \int_{(0,1]} y^\ell \, F(dy)
\end{align}
by \eqref{cond2:sagitovs-v2} and \eqref{cond2:sagitovs-v2b}, so 
\eqref{cond2:sagitovs-v3} follows because the moments characterise a 
probability law on $[0,1]$. 
One can check (along the lines of \cite{S99}) that under Assumption\ 
\eqref{cond1:cnvanishes}, both \eqref{cond2:sagitovs-v2} and 
\eqref{cond2:sagitovs-v3} are in fact equivalent to \eqref{cond2:sagitovs}.
\smallskip

The proof of Theorem~\ref{thm:conv.general} is now a relatively 
straightforward adaptation of the proof of Theorem~\ref{thm:conv} 
discussed in Sections~\ref{sn:transitions} and \ref{ssn:resultproof} 
above. Scaling by $N$ is throughout replaced by scaling with $a_N = N/\EE{\PsiN}$ 
and scaling by $N^2$ becomes scaling with $b_N = N(N-1)/\EE{\PsiN(\PsiN + 3)}$:
\begin{enumerate}
\item 
When currently following $b \ge 1$ individuals, the probability that none of 
them is an offspring in the previous reproduction event 
(and hence the sample configuration remains unchanged) is 
\[
\E\left[\frac{ {N-b \choose \sub{\Psi}{N}} }{ {N \choose \sub{\Psi}{N}} }\right]
= \E\left[ \prod_{j=0}^{\sub{\Psi}{N}-1} \frac{N-b-j}{N-j} \right] 
= \E\left[ \prod_{j=0}^{\sub{\Psi}{N}-1} \big( 1 -\frac{b}{N-j} \big) \right]
= 1 - O\left( b \frac{\E\left[\PsiN \right]}{N} \right)
= 1- O\big(a_N^{-1}\big).
\]
This is analogous to transitions discussed in Section~\ref{subssnochild} 
and happens ``all the time'' (leading to the projecting transitions part in the limit). 
\item 
When currently following $b \ge 1$ individuals, say the $i$-th of which 
is double-marked, the probability that the $i$-th individual is the only 
offspring in the sample, and the sample also does not contain a parent, is (we write $(x)_k = x (x-1) \cdots (x-k+1)$ for the $k$-th falling 
factorial) 
\[
\E\left[ \frac{\PsiN (N-\PsiN-2)_{b-1}}{(N)_b} \right] 
\sim \E\left[ \frac{\PsiN}{N} \Big( 1 - \frac{\PsiN}{N} \Big)^{b-1} \right] 
= a_N^{-1} \big( 1 + o(1) \big).
\]
The projection matrix $\sub{A}{N}$ now becomes 
\begin{equation} 
\sub{A}{N}(\xi, \mathsf{disp}_i(\xi)) = 
\E\left[ \frac{\PsiN (N-\PsiN-2)_{b-1}}{(N)_b} \right] 
(1-\rN)^2, \qquad 1 \le i \le \beta-b 
\end{equation}
and $\sub{A}{N}(\xi, \xi) = 1 - (\beta-b) 
\E\left[ \frac{\PsiN (N-\PsiN-2)_{b-1}}{(N)_b} \right] 
(1-\rN)^2$; the analogue of Proposition~\ref{prop:an} is then 
\begin{equation} 
\lim_{C\to\infty} 
\lim_{N\to\infty} \sup_{r \geq C a_N} || \sub{A}{N}^r - P || = 0.
\end{equation}
\item 
From now on we can work on the ``projected'' space 
$\mathscr{A}_n^{\tt sm}$. 
The distinction between small and large reproduction events 
is irrelevant in the general case.   Hence, it is more 
suitable to distinguish whether a parent and an offspring are in the 
sample or whether several offspring (but no parent) is in the sample. 
In analogy with \eqref{eq:BNstar} and \eqref{eq:BNhat}, 
we split $\sub{\Pi}{N}$ into ``fast'' and ``slow'' parts and define 
\begin{equation} 
  B^*_N := b_N (\Pi_N - \sub{A}{N}), \quad \widehat{B}_N := P B_N^* P.
\end{equation}
It then remains to check that 
\begin{equation} 
\label{eq:BNhatconv.general}
\widehat{B}_N \to G \quad \text{with $G$ defined in \eqref{eq:DefG.general}}, 
\end{equation}
whence Theorem~\ref{thm:conv.general} follows from 
Lemma~\ref{lemma:twoscales} together with Remark~\ref{rem:twoscales.general}. 
\medskip

\hspace{-3em} We now verify \eqref{eq:BNhatconv.general}: 

\item Recombination events: These give the correct limit, see the discussion 
below \eqref{eq:recombratescaling.general}. 

\item ``Large:'' The probability that exactly $k \ge 2$ individuals
  among $b$ (excluding the parents) is, using
  \eqref{eq:PsiNweighted},
\begin{equation} 
\label{eq:probkoffspr.gen}
\E\left[ \frac{(\PsiN)_k(N-\PsiN-2)_{b-k}}{(N)_b} \right] 
= \E\big[\PsiN(\PsiN +3)\big] 
\E\left[ \frac{(\sub{\widetilde{\Psi}}{N})_k(N-2-\sub{\widetilde{\Psi}}{N})_{b-k}}{\sub{\widetilde{\Psi}}{N} (\sub{\widetilde{\Psi}}{N}+3) (N)_b} \right], 
\end{equation}
thus $1/\sub{c}{N}$ times this probability is 
\begin{align} 
& N(N-1) \E\left[ \frac{(\sub{\widetilde{\Psi}}{N})_k (N-2-\sub{\widetilde{\Psi}}{N})_{b-k}}{\sub{\widetilde{\Psi}}{N} (\sub{\widetilde{\Psi}}{N}+3) (N)_b} \right] \notag  = \frac{1}{(N-2)_{b-2}} 
\E\left[ (\sub{\widetilde{\Psi}}{N}-2)_{k-2} 
(N-2-\sub{\widetilde{\Psi}}{N})_{b-k} \right] 
+ O\big(\frac1N\big) \\
& \mathop{\longrightarrow}_{N\to\infty} \int_{(0,1]} y^{k-2} (1-y)^{b-k} \, F(dy)
\end{align}
by \eqref{cond2:sagitovs-v3}.
Furthermore, the probability that at least $2$ 
offspring \emph{and} at least one parent are in the sample is at most
\begin{align} 
b {b-1 \choose 2} \E\left[ \frac{2(\PsiN)_2}{(N)_3}  \right] = 
O\big(\sub{c}{N}/N\big)
\end{align}
hence such events become negligible in the limit.


\item \label{item:small} ``Small''  (=a merger of a single pair, which can result either from 
one offspring and one parent in the sample, or from two offspring but no 
parent in the sample): 
Here, the weight of $F(\{0\})$ plays a role. 


The probability that exactly two given single-marked individuals in a sample of size 
$b$ are offspring (and none are parents) is 
\begin{align} 
\E\left[ \frac{(\PsiN)_2  (N-2-\PsiN)_{b-2}}{(N)_b} \right],
\end{align}
and the probability that among a pair of two given single-marked individuals, one is 
a parent, the other an offspring and no other element of the sample is 
affected by the reproduction event is 
\begin{align} 
\E\left[ \frac{2 (2)_1 (\PsiN)_1  (N-\PsiN -2)_{b-2}}{(N)_b} \right],
\end{align}
thus, $1/\sub{c}{N}$ times the probability that exactly one given pair 
(of single-marked individuals) is involved in a reproduction event is 
\begin{align} 
\frac1{\sub{c}{N}} & 
\E\left[ \frac{\PsiN (\PsiN+3)  (N-\PsiN -2)_{b-2}}{(N)_b} \right] 
= \E\left[ \frac{(N-2 - \sub{\widetilde{\Psi}}{N})_{b-2}}{(N-2)_{b-2}} \right] 
\notag \\
& \hspace{4em} 
\mathop{\longrightarrow}_{N\to\infty} \int_{[0,1]} (1-y)^{b-2} \, F(dy) 
= F(\{0\}) + \int_{(0,1]} (1-y)^{b-2} \, F(dy)
\end{align}
by \eqref{cond2:sagitovs-v3}. 
\item (Combinatorial connections between participation in reproduction
  events and merging of ancestral chromosomes) The rest of the
  argument in order to replace \eqref{eq:Xirate} by
  \eqref{eq:Xirate.general} is purely combinatorial; it is only
  concerned with possible groupings of the $k$ single-marked offspring
  into up to four groups depending on which of the four parental
  chromosomes they descend from.

In both cases considered in (\ref{item:small}) the probability
that the chromosomes actually coalesce is $\frac14$ because they must
descend from the same chromosome in the same parent, or from the
particular chromosome in the particular parent we are following,
respectively.

\end{enumerate}

\subsection{Correlation in coalescence times}%
\label{sn:correlation}

In this section we outline the calculations to obtain the correlation
in coalescence times $T_1$ and $T_2$ of types at two loci (1 and
2). As our sample consists of two unlabelled chromosomes typed at two
loci, we will sometimes find it convenient to denote an unlabelled
chromosome carrying ancestral segments at both loci with the symbol
$\chr$, while chromosomes carrying ancestral segments at only one
locus with the symbols $\achr$ and $\bchr$.  Loci at which types have
coalesced will be denoted by $\cachr$, or $\cabchr$.  The states $\mathfrak{S}$ of
the unlabelled process for a sample of size two at two loci will
also be numbered as follows:
 \begin{center}%
    \begin{tabular}{cc}
      \hline
      $\mathfrak{S}$ & in symbols \\
      \hline
      2 & $(\chr)(\chr)$ \\\\
      1 & $(\chr)(\achr)(\bchr)$\\\\
      0 & $(\achr)(\achr)(\bchr)(\bchr)$\\\\
      $-1$ &  $(\bchr)(\bchr)$ \\\\
      $-2$ & $(\achr)(\achr)$ \\
      \hline
      \end{tabular}%
    \end{center}%
    in which states $\{0,1,2\}$ denote the three possible sample
    states, before coalescence at either loci has occurred.  States
    $\{-1, -2\}$ will be needed when deriving the variance of pairwise
    differences.

    Let $h(i) := \mathbb{P}\left( \{T_1 = T_2\} | i\right)$ denote the
    probability of the event $T_1 = T_2$, when $B$ is in state $i$.
    Excluding large offspring numbers, one readily obtains ($h(i) = 0$
    for $i \neq \{0,1,2\}$)
    \begin{equation}%
      \label{eq:hiKlimr}%
      \begin{split}%
        h(2) & = \tfrac{r + 9}{2r^2 + 13r + 9} \\
        h(1) & = \tfrac{3}{2r^2 + 13r + 9} \\
        h(0) & = \tfrac{2}{2r^2 + 13r + 9} \\
        \end{split}%
      \end{equation}%
      For each $i \in \{0,1,2\}$, the expression for $h(i)$ is the
      same as the one for the correlation between $T_1$ and $T_2$ when
      in state $i$, excluding large offspring numbers.  The expected
      value $w(i) = \mathbb{E}_i[T_s]$ of the time $T_s$ until a
      coalescence event at either locus starting from state $i \in
      \{0,1,2 \}$ is, again excluding large offspring numbers,
\begin{displaymath}%
  \begin{split}%
 w(2) & =  \tfrac{r + 9}{2(2\, r^2 + 13\, r + 9)} + \tfrac{1}{2} = \tfrac{1}{2}\left( 1 + h(2) \right) \\
 w(1) & = \tfrac{3}{2(2\, r^2 + 13\, r + 9)} + \tfrac{1}{2} = \tfrac{1}{2}\left(1 + h(1) \right), \\
 w(0) & = \tfrac{1}{2r^2 + 13r + 9} + \tfrac{1}{2} = \tfrac{1}{2}\left(1 +  h(0) \right) , \\
  \end{split}%
  \end{displaymath}%
  obtained by solving the recursions
  \begin{displaymath}%
    \begin{split}%
      w(2) & = (1 + 2rw(1))/(1 + 2r) \\
      w(2) & = (1 + 2rw(1))/(1 + 2r) \\
     w(1) & = (1 + w(2) +  rw(0))/(r + 3) \\
     w(0) & = (1 + 4w(1))/6\\
    \end{split}%
    \end{displaymath}%
    Let $v(i) := \mathbb{E}_i[T_s^2]$ denote the expected value of
    $T_s^2$ when starting from state $i \in \{0,1,2\}$.  One can
    follow \cite{D02} to obtain the recursions
    \begin{equation}%
      \label{eq:recvi}%
    v(i) =  \frac{2}{q_i^2} + \frac{2}{q_i} \sum_{k \ne i}\frac{\sub{q}{ik}}{q_i}w(k) +  \sum_{k \ne i}\frac{\sub{q}{ik}}{q_i} v(k)
    \end{equation}%
    in which $q_i = \sum_{k \ne i}\sub{q}{ik}$ is the sum of the
    transition rates out of state $i$.  To obtain \eqref{eq:recvi} let
    $J$ denote the exponential waiting time until the first
    transition, and $X_J$ the state of the process immediately after
    the first transition.  The random variables $J$ and $X_J$ are
    independent.  One can write
    \begin{displaymath}%
      \begin{split}%
      \EE{T_s^2 | J, X_J} & = \EE{ (T_s - J + J)(T_s - J + J) | J, X_J }  \\
      & = \EE{ (T_s - J)^2  + 2J(T_s - J) + J^2 | J, X_J } \\
      & = \EE{(T_s - J)^2 | J, X_J} + 2J\EE{T_s - J | X_J} + \EE{J^2}
      \end{split}%
      \end{displaymath}%
      Taking expectations gives \eqref{eq:recvi}.

    The variance $\mathbb{V}_i[T_s]$ of $T_s$ when starting in state
    $i$ is given by
    \begin{displaymath}%
      \begin{split}%
        \mathbb{V}_2[T_s] & =  \frac{r^3 + \tfrac{31\, r^2}{2} + \tfrac{153\, r}{2} + 81}{\left(2\, r + 1\right)\, \left(r + 6\right)\, \left(2\, r^2 + 13\, r + 9\right)} + \frac{1}{2} - \frac{1}{4}(1 + h(2))^2 \\
        \mathbb{V}_1[T_s] & = \frac{r + 9}{\left(r + 6\right)\, \left(2\, r^2 + 13\, r + 9\right)} + \frac{1}{2}   - \frac{1}{4}(1 + h(1))^2 \\
        \mathbb{V}_0[T_s] & =  \frac{r + 8}{\left(r + 6\right)\, \left(2\, r^2 + 13\, r + 9\right)} + \frac{1}{2}    - \frac{1}{4}(1 + h(0))^2 \\
        \end{split}%
      \end{displaymath}%
Hence, $\lim_{r \to \infty}\mathbb{V}_i[T_s] = 1/4$ for $i \in \{2,1,0 \}$, and 
  \begin{displaymath}%
    \begin{split}%
      \lim_{r \to 0}\mathbb{V}_2[T_s] & = 1 \\
      \lim_{r \to 0}\mathbb{V}_1[T_s] & = 2/9 \\
      \lim_{r \to 0}\mathbb{V}_0[T_s] & = 89/324  \\
  \end{split}%
    \end{displaymath}%
    Denote by $T_l$ the time until coalescence has occurred at both
    loci.  The marginal coalescence times are exponential with rate 1,
    when excluding large offspring numbers.  Solving the recursions
    \begin{displaymath}%
      \begin{split}%
        \mathbb{E}_2[T_l] & = (1 + 2r\mathbb{E}_1[T_l])/(1 + 2r) \\
        \mathbb{E}_1[T_l] & = (1 + \mathbb{E}_2[T_l] +  r\mathbb{E}_0[T_l] + 2)/(r + 3) \\
        \mathbb{E}_0[T_l] & = (1 + 4\mathbb{E}_1[T_l] + 2)/6 \\
        \end{split}%
      \end{displaymath}%
      yields 
      \begin{displaymath}%
        \begin{split}%
         \mathbb{E}[T_l^{(2)}] & = \tfrac{3}{2} - \tfrac{r + 9}{2\left(2\, r^2 + 13\, r + 9\right)} = \tfrac{1}{2} \left(3 - h(2) \right)  \\
          \mathbb{E}[T_l^{(1)}] & = \tfrac{3}{2} - \tfrac{3}{2\, \left(2\, r^2 + 13\, r + 9\right)} = \tfrac{1}{2}\left(3 - h(1) \right) \\
          \mathbb{E}[T_l^{(0)}] & =  \tfrac{3}{2} - \tfrac{1}{2\, r^2 + 13\, r + 9} = \tfrac{1}{2}\left(3 - h(0) \right) \\
          \end{split}%
        \end{displaymath}%
           Applying the recursions (\ref{eq:recvi}) yields the variances $\mathbb{V}_i[T_l]$;
           \begin{displaymath}%
             \begin{split}%
               \mathbb{V}_2[T_l] & = \frac{2\, r^3 + \frac{111\, r^2}{4} + \frac{171\, r}{2} - \frac{81}{4}}{{\left(2\, r^2 + 13\, r + 9\right)}^2} + \frac{5}{4} \\
                 \mathbb{V}_1[T_l] & = \frac{4\, r^2 + 17\, r - \frac{45}{4}}{{\left(2\, r^2 + 13\, r + 9\right)}^2} + \frac{5}{4} \\
                  \mathbb{V}_0[T_l] & = \frac{2\, r^2 + 7\, r - 10}{{\left(2\, r^2 + 13\, r + 9\right)}^2} + \frac{5}{4} 
               \end{split}%
             \end{displaymath}%
             with $\lim_{r \to \infty}  \mathbb{V}_i[T_l] = 5/4$ for $i \in \{0,1,2\}$, and
            \begin{displaymath}%
              \begin{split}%
                \lim_{r \to 0} \mathbb{V}_2[T_l] & = 1, \\
                \lim_{r \to 0} \mathbb{V}_1[T_l] & = 10/9, \\
                \lim_{r \to 0} \mathbb{V}_0[T_l] & = 365/324.\\
                \end{split}%
              \end{displaymath}%

              Now we admit large offspring numbers, take $\eN =
              c/N^2$, and $\rN = r/N$.  Ignoring the labelling of
              the chromosomes, the limit process has three `effective'
              sample states, depending on the number of double-marked
              chromosomes $(\chr)$.  Denote the three sample states by
              $\ou{(\chr)}{(\chr)}$, $(\chr)\ou{(\achr)}{(\bchr)}$, and
              $\ou{(\achr)}{(\achr)}\ou{(\bchr)}{(\bchr)}$, in which
              $\achr$ and $\bchr$ denote single-marked chromosomes.
              The states of the limit process are composed of
              single-marked individuals only, and are therefore the
              same as those of the haploid Wright-Fisher process. By
              $\cachr$ denote a chromosome carrying a common ancestor
              at one locus, and $(\cacb)$ denotes the absorbing
              states.  The transition rates are summarized in the following table:
              \begin{displaymath}%
                \begin{matrix}%
                  & \ou{(\chr)}{(\chr)} & (\chr)\ou{(\achr)}{(\bchr)} & \ou{(\achr)}{(\achr)}\ou{(\bchr)}{(\bchr)} & (\cachr)\ou{(\bchr)}{(\bchr)} & (\cabchr)(\bchr) & (\cacb)  \\\hline
                  \\
        \ou{(\chr)}{(\chr)} & & 2r  &  & & & 1 + c\tfrac{\psi^2}{4}  \\\\
     (\chr)\ou{(\achr)}{(\bchr)} &  1 + c\tfrac{\psi^2 }{4}(1 - \tfrac{\psi}{4})    & &r  & & 2 + c\tfrac{\psi^2 }{2}(1 - \tfrac{\psi}{4})   & c\tfrac{\psi^3}{16} \\\\
    \ou{(\achr)}{(\achr)}\ou{(\bchr)}{(\bchr)} &  c\tfrac{3\psi^4}{32} & 4 + c\left(\psi^2 - \tfrac{\psi^3}{2} - \tfrac{\psi^4}{8}\right)   & & 2 + c\left(\tfrac{\psi^2}{2} - \tfrac{\psi^3}{4} - \tfrac{\psi^4}{16}\right) & c\tfrac{\psi^3}{4}\left(1 - \tfrac{\psi}{4} \right)    & c\tfrac{\psi^4}{16}  \\\\
    (\cachr)\ou{(\bchr)}{(\bchr)} & & & & &  2 + c\tfrac{\psi^2}{2}\left(1 - \tfrac{\psi}{4}\right) & 1 + c\tfrac{\psi^2}{4} \\\\
   (\cabchr)(\bchr) &&&&r & & 1 + c\tfrac{\psi^2}{4} \\ \hline
                  \end{matrix}%
                \end{displaymath}%
                By way of example, the rate of the transition from 1
                to 2 by coalescence of the chromosomes $\achr$ and
                $\bchr$ is $1 + cC_{3;2;1}$, the transition rate from
                0 to 1 is $4\left(1 + cC_{4;2;2}\right)$, and the
                transition rate from 0 to the absorbing state
                ($(\cacb)$ or $(\cachr)(\acb)$) is $c\left(C_{4;4;0} +
                  C_{4;2,2;0}\right)$.

                As before, let $h(i)$ denote the probability the two
                loci coalesce at the same time.  One obtains limit
                results
                    \begin{equation}%
                      \label{eq:limhi}
                      \begin{split}%
                    \lim_{r \to \infty}h(i) & = \tfrac{c\psi^4}{32 + 8c\psi^2 - c\psi^4}, \quad i \in \{0,1,2\} \\
                    \lim_{c \to \infty}h(2) & = 1 \\
                    \lim_{c \to \infty}h(1) & = \tfrac{2}{6  -  \psi }  \\
                    \lim_{c \to \infty}h(0) & =  \tfrac{\frac{56\, \psi^2}{3} - 272\, \psi + 544}{\left(\psi - 6\right)\, \left(3\, \psi^2 + 16\, \psi - 48\right)} - \frac{5}{3}  \\
                    \end{split}%
                    \end{equation}%
                    The first equation in \eqref{eq:limhi} tells us
                    that the loci remain correlated due to multiple
                    mergers even when they are far apart on a
                    chromosome.  When the recombination rate $r$ is
                    quite small, one obtains
                    \begin{equation}%
                      \label{eq:hir0}
                      \begin{split}%
                        \lim_{r \to 0}h(2) & = 1 \\
                        \lim_{r \to 0}h(1) & = \tfrac{2\, \left(c\, \psi^2 + 4\right)}{ - c\, \psi^3 + 6\, c\, \psi^2 + 24} \\
                        \lim_{r \to 0}h(0) & = \tfrac{1}{3}\left( \tfrac{8\, c\, \psi^2 + 32}{ - c\, \psi^3 + 6\, c\, \psi^2 + 24} + \tfrac{ - 80\, c\, \psi^3 + 208\, c\, \psi^2 + 832}{ - 3\, c\, \psi^4 - 16\, c\, \psi^3 + 48\, c\, \psi^2 + 192} - 5\right) \\
                        \end{split}%
                      \end{equation}%
                      Let $\mathbb{E}_i[T_s]$, as before, denote the
                      time until coalescence at either loci, starting
                      from state i.  Admitting large offspring numbers, one obtains
                      \begin{displaymath}
                        \begin{split}%
                      \lim_{r \to \infty}\mathbb{E}_i[T_s] & =  \tfrac{16}{32 + 8c\psi^2  -  c\psi^4} , \quad i \in \{0,1,2 \}, \\
                      \lim_{c \to \infty}\mathbb{E}_i[T_s] & = 0,   \quad i \in \{0,1,2 \}, \\
                      \lim_{r \to 0}\mathbb{E}_2[T_s] &  = \tfrac{4}{c\psi^2 + 4} \\
                      \lim_{r \to 0}\mathbb{E}_1[T_s] & = \tfrac{c\, \left(16\, \psi^2 - 2\, \psi^3\right) + 64}{ - c^2\, \psi^5 + 6\, c^2\, \psi^4 - 4\, c\, \psi^3 + 48\, c\, \psi^2 + 96} \\
                      \lim_{r \to 0}\mathbb{E}_0[T_s] & =  \tfrac{16}{3\, \left(c\, \left(6\, \psi^2 - \psi^3\right) + 24\right)} - \tfrac{4\, \left(\psi - 8\right)}{\left(3\, \psi + 16\right)\, \left(c\, \psi^2 + 4\right)} - \tfrac{32\, \left(39\psi - 32\right)}{3\, \left(c\, \left(3\, \psi^4 + 16\, \psi^3 - 48\, \psi^2\right) - 192\right)\, \left(3\, \psi + 16\right)}  \\
                      \end{split}%
                      \end{displaymath}%
                      Let $\mathbb{E}_i[T_l]$, as before, denote the
                      expected value of the time $T_l$ until
                      coalescence has occurred at both loci, when
                      starting from state $i$.  Admitting large
                      offspring numbers, one obtains the limits 
                      \begin{displaymath}%
                        \begin{split}%
                          \lim_{r \to \infty}\mathbb{E}_i[T_l] &  =  \tfrac{c\, \left(48\, \psi^2 - 8\, \psi^4\right) + 192}{\left(c\, \psi^2 + 4\right)\, \left( - c\, \psi^4 + 8\, c\, \psi^2 + 32\right)}, \quad i \in \{0,1,2\},  \\
                           \lim_{c \to \infty}\mathbb{E}_i[T_l] & = 0,   \quad i \in \{0,1,2 \}, \\
                           \lim_{r \to 0} \mathbb{E}_2[T_l] & = \tfrac{4}{c\psi^2  + 4} \\
                            \lim_{r \to 0} \mathbb{E}_1[T_l] & =  \tfrac{c\, \left(32\, \psi^2 - 6\, \psi^3\right) + 128}{ - c^2\, \psi^5 + 6\, c^2\, \psi^4 - 4\, c\, \psi^3 + 48\, c\, \psi^2 + 96} \\
                           \lim_{r \to 0} \mathbb{E}_0[T_l]  & = \tfrac{\left(28\psi^7 - 56\psi^6 - 800\psi^5 + 1600\psi^4\right)\, c^2 + \left( - 608\psi^4 - 3200\psi^3 + 12800\psi^2\right)\, c + 25600}{a}
                         \end{split}
                         \end{displaymath}
                         in which
                         \begin{displaymath}
                           \begin{split}%
a & = 3\, c^3\psi^9 - 2\, c^3\psi^8 - 144\, c^3\psi^7 + 288\, c^3\psi^6 + 12\, c^2\psi^7 - 80\, c^2\psi^6 - 1152\, c^2\psi^5 \\
 & + 3456\, c^2\psi^4 - 288\, c\psi^4 - 2304\, c\psi^3 + 13824\, c\psi^2 + 18432 
                          \end{split}%
                        \end{displaymath}%
                        Considering the variance $\mathbb{V}_i[T_s]$
                        of the time $T_s$ when starting from state $i
                        \in \{0, 1, 2 \}$, and admitting large
                        offspring numbers, one obtains
                        \begin{displaymath}%
                          \begin{split}%
                           \lim_{r \to \infty} \mathbb{V}_i[T_s] & = \tfrac{256}{{\left(c\, \left(8\psi^2 - \psi^4\right) + 32\right)}^2}, \quad i \in \{0,1,2  \}, \\
                           \lim_{c \to \infty}  \mathbb{V}_2[T_s] & =  0,   \quad i \in \{0,1,2  \}, \\
                           \lim_{r \to 0} \mathbb{V}_2[T_s] & =  \tfrac{16}{\left(c\psi^2 + 4\right)^2} \\
                           \lim_{r \to 0} \mathbb{V}_1[T_s] & =  \tfrac{\left(12\, \psi^6 - 128\, \psi^5 + 384\, \psi^4\right)\, c^2 + \left(3072\, \psi^2 - 512\, \psi^3\right)\, c + 6144}{{\left(c\, \psi^2 + 4\right)}^2\, {\left( - c\, \psi^3 + 6\, c\, \psi^2 + 24\right)}^2}
                            \end{split}%
                          \end{displaymath}%

                          Correlations in coalescence times have been
                          employed to quantify linkage disequilibrium
                          (LD) \citep{M02}, in which LD is quantified
                          as the square of the correlation coefficient
                          of types at two loci \citep{HR68}.  A
                          description of how one can quantify linkage
                          disequilibrium as the square of the
                          correlation coefficient of types at two loci
                          can be found in \cite{HC89}.  Assuming a
                          very small mutation rate, \cite{M02} related
                          $\mathfrak{D}$ to covariances in coalescence
                          times. 
                          Writing $\sub{\textrm{Cov}}{i}(T_1,
                          T_2)$ as the covariance of $T_1$ and $T_2$
                          when starting from state $i \in \{0,1,2\}$,
                          \cite{M02} obtained
            \begin{displaymath}%
              \begin{split}%
            \mathfrak{D} & =  \frac{\sub{\textrm{Cov}}{2}\left[T_1, T_2  \right]  - 2\sub{\textrm{Cov}}{1}\left[ T_1, T_2  \right] + \sub{\textrm{Cov}}{0}\left[ T_1, T_2 \right]}{ \left(\mathbb{E}[T_1]\right)^2  +  \sub{\textrm{Cov}}{0}\left[ T_1, T_2 \right]   } \\
            & = 1 + \frac{\mathbb{E}_2\left[T_1T_2 \right]  - 2\mathbb{E}_1\left[T_1T_2  \right]   }{\mathbb{E}_0\left[T_1T_2   \right] }
            \end{split}
            \end{displaymath}%
            in which $T_1$ and $T_2$ denote the times until
            coalescence at the two loci, respectively, and the
            covariances are conditional on the sample configurations,
            as indicated.  Following e.g.\ \cite{D02} one can obtain
            the covariances under any population model.  Under our population model, 
            $\mathfrak{D} = \mathfrak{D}_1/\mathfrak{D}_2$,  in which
\begin{displaymath}%
  \begin{split}%
    \mathfrak{D}_1 &= 640c\psi^2 - 224c\psi^3 + 32c\psi^4 + 80c^2\psi^4 - 56c^2\psi^5 + 16c^2\psi^6 - c^2\psi^7 \\
    & + r(16c\psi^4 - 32c\psi^3 + 64c\psi^2 + 256) + 1280, \\
    \mathfrak{D}_2 & = 1408c\psi^2 - 352c\psi^3 + 8c\psi^4 + 512r^2 + 176c^2\psi^4 - 88c^2\psi^5 + 10c^2\psi^6 - c^2\psi^7 \\
    & + r(8c\psi^4 - 288c\psi^3 + 832c\psi^2 + 3328) + 2816.  \\
    \end{split}%
  \end{displaymath}%
  One obtains the limit results
  \begin{displaymath}%
    \begin{split}%
    \lim_{r \to \infty}\mathfrak{D} &= 0, \\
    \lim_{c \to \infty}\mathfrak{D} &= \tfrac{ \psi^3 - 16\psi^2  + 56\psi - 80}{ \psi^3 - 10\psi^2 + 88\psi - 176 }.%
    \end{split}%
    \end{displaymath}%
    \subsection{Correlations in coalescence times for random $\psi$ }

    In this section we consider the simple example of the probability
    measure $F$, evoked in relation to a random offspring
    distribution, taking the beta distribution with parameters
    $\vartheta$ and $\gamma$.  The following transition rates for a
    sample of size two at two loci are obtained:
     \begin{displaymath}%
                \begin{matrix}%
                  & \ou{(\chr)}{(\chr)} & (\chr)\ou{(\achr)}{(\bchr)} & \ou{(\achr)}{(\achr)}\ou{(\bchr)}{(\bchr)} & (\cachr)\ou{(\bchr)}{(\bchr)} & (\cabchr)(\bchr) & (\cacb)  \\\hline
                  \\
        \ou{(\chr)}{(\chr)} & & 2r  &  & & & 1  \\\\
     (\chr)\ou{(\achr)}{(\bchr)} &  \tfrac{\gamma + 3\vartheta/4 }{\vartheta + \gamma}    & &r  & & 2 \tfrac{\gamma + 3\vartheta/4 }{\vartheta + \gamma}    &  \tfrac{\vartheta}{4(\vartheta + \gamma)}  \\\\
    \ou{(\achr)}{(\achr)}\ou{(\bchr)}{(\bchr)} &  \tfrac{3}{8}\tfrac{(1 + \vartheta)\vartheta }{(1 + \vartheta + \gamma)(\vartheta + \gamma) }   & \tfrac{4(1 + \gamma)\gamma + 3\vartheta\gamma + \tfrac{3}{2}(1 + \vartheta)\vartheta }{(1 + \vartheta + \gamma)(\vartheta + \gamma)}   & & \tfrac{2(1 + \gamma)\gamma + \tfrac{3}{2}\vartheta\gamma + \tfrac{3}{4}(1 + \vartheta)\vartheta }{(1 + \vartheta + \gamma)(\vartheta + \gamma)}  &  \tfrac{\vartheta\gamma  +  \tfrac{3}{4}(1 +   \vartheta )\vartheta }{(1 + \vartheta + \gamma)(\vartheta + \gamma)}   & \tfrac{(\vartheta + 1)\vartheta }{4(\vartheta + \gamma + 1)(\vartheta + \gamma) }  \\\\
    (\cachr)\ou{(\bchr)}{(\bchr)} & & & & &  2\tfrac{\gamma + 3\vartheta/4 }{\vartheta + \gamma }  & 1  \\\\
   (\cabchr)(\bchr) &&&&r & & 1  \\ \hline
                  \end{matrix}%
                \end{displaymath}%
                As before, the transition rates given above can be
                employed to derive correlations in coalescence times.
                Here we only consider the probability $h(i)$.  One
                obtains $\lim_{\vartheta \to 0}h(i) = \lim_{\gamma \to
                  \infty}h(i)$ and the limit results are those
                obtained from the usual ARG \eqref{eq:hiKlimr}.

\subsection{Variance of pairwise differences}
The variance of pairwise differences between DNA sequences has been
employed to estimate recombination rates in low offspring number populations \citep{W97}.  Let the random variable
$\sub{K}{ij}$  denote the number of differences between sequences $i$
and $j$, with $\sub{K}{ii} = 0$.  The average number $\pi$ of pairwise
differences  for $n$ sequences is
\begin{displaymath}%
  \pi = \frac{2}{n(n - 1)}\sum_{i < j}\sub{K}{ij}
  \end{displaymath}%
  Under the infinitely many sites mutation model, $\mathbb{E}[\pi] =
  \theta\mathbb{E}[T]$, in which $T$ is the time until coalescence of
  two sequences.   Under our model,  $\mathbb{E}[T] = 1/(1 + c\psi^2/4)$.  
Define the variance
  $\sub{S}{\pi}^2$ of pairwise differences as
$$
\sub{S}{\pi}^2 = \frac{2}{n(n - 1)}\sum_{i < j} \left(\sub{K}{ij} - \pi  \right)^2
$$
To obtain an estimate of the recombination rate, one needs to compute
the expected value  $\mathbb{E}\left[\sub{S}{\pi}^2\right]$, 
$$
\mathbb{E}\left[\sub{S}{\pi}^2\right] = \frac{2}{n(n - 1)}\sum_{i < j}\mathbb{E}\left[ \left(\sub{K}{ij} - \pi\right)^2 \right] = \mathbb{E}\left[\left(\sub{K}{12} - \pi
  \right)^2\right]. $$
Thus, it suffices to consider $\mathbb{E}\left[\left(\sub{K}{12} - \pi
  \right)^2\right]$. Expanding, one obtains 
\newcommand{\ip}{\ensuremath{\hat{\imath}}}%
\begin{displaymath}%
  \begin{split}%
  \mathbb{E}\left[\left(\sub{K}{12} - \pi \right)^2\right] & = \mathbb{E}\left[ \left( \tfrac{2}{n(n - 1)}\sum_{i < j} (\sub{K}{12} - \sub{K}{ij} )  \right)^2  \right] \\\\
  & = \tfrac{4}{n^2(n - 1)^2}\sum_{i < j}\sum_{\ip < \jp}\mathbb{E}\left[ (\sub{K}{12} -  \sub{K}{ij})(\sub{K}{12} -  \sub{K}{\ip \jp})    \right].
  \end{split}%
\end{displaymath}%
Define the event $\sub{A}{ij}^{(\ell)}$ by
$$
\sub{A}{ij}^{(\ell)} := \left\{ \textrm{sequences $i$ and $j$ differ at locus $\ell$}    \right\}.
$$
Assuming each sequence consists of $L$ loci,  and  $1_{\sub{A}{ij}^{(\cdot)}}$ are  indicator functions, 
$$
\sub{K}{12} - \sub{K}{ij} = \sum_{\ell = 1}^L \left(1_{\sub{A}{12}^{(\ell)} } -  1_{\sub{A}{ij}^{(\ell)} }\right)
$$
yielding, in case $i = i^\prime = 1$, and $j = \jp = 3$,
\newcommand{\lh}{\ensuremath{\hat{\ell}}}%
\begin{displaymath}%
  \begin{split}%
\mathbb{E}\left[ (\sub{K}{12} - \sub{K}{13} )^2 \right] & = \sum_{\ell = 1}^L \sum_{\lh = 1}^L\mathbb{E} \left[\left( 1_{\sub{A}{12}^{(\ell)} } -  1_{\sub{A}{13}^{(\ell)} }\right)\left( 1_{\sub{A}{12}^{(\lh)} } -  1_{\sub{A}{13}^{(\lh)} }\right) \right]  \\\\
& = 2\sum_{\ell = 1}^L\sum_{\lh = 1}^L \mathbb{P}\left( \sub{A}{12}^{(\ell)}\cap \sub{A}{12}^{(\lh)} \right) -  \mathbb{P}\left( \sub{A}{12}^{(\ell)}\cap \sub{A}{13}^{(\lh)} \right).
\end{split}%
\end{displaymath}%
In general, 
\begin{equation}%
  \label{eq:EKiKjKiKjgeneral}%
  \begin{split}%
 & \mathbb{E}\left[\left(\sub{K}{12} - \sub{K}{ij}\right)\left(\sub{K}{12} - \sub{K}{\ip \jp}\right) \right] \\\\ 
= & \sum_{\ell = 1}^L\sum_{\lh = 1 }^L \mathbb{E}\left[ \left(  1_{\sub{A}{12}^{(\ell)} } -  1_{\sub{A}{ij}^{(\ell)} }  \right)  \left(  1_{\sub{A}{12}^{(\lh)} } -  1_{\sub{A}{\ip \jp}^{(\lh)} }  \right)   \right]\\\\
 = &  \sum_{\ell = 1}^L\sum_{\lh = 1 }^L \mathbb{P}\left( \sub{A}{12}^{(\ell)}\cap \sub{A}{12}^{(\lh)}  \right) -  \mathbb{P}\left( \sub{A}{12}^{(\ell)}\cap \sub{A}{\ip \jp}^{(\lh)}  \right)  - \mathbb{P}\left( \sub{A}{12}^{(\lh)}\cap \sub{A}{ij}^{(\ell)}  \right) + \mathbb{P}\left( \sub{A}{ij}^{(\ell)}\cap \sub{A}{\ip \jp}^{(\lh)}  \right).
  \end{split}%
  \end{equation}%
  Now consider the probability $\mathbb{P}( \sub{A}{12}^{(\ell)}\cap
  \sub{A}{12}^{(\lh)} )$ of the event that sequences 1 and 2 differ at
  both loci $\ell$ and $\lh$.  Admitting mutation introduces two new
  states, namely the states $\ou{(\achr)}{(\achr)}$ and
  $\ou{(\bchr)}{(\bchr)}$.  Define
  $$
  g(\mathfrak{S}) := \mathbb{P}\left(\textrm{both loci separated by mutation, starting from state $\mathfrak{S}$} \right)
  $$

    Thus, $\mathbb{P}\left(
      \sub{A}{12}^{(\ell)}\cap\sub{A}{12}^{(\lh)} \right) =
    g(2)$,  $\mathbb{P}\left(
      \sub{A}{12}^{(\ell)}\cap\sub{A}{13}^{(\lh)} \right) =
    g(1)$, and     $\mathbb{P}\left(\sub{A}{12}^{(\ell)}\cap\sub{A}{34}^{(\lh)} \right) =
    g(0)$, for $\ell \ne \lh$.   
Now,
\begin{displaymath}%
  \begin{split}%
  g(2) & = \frac{\theta_1g(-1) + \theta_2g(-2) + 2rg(1) }{\theta_1 + \theta_2 + 1 + c\tfrac{\psi^2}{4} + 2r } \\
  g(-1) & = \frac{\theta_2}{\theta_2  +  1 + c\tfrac{\psi^2}{4}} \\
  g(-2) & =  \frac{\theta_1}{\theta_1  +  1 + c\tfrac{\psi^2}{4}} \\
  g(1) & = \frac{\theta_1g(-1) + \theta_2g(-2) + rg(0) + \left(1 + c\psi^2/4 \right)g(2) }{\theta_1 + \theta_2 + r + 3 + 3c\tfrac{\psi^2}{4}(1 - \tfrac{\psi}{4}) + c\tfrac{\psi^3}{16}  } \\
  g(0) & =  \frac{\theta_1g(-1) + \theta_2g(-2) +  c\tfrac{3\psi^4}{32}g(2)  +  \left(c\left(\psi^2  -  \tfrac{\psi^3 }{2} - \tfrac{\psi^4}{8}\right) + 4\right)g(1)}{c\tfrac{3\psi^4}{32} + c\left(\psi^2 - \tfrac{\psi^3}{2} - \tfrac{\psi^4}{8}\right) +  c\left( \tfrac{\psi^2}{2} - \tfrac{\psi^3}{4} - \tfrac{\psi^4}{16} \right) + 6 + c\tfrac{\psi^3}{4}\left(1 - \tfrac{\psi}{4} \right) + c\tfrac{\psi^4}{16} + \theta_1 + \theta_2 }
  \end{split}%
  \end{displaymath}%

  In view of expression  \eqref{eq:EKiKjKiKjgeneral},  one obtains
  \begin{equation}%
    \label{eq:PAAsamelocus}%
    \begin{split}%
      \PP{A_{12}^{(\ell)}\cap A_{12}^{(\ell)}} & =  \PP{A_{12}^{(\ell)}} =    \frac{\theta_\ell}{\theta_\ell + 1 + c\psi^2/4 }, \\
      \PP{A_{12}^{(\ell)}\cap A_{13}^{(\ell)} } & = \frac{\theta_\ell}{\tfrac{3\theta_\ell}{2} + \lambda_3} +  \tfrac{ \lambda_2}{\tfrac{3\theta_\ell}{2} + \lambda_3}\frac{\theta_\ell}{\theta_\ell + \lambda_2 }, \\ 
      \PP{A_{12}^{(\ell)}\cap A_{34}^{(\ell)}} & = \frac{2\theta_\ell}{2\theta_\ell + \lambda_4}\frac{\theta_\ell}{2\theta_\ell + \lambda_4 } +  \frac{\lambda_{4;2}}{2\theta_\ell + \lambda_4}\left( \frac{\theta_\ell/2 }{\tfrac{3\theta_\ell}{2} + \lambda_3 } +  \left(\frac{\theta_\ell/2}{\tfrac{3\theta_\ell}{2} + \lambda_3}\right)^2 \right).
    \end{split}%
    \end{equation}%
    The event $A_{12}^{(\ell)}\cap A_{34}^{(\ell)}$
    \eqref{eq:PAAsamelocus} occurs if the first two events in the
    history of the four sequences are mutations on appropriate
    ancestral lineages, or if lineages labelled 2 and 3 coalesce,
    followed by appropriately-placed mutations.

\clearpage
\pagebreak
\newpage

\begin{table}%
  \caption{Estimates $\overline{R}_i$ of the expected values $E[R_i]$ of the ratios $R_i := L_i/L$ for $1 \le i \le 4$ at one marginal locus, along with estimates $\widehat{R}_i$ of the standard deviations of $R_i$.  Estimates are obtained from $10^5$ simulated gene genealogies. }
  \label{tab:ERi}
\begin{tabular}{llllllllllll}%
  \hline \\
   $\psi$ &$c$ & $n$ & $\overline{R}_1$ & $\overline{R}_2$  & $\overline{R}_3$  & $\overline{R}_4$  & $\widehat{R}_1$  & $\widehat{R}_2$ &  $\widehat{R}_3$   & $\widehat{R}_4$ \\
   \hline
   -- & 0 & 6 & $0.466$    & $ 0.219$   & $ 0.138$  & $ 0.100$ & $0.183$    &   $0.167$ &  $0.198$  & $0.124$ \\
   & & 10 & $0.378$ &  $0.180$   & $0.117$   & $ 0.085$ & $0.156$  & $0.132$  &  $0.120$    &  $0.110$ \\
   && 20 & $0.300$    &  $0.146$   & $ 0.096$    & $ 0.070$ & $0.119$  & $ 0.097$   & $ 0.088$  & $ 0.081$ \\
   && 50 & $0.235$  & $ 0.116$   & $ 0.077$    & $ 0.057$ & $0.080$    & $ 0.063$   &   $0.058$    & $0.055$ \\
   \hline
 $0.005$ & 1& 6 & $0.466$ &  $0.219$  &  $0.138$  & $ 0.100$ & $0.183$  &  $0.167$  &  $0.198$  &  $0.124$ \\
 & & 10 & $0.377$  &  $0.181$  &  $0.117$  &  $0.085$ & $0.156$  &  $0.133$   &  $0.120$  &  $0.111$ \\
 && 20 & $0.299$   &  $0.146$   &  $0.095$  &  $0.071$ & $0.118$    &  $0.097$   &  $0.088$    &  $0.082$ \\
 && 50 & $0.234$    &  $0.116$  &  $0.076$  &  $0.057$ & $0.080$    &  $0.064$   &  $0.057$   &  $0.054$ \\
 & $1000$& 6 & $0.467$    &  $0.219$   &  $0.137$    &  $0.100$ & $0.182$  &  $0.167$   &  $0.198$  &  $0.124$ \\
 && 10 & $0.377$  &  $0.181$  &  $0.117$  &  $0.085$ & $0.156$  &  $0.133$   &  $0.120$  &  $0.110$ \\
 && 20 & $0.299$  &  $0.146$ &  $0.095$    & $0.071$ & $0.119$    & $0.097$   & $0.088$    & $0.082$ \\
 && 50 & $0.235$  &  $0.116$ &  $0.077$    & $0.057$ & $0.080$  &  $0.064$   &  $0.058$    & $0.054$ \\
 $0.5$ & 1 & 6 &  $0.468$    &  $0.217$  & $0.138$  &  $0.099$ & $0.184$  & $0.166$ & $0.199$ & $0.124$ \\
 &&10 & $0.381$  & $0.179$   &  $0.115$  & $0.085$ & $0.157$    &  $0.132$  & $0.120$   &  $0.110$ \\
 &&20& $0.304$  &  $0.145$ &  $0.095$  &  $0.070$ & $0.120$   &  $0.097$   &  $0.088$   &  $0.081$ \\
 &&50& $0.242$    &  $0.117$   &  $0.077$  &  $0.056$ & $0.081$   &  $0.064$   &  $0.058$    &  $0.054$ \\
 & 1000& 6 & $0.541$  &  $0.173$ &  $0.116$  &  $0.089$ & $0.184$  & $ 0.152$   &  $0.177$  &  $0.116$ \\
 && 10 & $0.566$  &  $0.117$   &  $0.078$    &  $0.058$ & $0.159$  &  $0.101$   &  $0.090$  &  $0.082$ \\
 && 20 & $0.743$  &  $0.101$   &  $0.035$    &  $0.022$ & $0.084$  &  $0.053$   &  $0.033$  &  $0.027$ \\
 && 50 & $0.576$  &  $0.195$   &  $0.089$    &  $0.046$ & $0.058$  &  $0.051$   &  $0.037$  &  $0.026$ \\
  \hline
  \end{tabular}%
\end{table}%

\clearpage\pagebreak\newpage

\begin{sidewaystable}
\caption{Estimates of the correlation \corr{X^{(1)}}{Y^{(2)}} between $X^{(1)}$ and $Y^{(2)}$, where $X^{(1)}$ represents a statistic for locus 1, and $Y^{(2)}$ for locus 2, as follows: the time $T$ until most recent common ancestor at a locus; $L$ the total length of the gene genealogy at a locus, and $R_i := L_i/L$, in which $L_i$ denotes the total length of branches ancestral to $i$ sequences.   Estimates are based on $10^5$ simulated ancestral recombination graphs each for a sample of size $50$.}
\label{tab:corXYa}
\begin{tabular}{llllllllllll}%
  \hline \\
   $c$ & $\psi$&  $r$ &   \corr{T^{(1)}}{T^{(2)}} &  \corr{L^{(1)}}{L^{(2)}} & \corr{L^{(1)}_1}{L^{(2)}_1} & \corr{L^{(1)}_2}{L^{(2)}_2} & \corr{L^{(1)}_3}{L^{(2)}_3} & \corr{L^{(1)}_4}{L^{(2)}_4}  \\ \hline
\hline
0 & $-$ & 1 & $0.311$  & $0.418$  & $0.586$  & $0.501$  &  $0.434$  &  $0.378$ \\
&   &10 & $0.016$  & $0.058$   &  $0.169$   & $0.089$   &  $0.047$  & $0.036$  \\
\hline
1 & $0.005$& 1  & $0.306$  & $0.415$   &  $0.588$   &  $0.508$   &  $0.431$  & $0.380$ \\
& & 10 & $0.015$  &  $0.055$   &  $0.171$ &  $0.090$  &  $0.049$  &  $0.034$ \\
1000&$0.005$& 1& $0.308$  & $0.419$   &  $0.585$   &  $0.509$   &  $0.438$  &  $0.376$ \\
&& 10& $0.013$  & $0.051$   &  $0.168$   &  $0.093$ &  $0.052$  & $0.030$ \\
1 & $0.5$&1& $0.328$  & $0.447$   &  $0.601$  & $0.516$   &  $0.449$  &  $0.389$ \\
1 && 10 & $0.024$  & $0.085$ &  $0.193$   & $0.107$   &  $0.064$  & $0.036$ \\
1000&&1 & $0.982$  & $0.995$ &  $0.976$   & $0.950$   &  $0.918$  & $0.879$ \\
&&10& $0.924$  & $0.947$   &  $0.763$   &  $0.623$   &  $0.503$  & $0.396$ \\
\hline
$c$ & $\psi$ & $r$ & \corr{L^{(1)}_1}{L^{(2)}_2} & \corr{L^{(1)}_1}{L^{(2)}_3} & \corr{L^{(1)}_1}{L^{(2)}_4} & \corr{L^{(1)}_2}{L^{(2)}_3} & \corr{L^{(1)}_2}{L^{(2)}_4} & \corr{L^{(1)}_3}{L^{(2)}_4}\\
\hline
0 & $-$&1& $-0.031$  &  $-0.031$  &  $-0.021$  & $-0.005$  &  $-0.018$  & $0.009$ \\
&& 10& $0.005$  & $-0.006$   &  $-0.001$   & $0.012$   &  $0.005$  & $0.013$  \\
\hline
1 &$0.005$&1& $-0.035$ & $-0.025$ & $-0.021$   &   $-0.001$   &  $-0.019$  & $0.009$ \\
&&10 & $0.000$  & $-0.002$   & $0.008$ &  $0.009$  &  $0.005$  &   $0.014$ \\
1000 &$0.005$&1& $-0.036$  & $-0.029$   & $-0.021$   & $-0.006$   &  $-0.018$  & $0.010$ \\
&&10& $-0.002$  &  $-0.003$   &  $0.003$   & $0.014$   &  $0.004$  &  $0.005$ \\
1 &$0.5$&1& $-0.022$  & $-0.014$   &  $-0.007$   & $0.004$   &  $-0.004$  & $0.023$ \\
&&10& $0.009$  & $0.006$  & $0.010$  & $0.022$   & $0.014$  &  $0.025$  \\
1000&&1& $0.326$  &  $0.314$   & $0.305$   & $0.238$   & $0.218$  & $0.176$ \\
&&10& $0.311$  & $0.284$   &  $0.266$   &  $0.289$   & $0.239$  & $0.262$ \\
\hline
\end{tabular}\end{sidewaystable}%

\begin{sidewaystable}
\caption{Estimates of the correlation \corr{X^{(1)}}{Y^{(2)}} between $X^{(1)}$ and $Y^{(2)}$, where $X^{(1)}$ represents a statistic for locus 1, and $Y^{(2)}$ for locus 2, as follows: 
 the time $T$ until most recent common ancestor at a locus; $L$ the total length of the gene genealogy at a locus, and $R_i := L_i/L$, in which $L_i$ denotes the total length of branches ancestral to $i$ sequences.  Estimates are based on $10^5$ simulated ancestral recombination graphs each for a sample of size $50$.}
\label{tab:corXYb}%
\begin{tabular}{llllllllllll}%
  \hline \\
 $c$ & $\psi$ & $r$ & \corr{R^{(1)}_1}{R^{(2)}_1} & \corr{R^{(1)}_2}{R^{(2)}_2} & \corr{R^{(1)}_3}{R^{(2)}_3} & \corr{R^{(1)}_4}{R^{(2)}_4} \\
  \hline
0&$-$&1& $0.570$  &  $0.548$  &  $0.486$   &  $0.431$  \\
&&10& $0.116$  & $0.089$  & $0.052$   & $0.042$  \\  
\hline
1&$0.005$&1& $0.566$  & $0.552$  &  $0.487$   &  $0.435$ \\
&&10& $0.115$  &  $0.091$  &  $0.054$   &  $0.035$ \\
1000&$0.005$&1& $0.570$  &  $0.551$  &  $0.491$   &  $0.434$ \\
&&10& $0.115$  & $0.095$  &  $0.059$   & $0.031$ \\
1&$0.5$&1& $0.583$  &  $0.557$  &  $0.504$   &  $0.447$ \\
&&10& $0.135$  & $0.102$  &  $0.063$   &  $0.038$  \\
1000&$0.5$&1& $0.955$  &  $0.927$  &  $0.900$   &  $0.866$ \\
&&10& $0.679$  & $0.469$  &  $0.384$   &  $0.304$  \\
\hline
 $c$ & $\psi$ & $r$ & \corr{R^{(1)}_1}{R^{(2)}_2} & \corr{R^{(1)}_1}{R^{(2)}_3} & \corr{R^{(1)}_1}{R^{(2)}_4} & \corr{R^{(1)}_2}{R^{(2)}_3} & \corr{R^{(1)}_2}{R^{(2)}_4} & \corr{R^{(1)}_3}{R^{(2)}_4} \\
  \hline
0& $-$ & 1 & $-0.023$ & $-0.040$ &  $-0.042$ & $-0.026$   &  $-0.042$   & $-0.014$ \\  
 & & 10 &  $-0.022$   & $-0.023$   &  $-0.020$   &  $0.003$   &  $-0.005$   & $0.005$ \\
  \hline
1&$0.005$&1& $-0.024$   & $-0.038$   & $-0.042$   & $-0.023$   & $-0.046$   & $-0.014$ \\
&&10& $-0.027$   & $-0.018$   &  $-0.015$   &  $0.001$   & $-0.007$   &  $0.011$ \\
1000&$0.005$&1& $-0.028$   & $-0.038$   & $-0.038$   &  $-0.031$   &  $-0.043$   &  $-0.012$ \\
&&10& $-0.030$   &  $-0.024$   &  $-0.016$   &  $0.003$   &  $-0.008$   & $-0.001$ \\
1&$0.5$&1 & $-0.023$ & $-0.035$   &  $-0.035$   &  $-0.028$   &  $-0.034$   & $-0.007$ \\
1&$0.5$&10& $-0.029$   &  $-0.023$   & $-0.015$   &  $0.004$   &  $0.000$   &  $0.016$ \\
1000&&1& $-0.622$   & $-0.348$   &  $-0.112$   &  $-0.100$   &  $-0.038$   &  $-0.016$ \\
&&10& $-0.330$   &  $-0.255$   &  $-0.135$   &  $0.009$   &  $0.004$   &  $0.096$ \\
\hline
  \end{tabular}%
\end{sidewaystable}%

\clearpage
\pagebreak
\newpage

\begin{figure}
  \caption{The probabilities $h(2)$, $h(1)$, and $h(0)$ as functions of 
    $\psi$ (lines) for different values of $r$ and $c$. 
    Values of $h(\cdot)$ obtained from the usual Moran model are shown 
    for reference (symbols).}
    \label{figure1}  
  \begin{picture}(50,50)
    \put(-40,-450){\includegraphics[height=6.5in,width=6.5in,angle=0,scale=1.3]{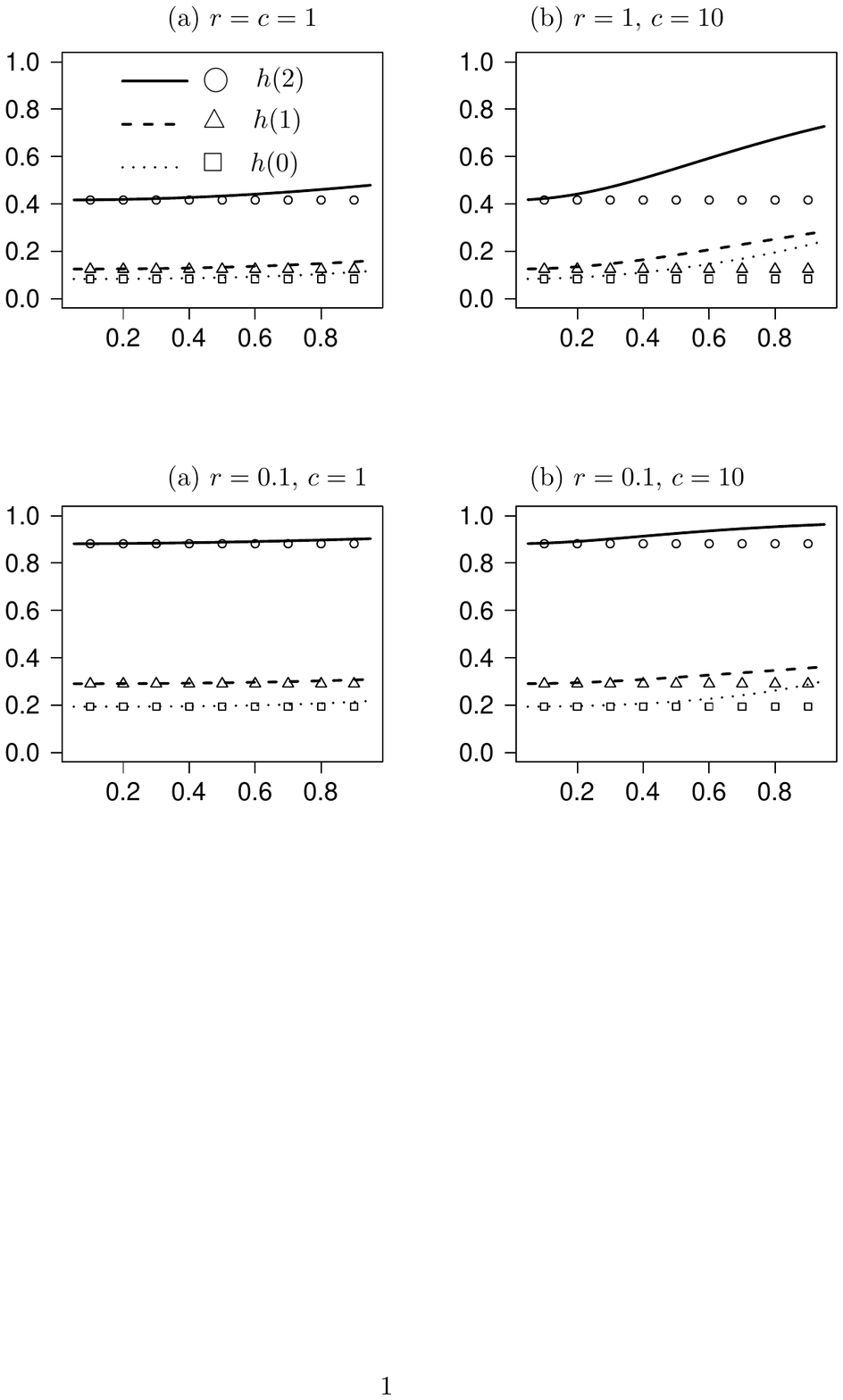}}%
  \end{picture}%
\end{figure}

\clearpage
\pagebreak
\newpage

\begin{figure}[T!]
  \caption{The expected time $\mathbb{E}[T_s^{(i)}]$ as a function of
    $\psi$ for different values of $c$ and $r$.  Values of
    $\mathbb{E}[T_s^{(i)}]$ associated with the case $c = 0$ are shown for
    reference (symbols).}
  \label{figure2}%
  \begin{picture}(50,50)%
     \put(-40,-450){\includegraphics[height=6.5in,width=6.5in,angle=0,scale=1.3]{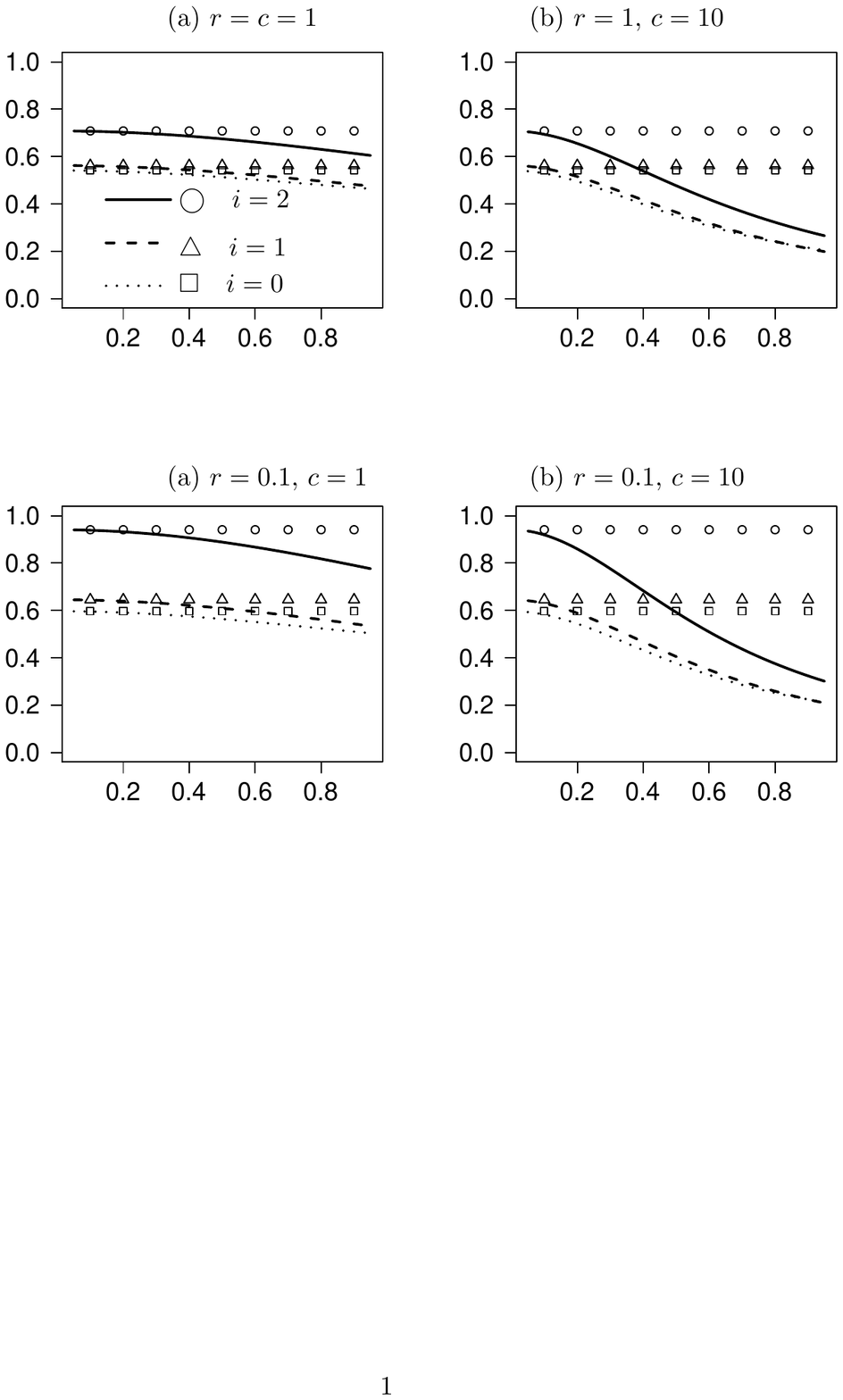}}%
  \end{picture}%
  
\end{figure}

\clearpage
\pagebreak
\newpage

\begin{figure}[T!]
   \caption{The expected time $\mathbb{E}[T_l^{(i)}]$ as a function of
    $\psi$ for different values of $c$ and $r$.  For explanation of
    symbols, see Figure~\ref{figure2}.}
  \label{figure3}
  \begin{picture}(50,50)%
     \put(-40,-450){\includegraphics[height=6.5in,width=6.5in,angle=0,scale=1.3]{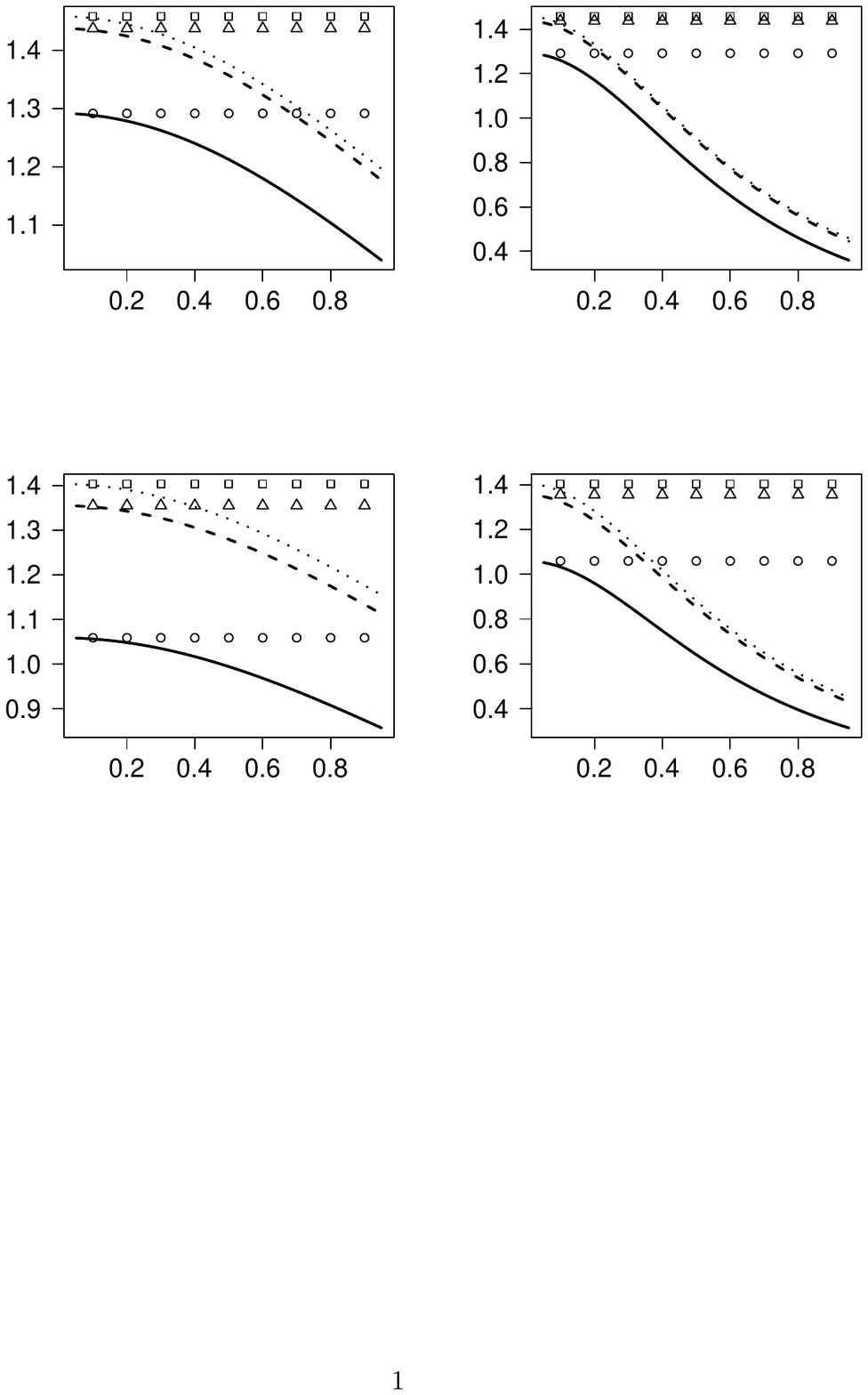}}%
    \put(120,0){(a) $r = c = 1$}\put(320,0){(b) $r = 1$, $c = 10$}%
    \put(120,-140){(c) $r = 0.1$, $c = 1$}\put(320,-140){(d) $r = 0.1$, $c = 10$}%
    \put(160,-258){$\psi$}\put(360,-258){$\psi$}%
  \end{picture}%
 
\end{figure}

\clearpage
\pagebreak
\newpage

\begin{figure}[T!]
   \caption{Correlation of the time to coalescence at two loci as a
    function of $\psi$, for different values of $c$ and $r$.  For
    explanation of symbols, see Figure~\ref{figure2}.}
  \label{figure4}  
  \begin{picture}(50,50)%
     \put(-40,-450){\includegraphics[height=6.5in,width=6.5in,angle=0,scale=1.3]{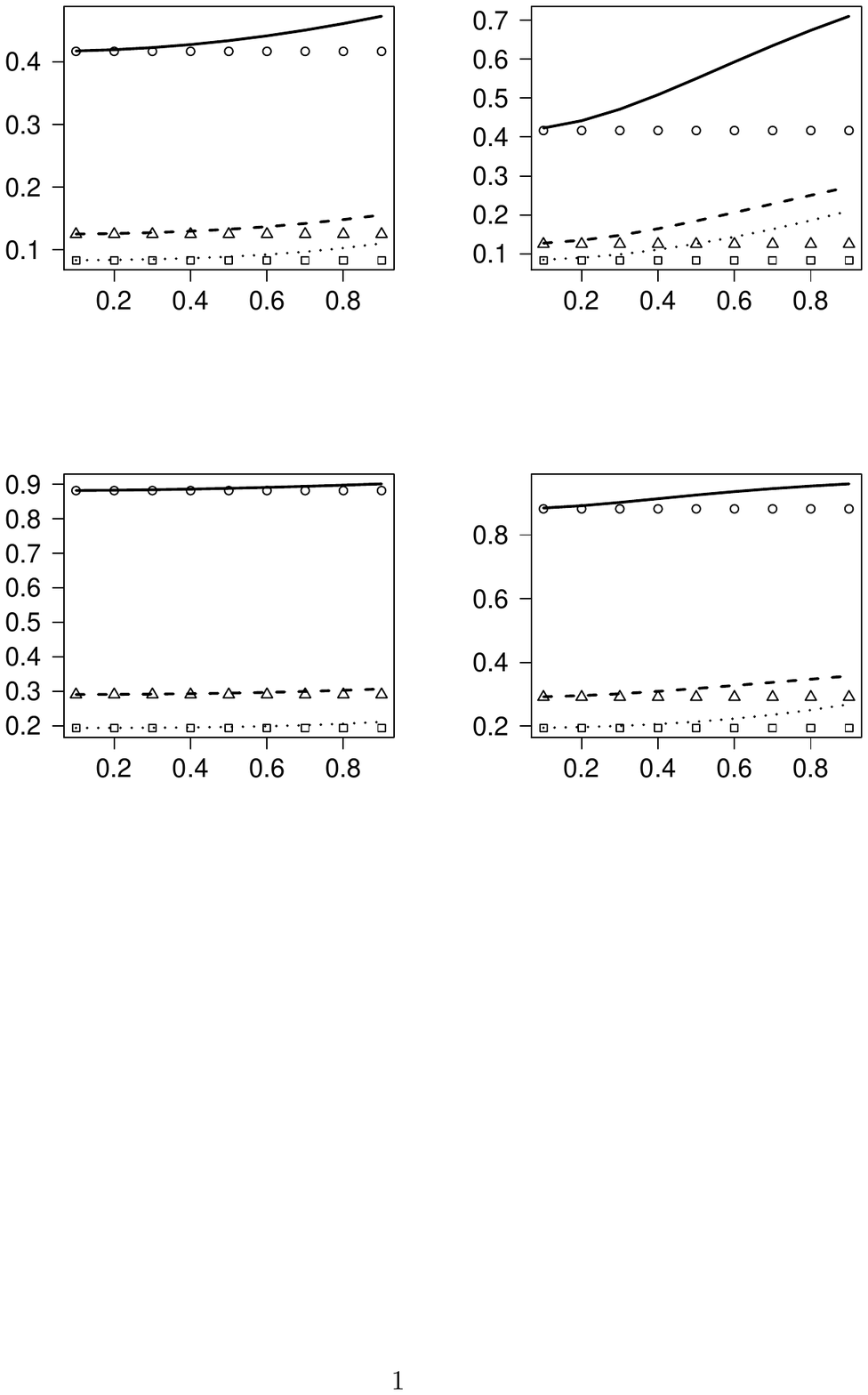}}%
      \put(120,0){(a) $r = c = 1$}\put(320,0){(b) $r = 1$, $c = 10$}%
    \put(120,-140){(c) $r = 0.1$, $c = 1$}\put(320,-140){(d) $r = 0.1$, $c = 10$}%
    \put(160,-258){$\psi$}\put(360,-258){$\psi$}%
  \end{picture}%
 
\end{figure}

\clearpage
\pagebreak
\newpage

\begin{figure}
  \caption{The estimate $\mathfrak{D}$ of the expected value
    $\mathbb{E}[r^2]$ as a function of $\psi$ for different values of $c$
    (see panels), and $r$.  The solid lines represent the value of
    $\mathfrak{D}$ associated with the usual Wright-Fisher model.}
  \label{figure7}
  \begin{picture}(50,50)%
     \put(-40,-440){\includegraphics[height=6.5in,width=6.5in,angle=0,scale=1.3]{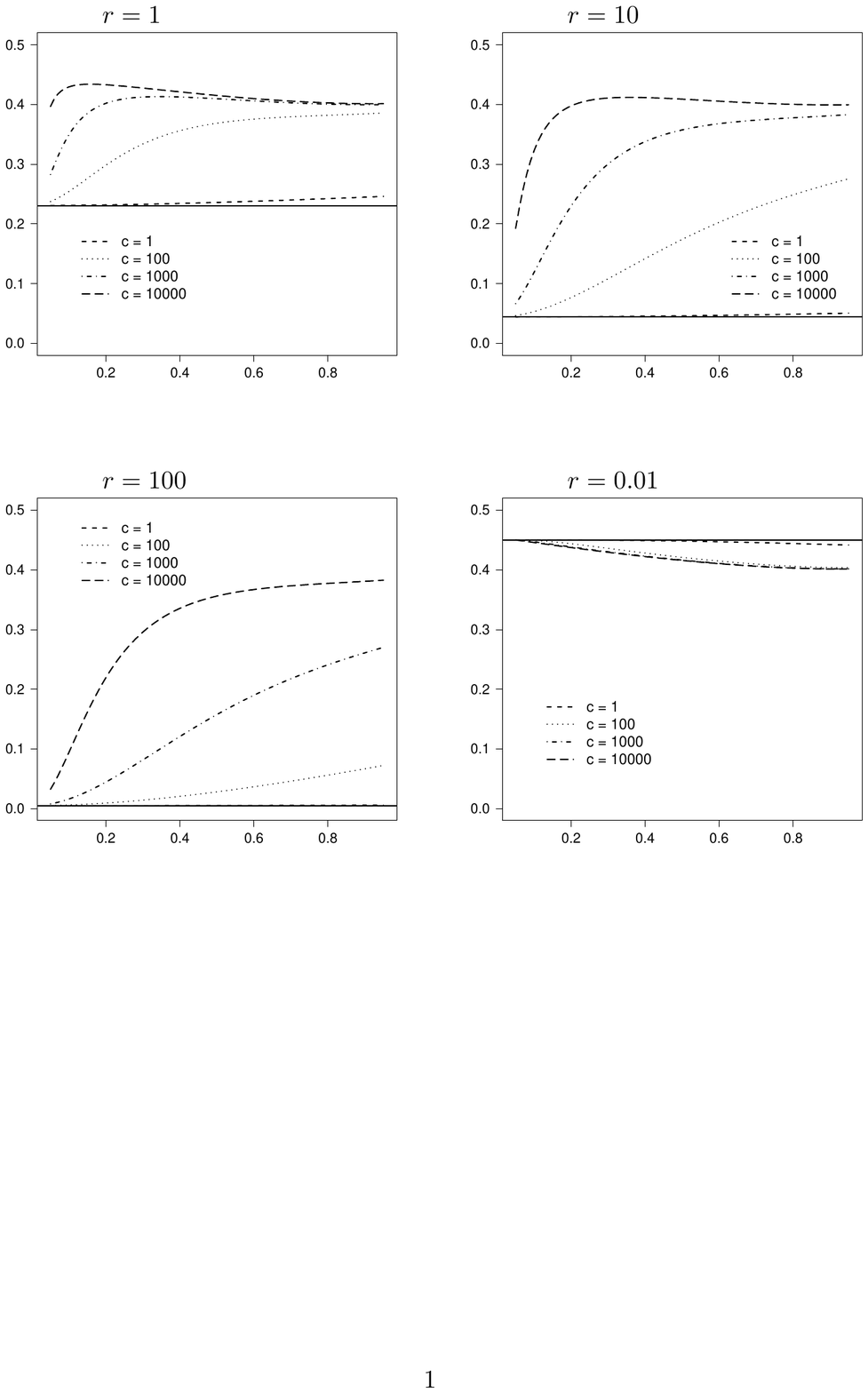}}%
     \put(150,-260){$\psi$}\put(350,-260){$\psi$}%
  \end{picture}%
\end{figure}
\clearpage\pagebreak\newpage

\begin{figure}%
  \caption{The prediction $\mathfrak{D}$ of linkage disequilibrium
    obtained from the ARG associated with the Beta$(\vartheta,\gamma)$
    dist.  The different lines represent different values of $\gamma$
    (upper panels) or $\vartheta$ (lower panels).   The broken horizontal line represents the prediction obtained from the usual ARG.   }%
  \label{fig:LDgammatheta}%
   \begin{picture}(50,50)%
     \put(20,-360){\includegraphics[height=4.5in,width=4.5in,angle=0,scale=1.3]{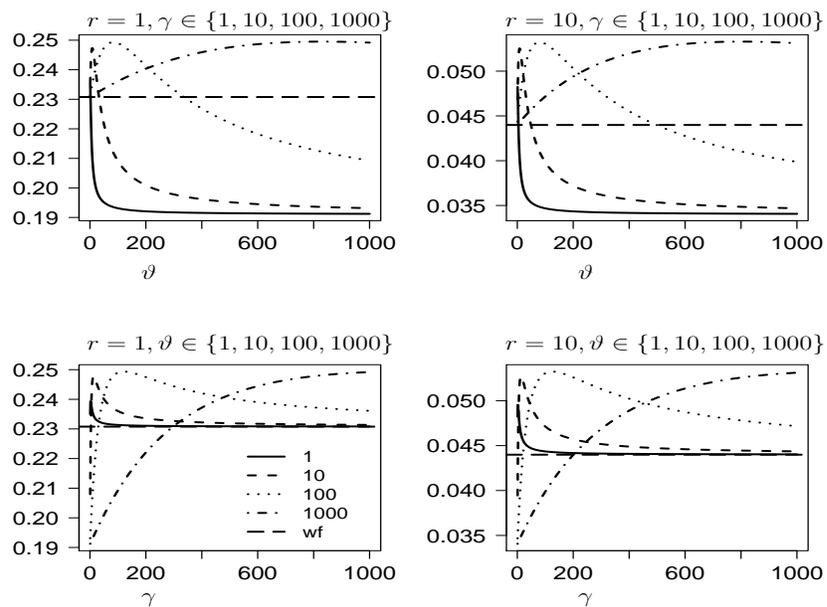}}%
  \end{picture}%
  \end{figure}%

\clearpage\pagebreak\newpage

\begin{figure}%
  \caption{The expected variance of pairwise differences for sample
    size 50 as a function of the recombination rate $r$ for different
    values of the parameters $c$, $\psi$, and $\theta$ as shown.  }%
  \label{fig:pairvarrho0}%
  \begin{picture}(50,50)
   \put(20,-430){\includegraphics[height=7in,width=7in]{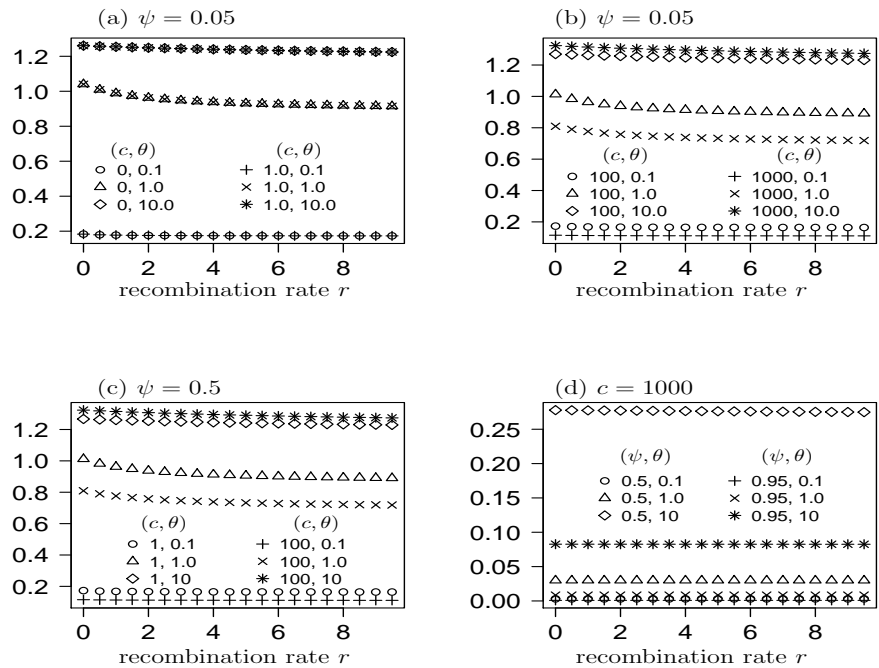}}%
   \end{picture}%
  \end{figure}%

\clearpage\pagebreak\newpage

\begin{figure}%
  \caption{The expected variance of pairwise differences as a function of sample
    size for different values of the parameters $c$, $\psi$,  $r$, and $\theta$ as
    shown.  }%
  \label{fig:pairwnn0}%
  \begin{picture}(50,50)
   \put(20,-430){\includegraphics[height=7in,width=7in]{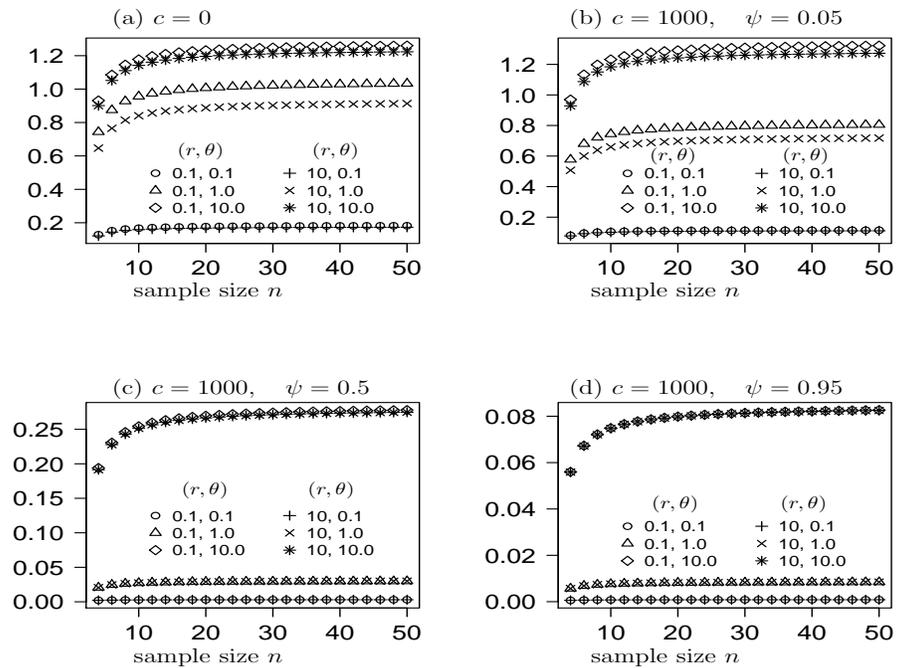}}%
   \end{picture}%
  \end{figure}%

\end{document}